\newif\ifTR
\TRtrue

\newif\ifPosting
\Postingtrue

\newif\ifOldAppendix
\OldAppendixfalse

\ifTR
\documentclass[9pt,reprint,nocopyrightspace]{sigplanconf}
\else
\documentclass[9pt,reprint,nocopyrightspace]{sigplanconf}
\fi

\usepackage{mathptmx}
\usepackage{microtype}

\ifTR
\fi
\ifPosting
\fi

\ifTR
  \newcommand{\refappendix}[1]{in \appref{#1}}
\else
  \usepackage{xr}
  \newcommand{\refappendix}[1]{in Appendix~\ref{TR-app:#1} of the supplementary material}
\fi

\newcommand{\IfTR}[2]{\ifTR{#1}\else{#2}\fi}

\renewcommand{\topfraction}{0.95}

\renewcommand{\floatpagefraction}{0.95}

\usepackage{bec}
\setboolean{showcomments}{true}
\setboolean{showpunch}{false}
\verbosecitesfalse

\usepackage[caption=false,font=footnotesize]{subfig}

\usepackage{mathpartir}
\renewcommand{\TirNameStyle}[1]{\footnotesize\textsc{#1}}
\usepackage{tikz}
  \usetikzlibrary{arrows,shapes}
\usepackage{framed}

\renewcommand{\theequation}{\fnsymbol{equation}}
\addtocounter{equation}{1}

\newcommand{\SQUEEZE}[1]{}

\newlength{\saveboxsep}

\makeatletter
\long\def\optpair@[#1][#2]{\tuple{#1, #2}}

\newcommand{\ToolName}{Divva}
\newcommand{\Tool}{\textsc{\ToolName}}

\renewcommand{\vec}[1]{\overline{#1}}

\newcommand{\unfold}{\stackrel{\rightsquigarrow}{\textrm{\tiny\vphantom{u}\smash{unfold}}}}
\newcommand{\fold}{\stackrel{\leftsquigarrow}{\textrm{\tiny\vphantom{u}\smash{fold}}}}
\newcommand{\inconc}[3][]{#2 \in \AIconcfn_{#1}(#3)}

\newcommand{\jinclusion}[4][]{#2 \vdash #3 \sqsubseteq_{#1} #4}
\newcommand{\jjoin}[5]{#1, #2 \vdash #3 \AIjoin #4 \hookrightarrow #5}

\newcommand{\jpvaluacmd}[2]{#1 \vdash #2}
\newcommand{\bind}{\operatorname{bind}}

\newcommand{\absof}[1]{\ensuremath{\widehat{\smash{#1}}}}
\newcommand{\progof}[1]{\ensuremath{\check{\smash{#1}}} }
\newcommand{\absofset}[1]{\ensuremath{\hat{#1}}}

\newcommand{\disj}{\operatorname{disj}}
\newcommand{\disji}{i}

\newcommand{\fmtkw}[1]{\mathbf{#1}}
\newcommand{\var}{x}
\newcommand{\varx}{x}
\newcommand{\vary}{y}

\newcommand{\expr}{e}
\newcommand{\exprset}{\SetName{Expr}}
\newcommand{\fld}{f}
\newcommand{\cmd}{c}
\newcommand{\pcmd}{\progof{\cmd}}
\newcommand{\ploc}{\ell}
\newcommand{\plocset}{\SetName{ProgLoc}}
\newcommand{\ploccmd@}{}
\long\def\ploccmd@[#1][#2]{#1 \colon #2}
\newcommand{\ploccmd}{\optparams{\ploccmd@}{[\ploc][\pcmd]}}

\newcommand{\pcassumekw}{\fmtkw{assume}}

\newcommand{\pcunfoldkw}{\fmtkw{unfold}}
\newcommand{\pcjoinkw}{\fmtkw{join}}
\newcommand{\pcassume}[2]{\pcassumekw\;#2\;\{#1\}}

\newcommand{\pcunfold}[2]{\pcunfoldkw\;#2\;\{#1\}}
\newcommand{\pcjoin}[3]{\pcjoinkw(#2; #3)\;\{#1\}}
\newcommand{\pejoin}[2]{\phi(#1,#2)}

\newcommand{\cenv}{\eta}

\newcommand{\cmem}{\sigma}
\newcommand{\cmemset}{\SetName{Mem}}
\newcommand{\caddr}{a}
\newcommand{\caddrset}{\SetName{Addr}}
\newcommand{\aaddr}{\absof{\caddr}}

\newcommand{\cval}{v}
\newcommand{\cvalset}{\SetName{Val}}
\newcommand{\aval}{\absof{\cval}}
\newcommand{\avalset}{\absofset{\cvalset}}
\newcommand{\cvalua}{\pi}
\newcommand{\pvalua}{\progof{\cvalua}}
\newcommand{\cvaluaset}{\SetName{Valua}}
\newcommand{\cmemvalua}{\optparams{\optpair@}{[\cmem][\cvalua]}}

\newcommand{\aenv}{\absof{\cenv}}

\newcommand{\amem}{\absof{\cmem}}

\newcommand{\amemset}{\absofset{\cmemset}}
\newcommand{\astore}{\absof{\rho}}
\newcommand{\astoreset}{\absofset{\SetName{Store}}}
\newcommand{\afinstoreset}{\absofset{\SetName{S}}}
\newcommand{\apure}{\absof{\cvalua}}
\newcommand{\apuredom}{\absofset{\SetName{\cvaluaset}}}
\long\def\astorepure@[#1][#2]{#1 \land #2}
\newcommand{\astorepure}{\optparams{\astorepure@}{[\astore][\apure]}}
\newcommand{\aview}{\absof{\omega}}
\newcommand{\aviewset}{\absofset{\SetName{View}}}
\newcommand{\astoreview@}{}
\long\def\astoreview@[#1][#2]{#1 \landtrue #2}
\newcommand{\astoreview}{\optparams{\astoreview@}{[\astore][\aview]}}
\newcommand{\avaluatr}{\Psi}

\newcommand{\chker}{k}
\newcommand{\achkercall}[1][]{\icall{\aaddr_{#1}}{\chker}{\vec{\aval_{#1}}}}
\newcommand{\achkreg}{\kappa}

\newcommand{\amultichkreg}{\varepsilon}
\long\def\astoremulti@[#1][#2]{#1 \dashv #2}
\newcommand{\astoremulti}{\optparams{\astoremulti@}{[\astore][\amultichkreg]}}

\newcommand{\ichk}[1]{#1\textsf{sc}}
\newcommand{\iseg}[2]{#1#2\textsf{segsc}}

\newcommand{\fmtvar}[1]{\mathtt{#1}}
\newcommand{\fmtfld}[1]{\mathtt{#1}}
\newcommand{\fmtidef}[1]{\mathsf{#1}}
\newcommand{\icall}[3]{#1 \cdot #2(#3)}

\newcommand{\pt}{\ptsto}

\newcommand{\fldoff}[2]{#1 \cdot #2}
\newcommand{\fldpt}[3]{\fldoff{#1}{#2} \pt #3}

\newcommand{\nullval}{\mathsf{null}}

\newcommand{\bst}{\text{\colordullx{0}$\fmtidef{bst}$}}
\newcommand{\bstsc}{\text{\colordullx{1}$\fmtidef{bstsc}$}}
\newcommand{\hashtrie}{\fmtidef{hashtrie}}
\newcommand{\bstcall}[3]{\icall{#1}{\bst}{#2, #3}}
\newcommand{\bstsccall}[5]{\icall{#1}{\bstsc}{#2, #3, #4, #5}}
\newcommand{\sym}[1]{\absof{\mathit{#1}}}
\newcommand{\symt}{\sym{t}}
\newcommand{\symn}{\sym{n}}
\newcommand{\fmthypp}[1]{\text{\colordullx{3}$#1$}}
\newcommand{\fmtchkp}[1]{#1}
\newcommand{\symmin}{\fmtchkp{\sym{min}}}
\newcommand{\symmax}{\fmtchkp{\sym{max}}}
\newcommand{\symhmin}{\fmthypp{\sym{hmin}}}
\newcommand{\symhmax}{\fmthypp{\sym{hmax}}}
\newcommand{\symu}{\sym{u}}
\newcommand{\symumin}[1][]{\MVarNamed{\sym{u}}{min$#1$}}
\newcommand{\symumax}[1][]{\MVarNamed{\sym{u}}{max$#1$}}
\newcommand{\syma}{\aaddr}
\newcommand{\symv}{\aval}
\newcommand{\symvninf}{\symv_{\operatorname{-}\infty}}
\newcommand{\symvpinf}{\symv_{\infty}}
\newcommand{\syml}{\sym{l}}
\newcommand{\symr}{\sym{r}}
\newcommand{\fldv}{\fmtfld{v}}
\newcommand{\fldl}{\fmtfld{l}}
\newcommand{\fldr}{\fmtfld{r}}
\newcommand{\varc}{\fmtvar{c}}
\newcommand{\vart}{\fmtvar{t}}
\newcommand{\varn}{\fmtvar{n}}
\newcommand{\varp}{\fmtvar{p}}
\newcommand{\varoldv}{\fmtvar{oldv}}
\newcommand{\varhmin}{\fmtvar{hmin}}
\newcommand{\varhmax}{\fmtvar{hmax}}
\newcommand{\varmin}{\fmtvar{min}}
\newcommand{\varmax}{\fmtvar{max}}
\newcommand{\fldderef}[2]{\fmtvar{#1}\mathord{\texttt{->}}#2}

\makeatother

\tikzstyle{nd}=[circle, draw, inner sep=1pt, minimum size=5mm, font=\footnotesize, text centered]
\tikzstyle{chk}=[-stealth, line width=2pt]
\tikzstyle{fld}=[->, line width=0.4pt]
\tikzstyle{ebelowleft}=[pos=0,below right, font=\footnotesize]
\tikzstyle{ebelowright}=[pos=1,below left, font=\footnotesize]
\tikzstyle{ebelowmid}=[pos=0.5,below, font=\footnotesize]
\tikzstyle{eabovemid}=[pos=0.5,above, font=\footnotesize]
\tikzstyle{nabove}=[above=2.5mm, font=\scriptsize]
\tikzstyle{nbelow}=[below=2.5mm, font=\small]
\tikzstyle{levelseg}=[level distance=42mm]
\tikzstyle{levelchk}=[level distance=18mm]
\tikzstyle{levelfld}=[level distance=15mm]
\tikzstyle{leveldata}=[level distance=15mm]

\tikzstyle{edge from parent}=[draw,fld]

\mprset{sep=1.5em}

\tikzstyle{rect}=[rectangle,thick,minimum size=0.7cm,draw=blue!80]
\tikzstyle{redrect}=[rectangle,thick,minimum size=0.7cm,draw=red!80]

\begin{document}


\title{Synthesizing Short-Circuiting Validation of Data Structure Invariants}

\authorinfo{
\begin{tabular}[t]{c}
{\large Yi-Fan Tsai\and Devin Coughlin\and Bor-Yuh Evan Chang} \\[3pt]
{\normalsize University of Colorado Boulder} \\
{\normalsize\textsf{\{yifan.tsai,devin.coughlin,evan.chang\}@colorado.edu}}
\end{tabular}
\and
\begin{tabular}[t]{c}
{\large Xavier Rival} \\[3pt]
{\normalsize INRIA/ENS/CNRS Paris} \\
{\normalsize\textsf{xavier.rival@ens.fr}}
\end{tabular}
}{}{}

\maketitle

\begin{abstract}

This paper presents \emph{incremental verification-validation}, a novel approach for checking rich data structure invariants expressed as separation logic assertions. Incremental verification-validation combines static verification of separation properties with efficient, \emph{short-circuiting} dynamic validation of arbitrarily rich data constraints. A data structure invariant checker is an inductive predicate in separation logic with an executable interpretation; a short-circuiting checker is an invariant checker that stops checking whenever it detects at \emph{run time} that an assertion for some sub-structure has been fully proven \emph{statically}. At a high level, our approach does two things: it statically proves the separation properties of data structure invariants using a static shape analysis in a standard way but then leverages this proof in a novel manner to synthesize short-circuiting dynamic validation of the data properties. As a consequence, we enable dynamic validation to make up for imprecision in sound static analysis while simultaneously leveraging the static verification to make the remaining dynamic validation efficient. We show empirically that short-circuiting can yield asymptotic improvements in dynamic validation, with low overhead over no validation, even in cases where static verification is incomplete.

\end{abstract}



\section{Introduction}
\label{sec:introduction}


We consider the problem of efficient and safe dynamic validation of deep, program-specific data structure invariants.  Such invariants combine
pointer-shape 
properties (e.g., ``\code{l} is an acyclic, singly-linked list'')
with
data-value properties (e.g., ``the integer data elements of \code{l}
are sorted in non-decreasing order'').
Combined
shape-data properties are both difficult for the programmer to maintain and
challenging for fully-automatic static analysis to verify.
While modern shape analyzers can reason effectively about pointer-shape properties---especially for structures with limited sharing, such as lists and
trees---shape-\emph{data}
analyzers are limited by their dependency on and interaction with base
data-value abstract domains or
solvers\citeverbelse{chang+2008:relational-inductive,
  magill+2010:automatic-numeric, DBLP:conf/sas/McCloskeyRS10,
  DBLP:conf/pldi/BouajjaniDES11, gopan+2005:framework-numeric,
  gulwani+2008:lifting-abstract}{chang+2008:relational-inductive,
  magill+2010:automatic-numeric, DBLP:conf/sas/McCloskeyRS10,
  DBLP:conf/pldi/BouajjaniDES11}.

\newcommand{\Unsafe}{noV}
\newcommand{\Noninc}{dynV}
\newcommand{\Inc}{iVV}

\newcommand{\liststruct}{\texttt{olist}}
\newcommand{\listdrop}{\liststruct{} $m$ \texttt{drop}s}
\newcommand{\listconcat}{\liststruct{} $m$ \texttt{concat}s}
\newcommand{\listinsert}{\liststruct{} $m$ \texttt{insert}s}
\newcommand{\listdelete}{\liststruct{} $m$ \texttt{delete}s}
\newcommand{\listinsdel}{\liststruct{} 50-50}
\newcommand{\bststruct}{\texttt{bst}}
\newcommand{\bstdeletemin}{\bststruct{} $m$ \texttt{deletemin}s}
\newcommand{\bstexciseroot}{\bststruct{} $m$ \texttt{exciseroot}s}
\newcommand{\bstinsert}{\bststruct{} $m$ \texttt{insert}s}
\newcommand{\bstdelete}{\bststruct{} $m$ \texttt{delete}s}
\newcommand{\bstinsdel}{\bststruct{} 50-50}
\newcommand{\treapstruct}{\texttt{treap}}
\newcommand{\treapdelete}{\treapstruct{} $m$ \texttt{delete}s}
\newcommand{\treapinsert}{\treapstruct{} $m$ \texttt{insert}s}
\newcommand{\treapinsdel}{\treapstruct{} 50-50}
\newcommand{\hashtriestruct}{\texttt{hashtrie}}
\newcommand{\hashtriewrite}{\hashtriestruct{} $m$ \texttt{write}s}
\newcommand{\hashtrieinsert}{\hashtriestruct{} $m$ \texttt{insert}s}
\newcommand{\hashtriedelete}{\hashtriestruct{} $m$ \texttt{delete}s}
\newcommand{\hashtrieinsdel}{\hashtriestruct{} 50-50}

\begin{figure}\centering\small
\includegraphics[width=0.85\linewidth]{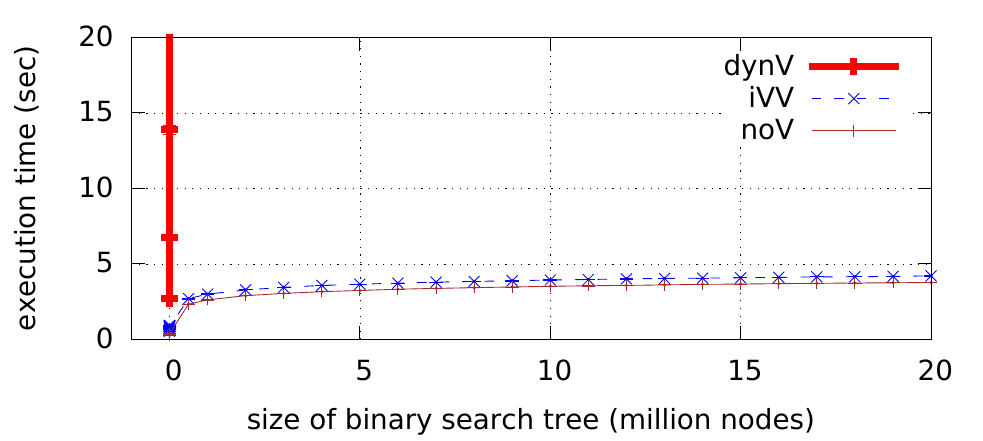}
\caption{Frequent, dynamic validation does not scale. This plot shows execution times for a binary search tree
  workload that performs a million random deletions, each followed by inserting
  the deleted value again over increasingly-sized trees.
  The \textbf{\Noninc{}} line applies dynamic validation after each
  delete-insert pair; \textbf{\Unsafe{}} has no validation; and
  \textbf{\Inc{}} validates after each pair using short-circuiting---this paper.
  The smallest size
  in all configurations is 100 nodes.  The size on the \Noninc{} line
  goes up by 100 at each step.
   }
\label{fig:unscalable}
\end{figure}

Static verification has always held the promise to improve software by
definitively ruling out entire classes of bugs before execution. But this promise
holds for a particular
program of interest only if the developer is able to obtain a full verification
for it. 
Dynamic validation compensates for the
shortcomings of static verification by trapping unsafe operations at
run time (e.g., by halting execution or raising an exception).
Several
systems\citeverbelse{DBLP:conf/sac/FahndrichBL10,barnett+2011:specification-verification:,burdy+2005:overview-tools,DBLP:conf/oopsla/ReichenbachISAG10} {DBLP:conf/sac/FahndrichBL10,barnett+2011:specification-verification:,DBLP:conf/oopsla/ReichenbachISAG10}
have made run-time checks easier to apply; however, run-time
enforcement of data structure invariants, in particular, have remained problematic
for two main reasons:

\begin{asparadesc}
\item[Checking Overhead.]
The overhead of run-time checking can be
  significant 
  because a data structure invariant is a universal
  property over a set of objects---and thus checking
  the invariant is usually linear-time in the size of the structure.
  Since adding dynamic validation can 
  change the
  asymptotic complexity of data structure operations, validation can cause
  dramatic slowdowns.
  \Figref{unscalable} shows such slowdowns for a 
  workload of binary
  search tree operations on increasing data structure sizes. The
  dynamic validation variant (\Noninc{}) can only reach a binary
  search tree with 300 nodes in 15 seconds, while the no validation
  variant (\Unsafe{}) can reach 20,000,000 nodes in 5 seconds.
  Because of the linear-time complexity of validation, the
  \Noninc{} variant is $O(n)$, while the \Unsafe{} variant is $O(\lg n)$
  where $n$ is the size of the tree.
  Thus, frequent checking of data
  structure invariants can be infeasible even in debugging
  scenarios.
\item[Additional Risk of Memory Faults.] Run-time checking may
  introduce new, unexpected pointer bugs.  Traversing pointers in a
  malformed data structure
  could lead to unanticipated bad dereferences (e.g., on null or
  freed pointers) or even non-termination (e.g., on unexpected cyclic structures).
  Thus, such checking may introduce unacceptable risk in
  deployment scenarios.
\end{asparadesc}
%

\PUNCH{Need: Dynamic validation may be fine for div-by-zero,
  null-pointer dereference, array-bounds, etc. but not complex data
  structure invariants. We would like programmers to sprinkle
  assertions throughout for safety but ...}

In this paper, we address these concerns by applying
static shape analysis to%
\begin{inparaenum}[(1)]
\item synthesize short-circuiting dynamic data structure
  validation
  to reduce run-time checking overhead
  and
\item statically prove separation properties
to eliminate the risk of memory faults.
\end{inparaenum}

\PUNCH{The run-time checking overhead has been
  addressed by trading-off by storage overhead.  Separation properties
  has been checked dynamically with some storage overhead and modified
  run-time. We have no heap storage overhead.}

\paragraph{Contributions.}

This paper is a response to the inevitable
imprecision of static shape-data analysis.  While a static analysis can
always be made precise enough to verify any given set of benchmarks
(with an appropriate amount of research, development, and specification effort),
no
analysis can ever be precise enough for an ever growing set.  Our
approach is a step towards providing the benefit of \emph{incremental verification-validation} even in the face
of imprecise invariant inference or insufficient specification.
To clearly
distinguish when some part of a data
structure invariant is checked, we reserve \emph{verification} for
static checking and \emph{validation} for dynamic checking.

In \secref{shortcircuit}, we discuss the notion of \emph{short-circuiting invariant checkers}:
  validation code that leverages a hypothetical fact to short-circuit
  dynamic validation of a data structure invariant.
  Short-circuit validation is the result of a simple but key observation: when asked to prove $P(x)$, we can decompose this proof into proving, for example, $P(y)$ for some $y$ and separately that $x = y$.
  This observation enables us to
  connect facts that we verify statically (e.g., $P(y)$)
  and facts that we need to validate dynamically (e.g., $x = y$).
  In particular, we get the benefit of short-circuiting validation when we carefully choose the facts to validate dynamically to be fast and cheap to check yet helpful in making the static verification more feasible. A significant challenge in this work is that $P$ is inductive and effective short-circuiting can only be obtained when static verification and dynamic validation are interleaved in the recursion.

To do so, we make the following technical contributions:
\begin{itemize}\itemsep 0pt
\item We describe a static shape-data analysis needed to synthesize short-circuiting validation (\secref{static-analysis}) driven by the aforementioned proof decomposition. In particular, we introduce the \emph{validated view} abstract domain, which augments the static analysis with the ability to track arbitrarily complex, uninterpreted data-value facts (e.g., to statically prove $P(y)$).  Then, we describe a code instrumentation that reifies, as programs variables, the existential logic variables in the statically-inferred shape-data invariants. Conceptually, this \emph{logic-variable reification} provides hooks at run time to the proof constructed by the static shape-data analysis (e.g., to dynamically validate $x = y$).

%
%

\item We introduce a \emph{subtraction-directed synthesis} approach to generate both
  short-circuiting checkers
  and the calls to them, using invariants inferred by the static shape-data analysis
  and reified logic variables
  (\secref{dynamic-validation}).  The result is
  incrementalized dynamic validation of inductive data
  structure invariants with no run-time memoization.
  
\item We empirically evaluate our approach by
  applying a prototype tool (called \Tool{}) to
  verification and validation of operations with
  complex shape-data invariants (\secref{empirical-evaluation}).  We
  see evidence of asymptotic improvements in dynamic validation
  cost---for example, the \Inc{} line in
  \figref{unscalable}.
  The improvements hold even for invariants with data
  properties beyond the reach of most static verifiers (e.g.,
  involving hashing and bit masking).
\end{itemize}

Our mixed static-dynamic technique differs from
prior fully dynamic, run-time approaches to reducing checking overhead.
Ditto~\cite{shankar+2007:ditto:-automatic} incrementalizes data structure validation checks by
memoizing validation calls and then only executing
subsequent checks for the objects for which the memoized checks have
been invalidated (e.g., by writes). This fully dynamic approach essentially
trades off the execution overhead of validation checks for the memory
overhead of maintaining a memo table in a shadow heap.
In contrast, our approach leverages a static proof of
separation properties to synthesize 
incremental dynamic data structure validation without the memory
overhead of memoization or the execution overhead of write barriers.

\section{Goal: Short-Circuiting Validation}
\label{sec:shortcircuit}


In this section, we consider a simple but illustrative example of a 
data structure operation with an assertion of its invariant.  We
illustrate that this assertion can be incrementally verified and validated in a sound manner.

We consider binary search trees here, as the data structure requires
a well-known 
intertwined shape-data invariant. 
Our technique also applies to (and is even more important for) more complex data invariants.  In \secref{empirical-evaluation}, we evaluate our implementation on a hash trie, which is a tree-structured hash map 
that uses bit-blocks from hash keys to form a trie~\cite{bagwell2001:ideal-hash}.  This structure  
employs bit-masking, hash codes, and rehashing on collision.  Even for state-of-the-art static verifiers, such as Thor~\cite{magill+2010:automatic-numeric}, this data structure invariant is challenging and likely out of reach for current techniques~\cite{Magill-Personal-Communication-2014}.
Another important feature of our approach is the ability to obtain incremental dynamic validation even when the pre-condition specification is insufficient (not strong enough) to statically prove the post-condition, which we demonstrate here with the \code{setroot} example in \figref{validation}.

\newsavebox{\SBoxBtCDef}
\begin{lrbox}{\SBoxBtCDef}\small
\begin{lstlisting}[language=C]
typedef#\,#struct#\,#node#\,#{#\,#val v;#\,#struct#\,#node* l;#\,#struct#\,#node* r;}*#\,#bt;
\end{lstlisting}
\end{lrbox}

We consider a C type declaration for a binary tree \code{bt}:
\[
\scalebox{0.95}{\usebox{\SBoxBtCDef}}
\]
Each node contains three fields: \code{v}, which stores a data
value, and fields \code{l} and \code{r} that contain pointers to left and right
sub-trees. We leave unspecified the type of data values (\code{val}).
In \figref{bst-sepchecker}, we give an inductive definition for a binary search tree in a standard separation logic notation.
As a convention, we use hatted letters like $\symt$ for
symbolic logic variables used in static analysis and write inductive
definitions like $\bst$ with a distinguished recursion parameter as in
$\icall{\symt}{\bst}{\ldots}$.
A binary search tree is a binary tree where the data
values obtained from the in-order traversal are ordered---and thus the
search tree invariant specification $\bst$ includes the allowable range for the data
values in that sub-tree.
The $\icall{\symt}{\bst}{\symmin, \symmax}$ predicate states that $\symt$ points to a binary search tree with allowable range $[\symmin, \symmax)$; we emphasize with shading the parts of the $\bst$ definition capturing the search \emph{data property}.
As is well-understood, using separating conjunction $\lsep$ between the root node pointed to by $\symt$ and the sub-trees 
pointed to by $\syml$ and $\symr$ specifies that $\symt$ is structurally 
a tree (i.e., does not have sharing or cycles); in other words, the non-shaded parts of the definition correspond the binary tree \emph{shape property}.

In \figref{bst-sepsimple}, we rewrite the inductive definition $\bst$ in a shorthand that suggests an executable interpretation. Disjunctions are expressed as if-then-else conditionals and existential variables correspond only to values obtained from dereferencing the recursion parameter $\symt$. Note that this definition is simply a shorthand and has the same meaning as the definition in \figref{bst-sepchecker} expressed in a more standard notation.

\newcommand{\CodeSpacer}{\rule{0pt}{2ex}}
\newcommand{\InvBox}[1]{\hfill\fbox{#1}}

\setlength{\saveboxsep}{\fboxsep}
\fboxsep 1pt

\newsavebox{\SBoxBst}
\begin{lrbox}{\SBoxBst}\scriptsize
\begin{lstlisting}[language=C,alsolanguage=Spec,alsolanguage=highlighting,alsolanguage=SqueezeOps]
typedef#\,#struct btnode {#\,#val v; struct#\,#btnode* l; struct#\,#btnode* r;#\,#}* bt;
#\rule{0pt}{3ex}#bool |bst|(bt n, val min, val max) { val v;
  return n == NULL ? true : (v = n->v,  min <= v && v < max &&
    |bst|(n->l, min, v) && |bst|(n->r, v, max));
}
\end{lstlisting}
\end{lrbox}

\newsavebox{\SBoxSepCheckerExample}
\begin{lrbox}{\SBoxSepCheckerExample}\small
\( \begin{array}{@{}l@{}}
 \icall{\symt}{\bst}{\hibox{\symmin, \symmax}}
    \quad\defeq\quad \lemp \spland \symt = \nullval \qquad\lor\qquad\exists \symv,\syml,\symr.\\[0.5ex]
      \quad\fldpt{ \symt }{ \fldv }{ \symv }
      \splsep
      \fldpt{ \symt }{ \fldl }{ \syml }
      \splsep
      \fldpt{ \symt }{ \fldr }{ \symr }
      \splsep
      \icall{\syml}{\bst}{\hibox{\symmin,\,\symv}}
      \splsep
      \icall{\symr}{\bst}{\hibox{\symv,\,\symmax}}
      \\
      \qquad\qquad\qquad\spland
      \vphantom{\hibox{\symmin}}
    \symt \neq \nullval
      \spland \hibox{ \symmin \leq \symv \spland \symv < \symmax }
\end{array}\)
\end{lrbox}

\newsavebox{\SBoxBstSimple}
\begin{lrbox}{\SBoxBstSimple}\small
\begin{lstlisting}[language=Spec]
#$\icall{\symt}{\bst}{\hibox{\symmin, \symmax}} \quad\defeq\quad$# if #$\symt = \nullval$# then #$\lemp$# else
#$\quad\icall{\fldderef{\symt}{\fldl}}{\bst}{\hibox{\symmin, \fldderef{\symt}{\fldv}}} \splsep \icall{\fldderef{\symt}{\fldr}}{\bst}{\hibox{\fldderef{\symt}{\fldv}, \symmax}}$#                                           ##
#$\qquad\qquad\qquad\spland \hibox{\symmin \leq \fldderef{\symt}{\fldv} \spland \fldderef{\symt}{\fldv}< \symmax}$#
\end{lstlisting}
\end{lrbox}

\newsavebox{\SBoxSetRoot}
\begin{lrbox}{\SBoxSetRoot}\small\lstnonum
\begin{lstlisting}[language=C,alsolanguage=Spec,alsolanguage=highlighting,alsolanguage=SqueezeOps,style=number,name=SetRoot]
void setroot(bt t, val v) {#\lstbeginn#
 assume(#$\bstcall{\vart}{-\infty}{\infty}$#);#\label{line:setroot-assume}#
 assume(#$\vart \neq \nullval$#);#\label{line:setroot-assume-nonnull}#
 t->v = v;#\label{line:setroot-set}#
 assert(#$\bstcall{\vart}{-\infty}{\infty}$#);#\label{line:setroot-assert}#
#\label{line:inc-invocation-left}#
#\label{line:inc-invocation-right}\label{line:setroot-end}\lststopn#
}
\end{lstlisting}
\end{lrbox}

\newsavebox{\SBoxSetRootIncr}
\begin{lrbox}{\SBoxSetRootIncr}\small
\begin{lstlisting}[language=C,alsolanguage=Spec,alsolanguage=highlighting,alsolanguage=SqueezeOps,name=SetRootIncr]
void setroot(bt t, val v) {
 assume(#$\bstcall{\vart}{-\infty}{\infty}$#);
 assume(#$\vart \neq \nullval$#);
 val oldv = t->v; t->v = v;      ##
 assert(#$-\infty$# <= t->v && t->v < #$\infty$# &&
  |1:bstsc|(t->l, #$-\infty$#, t->v, #$-\infty$#, oldv)
  #\,#&& |1:bstsc|(t->r, t->v, #$\infty$#, oldv, #$\infty$#));
}
\end{lstlisting}
\end{lrbox}

\tikzstyle{vnd}=[rectangle, draw, inner sep=1pt, minimum size=5mm, font=\footnotesize, text centered, color=dullredrgb, text=black]
\tikzstyle{tnd}=[regular polygon, regular polygon sides=3, shape border rotate=90, draw, inner sep=1pt, minimum size=8mm, font=\footnotesize\color{black}, text centered, color=dullgreenrgb]

\newsavebox{\SBoxSetRootIncrExecution}
\begin{lrbox}{\SBoxSetRootIncrExecution}\small
\begin{tikzpicture}[grow=right]
\tikzstyle{level 1}=[level distance=15mm, sibling distance=22mm]
\tikzstyle{level 2}=[level distance=20mm, sibling distance=16mm]
\tikzstyle{level 3}=[level distance=20mm, sibling distance=12mm]
\tikzstyle{level 3}=[level distance=20mm, sibling distance=10mm]
\node[vnd,label=above:{\relsize{-1}$\texttt{oldv} = 8$},label=below:{$\texttt{t}$}] {$9$}
    child {
        node[vnd,label=below:{$(-\infty, 9, \hibox{-\infty, 8})$}] {$2$}
        child[fld] {
            node[tnd,label=right:{$(-\infty, 2, \hibox{-\infty, 2})$}] {}
        }
        child {
            node[vnd,label=below:{$(2, 9, \hibox{2, 8})$}] {$4$}
            child[] {
                node[tnd,label=below:{$(2, 4, \hibox{2, 4})$}] {}
            }
            child {
                node[vnd,label=above:{$(4, 9, \hibox{4, 8})$}] {$6$}
                child[] {
                    node[tnd,label=below:{$(4, 6, \hibox{4, 6})$}] {}
                }
                child[dashed,thick] {
                    node[] {...}
                }
            }
        }
        edge from parent node[ebelowmid] {$\fldl$}
    }
    child[] {
        node[vnd,label=above:{$(9, \infty, \hibox{8, \infty})$}] {$10$}
        child[dashed] {
            node[] {...}
        }
        child[fld] {
            node[tnd,label=below:{$(10, \infty, \hibox{10, \infty})$}] {}
        }
        edge from parent node[eabovemid] {$\fldr$}
    }
;
\end{tikzpicture}
\end{lrbox}

\newsavebox{\SBoxBstScCode}
\begin{lrbox}{\SBoxBstScCode}\small\lstnonum
\begin{lstlisting}[language=C,alsolanguage=Spec,alsolanguage=highlighting,style=number]
bool |1:bstsc|(bt n, val min, val max, val hmin, val hmax) {#\lstbeginn#
  assume( #$\bstcall{\varn}{\varhmin}{\varhmax}$# );
  if (min == hmin && max == hmax$\vphantom{\hibox{\texttt{h}}}$) return true;
  return n == NULL ? true : min <= n->v && n->v < max &&
    |1:bstsc|(n->l, min, n->v, #$\hibox{\texttt{hmin,\;n->v}}$#)
    #\,#&& |1:bstsc|(n->r, n->v, max, #$\hibox{\texttt{n->v,\;hmax}}$#);#\lststopn#
}
\end{lstlisting}
\end{lrbox}

\setlength{\fboxsep}{\saveboxsep}

\begin{figure}\centering
\subfloat[An inductive definition $\bst$ for binary search trees in separation logic.]{\label{fig:bst-sepchecker}\rule{1em}{0pt}\scalebox{0.9}{\usebox{\SBoxSepCheckerExample}}\rule{1em}{0pt}} \\
\subfloat[The inductive definition $\bst$ from (a) written in a shorthand for checkers. A dereference expression $\fldderef{\aaddr}{\fld}$ corresponds to the value of an implicitly separated points-to predicate (e.g., the referencing expression $\fldderef{\symt}{\fldv}$ corresponds to $\symv$ in the points-to predicate $\fldpt{\symt}{\fldv}{\symv}$ from (a)).]{\label{fig:bst-sepsimple}%
\rule{4em}{0pt}\scalebox{0.9}{\usebox{\SBoxBstSimple}}\rule{4em}{0pt}%
} \\
\subfloat[With full verification.]{\label{fig:validation:bst}\rule{0.7em}{0pt}\ovalbox{\scalebox{0.9}{\usebox{\SBoxSetRoot}}}}\rule{2pt}{0pt}\subfloat[With short-circuiting validation.]{\label{fig:validation:bstsc}\ovalbox{\scalebox{0.9}{\usebox{\SBoxSetRootIncr}}}} \\
\subfloat[Compilation of the short-circuiting invariant checker $\bstsc$ from \figref{bstsc:final} to C code (once separation has been proven statically).]{\label{fig:bstsc:code}%
\rule{1em}{0pt}\scalebox{0.9}{\usebox{\SBoxBstScCode}}%
}
\\
\subfloat[Illustrating short-circuiting validation, executing lines 4--6 of (d), on a binary search tree where the data value at the root node was changed from 8 to 9 by \texttt{setroot}.  The four-tuple labels correspond to the last four arguments of calls to \codex{language=highlighting}{|1:bstsc|}.]{\label{fig:validation:increxecution}\rule{0.6em}{0pt}\scalebox{0.9}{\usebox{\SBoxSetRootIncrExecution}}\rule{0.6em}{0pt}}
\caption{Validating the binary search tree invariant in \texttt{setroot} after the data
  value of the root node is set to \code{v}.
}
\label{fig:validation}
\end{figure}


\newsavebox{\SBoxBstScShallow}
\setlength{\saveboxsep}{\fboxsep}
\fboxsep 1pt
%
%
%
%

\begin{lrbox}{\SBoxBstScShallow}\small\lstnonum
\begin{lstlisting}[language=Spec,style=number]
#$\icall{\symn}{\fmtidef{bstsc\_a}}{\symmin, \symmax, \symhmin, \symhmax} \quad\defeq$##\lstbeginn#
  #$\bstcall{\symn}{\symhmin}{\symhmax} \limp$##\label{line:bstsca-assume}\label{line:bstsc-assume}#
#$\vphantom{\hibox{\bst(\symhmin)}}$#
#$\vphantom{\hibox{\bst(\symhmin)}}$#
#$\vphantom{\symhmin}$#
    if #$\hibox{\symmin = \symhmin \spland \symmax = \symhmax}$# then #$\ltrue$#                              ###\label{line:bstsca-short}#
    else #$\bstcall{\symn}{\symmin}{\symmax}$#
#$\vphantom{\hibox{\symhmin}}$##\label{line:bstscb-left}\label{line:bstsc-left}#
#$\vphantom{\hibox{\symhmin}}$##\label{line:bstscb-right}\label{line:bstsc-right}#
##
\end{lstlisting}
\end{lrbox}


\newsavebox{\SBoxBstScUnfold}
\begin{lrbox}{\SBoxBstScUnfold}\small\lstnonum
\begin{lstlisting}[language=Spec]
#$\icall{\symn}{\fmtidef{bstsc\_b}}{\symmin, \symmax, \symhmin, \symhmax} \quad\defeq$#
 (if #$\symn = \nullval$# then #$\lemp\vphantom{\symhmin}$#
  else #$\hibox{\bstcall{ \fldderef{\symn}{\fldl} }{\symhmin}{ \fldderef{\symn}{\fldv} }}$#
  #$\splsep \hibox{\bstcall{ \fldderef{\symn}{\fldr} }{ \fldderef{\symn}{\fldv} }{\symhmax}}$#
  #$\spland \symhmin \leq \fldderef{\symn}{\fldv} \land \fldderef{\symn}{\fldv} < \symhmax$#) #$\limp$#
    if #$\symmin = \symhmin \spland \symmax = \symhmax$##$\vphantom{\hibox{\symhmin}}$# then #$\ltrue$#                             ##
    else if #$\symn = \nullval$# then #$\lemp$#
    else #$\vphantom{\hibox{\symhmin}}\bstcall{ \fldderef{\symn}{\fldl} }{\symmin}{ \fldderef{\symn}{\fldv} }$#
    #$\splsep \vphantom{\hibox{\symhmax}}\bstcall{ \fldderef{\symn}{\fldr} }{ \fldderef{\symn}{\fldv} }{\symmax}$#
    #$\spland \symmin \leq \fldderef{\symn}{\fldv} \land \fldderef{\symn}{\fldv} < \symmax$#
\end{lstlisting}
\end{lrbox}

%
%
%

\newsavebox{\SBoxBstSc}
\begin{lrbox}{\SBoxBstSc}\small\lstnonum
\begin{lstlisting}[language=Spec]
#$\icall{\symn}{\bstsc}{\symmin, \symmax, \symhmin, \symhmax} \quad\defeq$#
 (if #$\symn = \nullval$# then #$\lemp\vphantom{\symhmin}$#
  else #$\vphantom{\hibox{\bst(\symhmin)}}\bstcall{ \fldderef{\symn}{\fldl} }{\symhmin}{ \fldderef{\symn}{\fldv} }$#
  #$\splsep \vphantom{\hibox{\bst(\symhmin)}}\bstcall{ \fldderef{\symn}{\fldr} }{ \fldderef{\symn}{\fldv} }{\symhmax}$#
  #$\spland \symhmin \leq \fldderef{\symn}{\fldv} \land \fldderef{\symn}{\fldv} < \symhmax$#) #$\limp$#
    if #$\symmin = \symhmin \spland \symmax = \symhmax$##$\vphantom{\hibox{\symhmin}}$# then #$\ltrue$#                             ##
    else if #$\symn = \nullval$# then #$\lemp$#
    else #$\icall{ \fldderef{\symn}{\fldl} }{\bstsc}{ \symmin, \fldderef{\symn}{\fldv}, \hibox{\symhmin, \fldderef{\symn}{\fldv}} }$#
    #$\splsep \icall{ \fldderef{\symn}{\fldr} }{\bstsc}{ \fldderef{\symn}{\fldv}, \symmax, \hibox{\fldderef{\symn}{\fldv}, \symhmax} }$#
    #$\spland \symmin \leq \fldderef{\symn}{\fldv} \land \fldderef{\symn}{\fldv} < \symmax$#
\end{lstlisting}
\end{lrbox}
\setlength{\fboxsep}{\saveboxsep}

\SQUEEZE{
\paragraph{Validating an Operation.}
}

Now consider the \code{setroot} procedure shown in \figref{validation:bst}.
The programmer states the assumption (i.e., pre-condition) that $\vart$ on entry points to a binary search tree satisfying $\bstcall{\vart}{-\infty}{\infty}$ (\reftxt{line}{line:setroot-assume}) and is non-null (\reftxt{line}{line:setroot-assume-nonnull}). We write $-\infty$ and $\infty$ for the compile-time constants corresponding to the minimum and maximum values in the \code{val} type (e.g., \code{INT_MIN} and \code{INT_MAX}), respectively. In this procedure, she then sets the data value at the root node to the passed-in parameter \code{v} (\reftxt{line}{line:setroot-set}) and then wants to assert that the modified tree is still a binary search tree with the \code[Spec]{assert} at \reftxt{line}{line:setroot-assert}.
First, note that the \code[Spec]{assert} at \reftxt{line}{line:setroot-assert} is not provable statically because the search property may be violated after the assignment on \reftxt{line}{line:setroot-set}, as there are no constraints on the value of \code{v}.
Thus, it would be unsound for any static verifier to eliminate this \code[Spec]{assert}.
Second, observe that a dynamic validation of $\bstcall{\vart}{-\infty}{\infty}$ is a linear-time operation in the number of nodes of the binary tree \code{t} and furthermore has the non-trivial obligation to dynamically validate the separation constraints---even though the binary tree shape is unchanged and the only data value that was modified was at the root node.
\paragraph{Incremental Verification-Validation.}

Even if full static verification of a data structure check is not possible (as with the \code[Spec]{assert} on \reftxt{line}{line:setroot-assert} in the \code{setroot} procedure), surprisingly we can still use the incomplete static reasoning to incrementalize the check.
Specifically, our incremental verification-validation approach on an \code[Spec]{assert} consists of two phases. First, it tries to statically verify the assertion using invariants derived by a static shape analysis. If it is successful, no dynamic validation is needed and the \code[Spec]{assert} can be fully eliminated.
If it is partially successful in that the separation properties are statically proven, then it proceeds to synthesize a short-circuiting, incremental assertion to replace the original, full-structure assertion. This short-circuiting, incremental assertion is easy to express in C code because it does not need to check separation properties that were verified statically.
If the static verifier is not able to prove separation, our tool simply raises an alarm and does not proceed to synthesis.

Returning to the \code[Spec]{assert} at \reftxt{line}{line:setroot-assert} in the \code{setroot} procedure, we can reason about the binary search tree invariant to see that a change to the data value at the root node could only possibly invalidate the invariant along
the right spine of the left sub-tree \code{t->l} or the left spine of the right sub-tree \code{t->r} (i.e., these are the only nodes whose allowable range $[\symmin,\symmax)$ could have changed).
An \emph{incremental validation}---one that takes advantage of the knowledge
that \code{t} was a binary search tree before the update---needs to traverse at most these two paths
and thus can re-validate the invariant while traversing only a logarithmic number of nodes.
We diagram these two paths in \figref{validation:increxecution} with \textcolordull{squares} for the traversed nodes and with \textcolordull[1]{triangles} for the skipped sub-trees (ignore the caption and the number labels for the moment).
As we will see, lines~\ref{line:setroot-assert}--\ref{line:setroot-end} of \figref{validation:bstsc} indeed traverse at most the two paths discussed above by replacing the assertion to the $\bst$ predicate with calls to the synthesized short-circuiting invariant checker \code[highlighting]{|1:bstsc|} shown in \figref{bstsc:code}. The rest of this section will focus on describing how one would define \code[highlighting]{|1:bstsc|} to obtain this short-circuiting validation.

As alluded in \secref{introduction},
the basic idea behind a short-circuiting checker like \code[highlighting]{|1:bstsc|} is to make explicit a hypothesis of which it can take advantage.  This idea can expressed in the following correctness condition for \code[highlighting]{|1:bstsc|} that we call its safe replacement criterion:
\begin{criteria}[Safe Replacement with a Short-Circuiting Invariant Checker] 
For
$\bstsc$, it is the case that
\par\smallskip\centering\fbox{\(\begin{array}{@{}l@{\;}l@{}}
\text{if} & \bstcall{\symn}{\symhmin}{\symhmax} \\
\text{then} &
  \bstsccall{\symn}{\symmin}{\symmax}{\symhmin}{\symhmax} 
  \limp
  \bstcall{\symn}{\symmin}{\symmax}
\end{array}\)}
\end{criteria}
For clarity, we first consider deriving the short-circuiting checker $\bstsc$ in separation logic and then subsequently consider the compilation to C code.
Observe that the short-circuiting checker $\bstsc$
takes two additional parameters
$\symhmin$ and $\symhmax$
(standing for \fmthypp{h}ypothesis min and max, respectively) compared to the original checker $\bst$.
The \emph{validation hypothesis} is explicitly $\bstcall{\symn}{\symhmin}{\symhmax}$. This validation hypothesis may be used to return early (i.e., short-circuit checking) when validating the desired property $\bstcall{\symn}{\symmin}{\symmax}$.
In other words, a short-circuiting checker does not need to validate the desired property directly but can do so under the assumption
that the validation hypothesis holds.

A trivial implementation of $\bstsc$ that satisfies the safe replacement criterion is 
to ignore the validation hypothesis and be defined as $\bstcall{\symn}{\symmin}{\symmax}$---of course this offers no benefit over $\bst$ itself.
%
A slightly smarter variant, $\fmtidef{bstsc\_a}$ in \figref{bstsc:shallow}, explicitly assumes the
validation hypothesis
 on \reftxt{line}{line:bstsc-assume} and then
checks a \emph{short-circuiting condition} on \reftxt{line}{line:bstsca-short} (shown shaded): if the goal is exactly the hypothesis, then it can soundly short circuit and return $\ltrue$ early without further checking.  This short circuiting is, however, shallow---if the short-circuiting condition is not satisfied, the entire tree rooted at \code{n} is traversed by $\bst$.

\begin{figure*}[t]
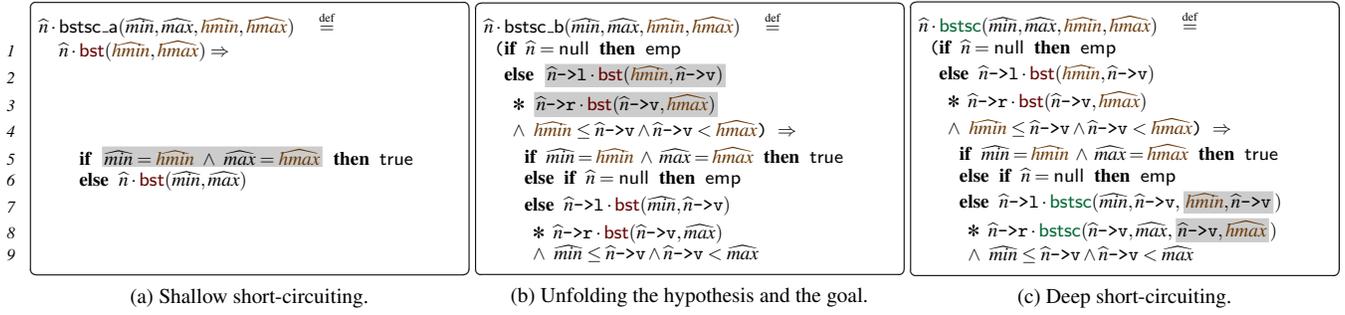
\centering
\rule{0.7em}{0pt}
\subfloat[Shallow short-circuiting.]{\label{fig:bstsc:shallow}\ovalbox{\scalebox{0.9}{\usebox{\SBoxBstScShallow}}}}
\hfill
\subfloat[Unfolding the hypothesis and the goal.]{\label{fig:bstsc:unfold}\ovalbox{\scalebox{0.9}{\usebox{\SBoxBstScUnfold}}}}
\hfill
\subfloat[Deep short-circuiting.]{\label{fig:bstsc:final}\ovalbox{\scalebox{0.9}{\usebox{\SBoxBstSc}}}}

\caption{The short-circuiting invariant checker
  $\bstsc$. It is synthesized
  by the following a proof of inclusion in separation logic.}
\label{fig:derivingbstsc}
\end{figure*}


To arrive at a deep short-circuiting, let us consider $\fmtidef{bstsc\_b}$ in \figref{bstsc:unfold}, which is semantically equivalent to $\fmtidef{bstsc\_a}$.  In $\fmtidef{bstsc\_b}$, we have inlined the calls to $\bst$.  In other words, logically, we have unfolded the inductive predicate given by $\bst$ in both the hypothesis and the goal.  By unfolding the hypothesis, we see that in the case that $\symn \neq \nullval$, we also have that
$\bstcall{\fldderef{\symn}{\fldl}}{\symhmin}{\fldderef{\symn}{\fldv}}$
and
$\bstcall{\fldderef{\symn}{\fldr}}{\fldderef{\symn}{\fldv}}{\symhmax}$.
We use these hypotheses (shown shaded) to derive the
deep short-circuiting checker $\bstsc$ by rewriting the recursive
$\bst$ checks on lines~\ref{line:bstscb-left}--\ref{line:bstscb-right} on the left and right sub-trees with recursive calls to the short-circuiting checker. The validation hypotheses of these recursive checks, $\icall{ \fldderef{\symn}{\fldl} }{\bstsc}{ \symmin, \fldderef{\symn}{\fldv}, \symhmin, \fldderef{\symn}{\fldv} }$ and $\icall{ \fldderef{\symn}{\fldr} }{\bstsc}{ \fldderef{\symn}{\fldv}, \symmax,\fldderef{\symn}{\fldv}, \symhmax }$, are satisfied by the hypotheses about $\fldderef{\symn}{\fldl}$ and $\fldderef{\symn}{\fldr}$, respectively.
This deep short-circuiting checker \code[highlighting]{|1:bstsc|} allows many
more opportunities to apply a short-circuiting hypothesis and return early (because the short-circuiting is checked on each recursive call), thus avoiding a full traversal of the tree.

Intuitively, the previous paragraph describes synthesizing the deep short-circuiting checker $\bstsc$ by following the proof of the safe replacement criterion:
\[\begin{array}{@{}l@{}}
\bstcall{\symn}{\symhmin}{\symhmax} \limp \\
\qquad \left(\bstsccall{\symn}{\symmin}{\symmax}{\symhmin}{\symhmax} \limp \bstcall{\symn}{\symmin}{\symmax}\right)
\end{array}\]
and abducting a definition for $\bstsc$.
In \secref{dynamic-validation}, we see that $\bstsc$ is synthesized by attempting to prove
\begin{equation}\label{eqn:sc-implies-imp}
\bstcall{\symn}{\symhmin}{\symhmax} \limp \bstcall{\symn}{\symmin}{\symmax}
\end{equation}
while permitting abduction of pure, data constraints.
We call this process \emph{subtraction-directed synthesis} because the proof strategy follows the standard approach of unfolding and subtraction to prove separation. When $\bstsc$ is successfully synthesized, we have that the spatial part of the above implication \eqref{eqn:sc-implies-imp} holds.
Thus, it is straightforward to compile the $\bstsc$ definition from \figref{bstsc:final} to C code using the \code{&&}-conjunction operator as shown in \figref{bstsc:code} because we have guaranteed that the separation properties hold statically (and thus do not need them to be validated dynamically).

\paragraph{Short-Circuiting Validation of \texttt{setroot}.}

We now return to our claim that the check on lines~\ref{line:setroot-assert}--\ref{line:setroot-end}
in the short-circuiting validation version of \code{setroot} shown in \figref{validation:bstsc} traverses at most two paths in tree \code{t}. This is illustrated in \figref{validation:increxecution}, which considers
an example where a call to \code{setroot} changes the data value at the root node from the value $8$ to the value $9$. 
The four-tuple labels on each node indicate the $(\varmin, \varmax, \varhmin, \varhmax)$ arguments passed to \code[highlighting]{|1:bstsc|} at that node from the initiating calls on lines~\ref{line:inc-invocation-left} and \ref{line:inc-invocation-right}.
This short-circuiting version
skips sub-trees that are not on the aforementioned two paths because their roots satisfy
the short-circuiting guard condition (i.e., that \code{min == hmin && max == hmax}).

Automatic synthesis of short-circuiting checkers like \code[highlighting]{|1:bstsc|}
is vastly preferable, for the developer, to manual
checker transformations. While it is perhaps possible for a clever developer to rewrite the validation routines themselves by hand to be short-circuiting, rewriting \code[Spec]{assert}s to make appropriate calls to these routines is challenging and error-prone because of the need to modify code to expose particular data values. 
An automated technique is needed to statically prove the validation hypothesis of the short-circuiting invariant checker (including the separation properties) and determine the appropriate arguments to it.
The goal of this paper is to automate this process of incrementalization---eliminating the tediousness and guaranteeing the safety of the transformation.

\paragraph{Soundness: Assumptions and Guarantees.}

The expected input for our approach is a program with \code[Spec]{assert}s of data structure invariants expressed in separation logic like $\bst$.
If modular verification-validation is desired, pre-conditions can be expressed additionally with \code[Spec]{assume}s.
Like with any modular verification technique, we assume that user-specified \code[Spec]{assume}s are sound pre-conditions (which could be guaranteed at each call site with corresponding \code[Spec]{assert}s).
The verification-validation problem that we address is to try to prove all \code[Spec]{assert}s.
If successful, the output of our technique is a static verification of the separation properties
(e.g., a proof that \code{t} is a binary tree)
and a rewriting of the \code[Spec]{assert}s to use short-circuiting dynamic validation
(e.g., an incremental run-time check that the binary tree \code{t} has the search property).
If unsuccessful, the output is an alarm for each assert where it was unable to statically verify the separation properties.
That is, like any other static verifier for separation logic, it soundly rejects programs where it cannot prove the separation properties.
As an alternative to rejecting such programs, one could consider dynamically validating such unproven separation properties using, for example, global heap coloring techniques~\cite{nguyen+2008:runtime-checking}.

\paragraph{Challenges: Synthesizing Short-Circuiting Validation.}

As we have seen, short-circuiting validation checkers can potentially reduce the asymptotic complexity of validation
while providing the same benefits as whole-structure checkers. Automatic generation of short-circuiting checkers
requires addressing two key challenges:
\begin{asparaenum}\itemsep 0pt
\item[(\textbf{Challenge 1})] \label{synthesizesccalls} Synthesis to transform whole-structure assertions (like $\bst$) into calls
to short-circuiting checkers (like \code[highlighting]{|1:bstsc|}). This process must generate both the additional
hypothesis arguments for such calls and
add the necessary scalar assertions. For example, the \code[Spec]{assert} on lines~\ref{line:setroot-assert}--\ref{line:setroot-end} in \figref{validation:bstsc} correspond to unfolding the original assertion $\bstcall{\vart}{-\infty}{\infty}$ once to account for the change at the root node.
\item[(\textbf{Challenge 2})] \label{synthesizesccode} Synthesis of deep short-circuiting invariant checkers (like \code[highlighting]{|1:bstsc|}) for general inductive data structure invariants (like $\bst$).
\item[(\textbf{Challenge 3})]
A technical challenge that arises to support the above is determining how to connect statically verified invariants with dynamic validation. For example, the new variable $\varoldv$ must be instrumented into the program on \reftxt{line}{line:setroot-set} so that it can be used in the \code[Spec]{assert} on lines~\ref{line:setroot-assert}--\ref{line:setroot-end} in \figref{validation:bstsc}.
\end{asparaenum}
\par The rest of this paper describes how we address these challenges with a synthesis approach based on inductive shape-data analysis.

\begin{figure*}%
\begin{tikzpicture}[grow=right,level distance=2.5cm,scale=1,every node/.style={scale=0.85},node distance=1.7cm]
\tikzstyle{level 1}=[level distance=2.1cm]
\tikzstyle{level 2}=[level distance=4.1cm]
\tikzstyle{level 3}=[level distance=3.8cm]
\tikzstyle{level 4}=[level distance=4cm]
\tikzstyle{level 5}=[level distance=3.5cm]
\node {} child {
  node [rect] {\parbox{4.5em}{\centering Shape-data analysis}} child {
    node [rect] {\parbox{6em}{\centering Logic-variable reification}} child {
      node [rect] (synnode) {\parbox{6.5em}{\centering Short-circuiting synthesis}} child {
        node [redrect] (exenode) {\parbox{4em}{\centering Execution}} child {
          node (endnode) {}
          edge from parent node[eabovemid] {
\begin{tabular}{@{}c@{}} 
      \mbox{\normalsize\em short-circuiting}
      \\
      \mbox{\normalsize\em condition}
      \\[1ex]
\code[C]{if (n==$\sym{n}$ && $\cdots$) return}
\end{tabular}
}
        }
        edge from parent node[eabovemid] {\code[Spec,alsolanguage=highlighting]{|1:bstsc|(t,n,m,$\sym{n}$,$\sym{m}$)}}
      }
      edge from parent node[eabovemid] {
      \begin{tabular}{@{}c@{}}
      \mbox{\normalsize\em programmatic}
      \\
      \mbox{\normalsize\em valuation}
      \\[1ex]
\code{$\aaddr_{\varp}$ = $\;\;\phi(\aaddr_\vart, \aaddr_1)$}
\end{tabular}
}
    }
    edge from parent node[eabovemid] {$\overbrace{
    \fldpt{ \aaddr_{\varp} }{ \fldl }{ \syml }
    \lsep\cdots
  }^{\makebox[0pt]{
    \(\begin{array}{@{}c@{}}
      \mbox{\normalsize\em validated view}
      \\[1ex]
      \mbox{\relsize{-1}$\icall{\aaddr_{\varp}}{\bst}{\ldots} \lreviwand
      (\icall{\syml}{\bst}{\ldots} 
      \lsep
      \icall{\symr}{\bst}{\ldots})$}
    \end{array}\)
  }}
  \lsep
  \cdots$}
  }
  edge from parent node[eabovemid] {$\bstcall{\vart}{\fmtvar{n}}{\fmtvar{m}}$}
};
\node (synalarm) [above right of=synnode] {};
\draw[->] (synnode) -- (synalarm) node[pos=0.3,above] {\parbox{1.5cm}{\colordull unproven separation alarm}};
\node (exealarm) [above right of=exenode] {};
\draw[->] (exenode) -- (exealarm) node[pos=0.3,above] {\parbox{1.2cm}{\colordull data assertion failure}};
\end{tikzpicture}
\caption{An overview of our approach to obtain short-circuiting data structure validation. The first three phases (blue boxes) are compile-time phases to synthesize code that runs in a normal execution environment (red box).}
\label{fig:overview-steps}
\end{figure*}
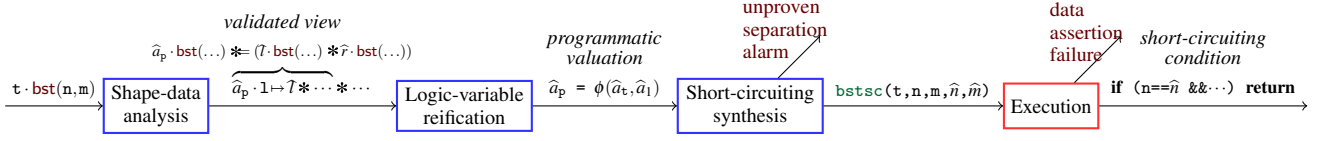

\paragraph{Overview.}

Our approach consists of three processes that we show schematically in \figref{overview-steps}:
\begin{asparaenum}
\item \textbf{Static shape-data analysis} infers loop invariants in separation logic.
It determines when previously established
data-value facts have been preserved---and can thus be
relied upon, at run time, during validation. A \emph{validated view} is a constraint we introduce to track arbitrary data constraints in materialized memory regions (\secref{validated-views}).

\item \textbf{Logic-variable reification}
transforms the program to
explicitly connect logic variables in the static analysis invariants
to concrete values at run time
by constructing a \emph{programmatic valuation} (\secref{logic-variable-instrumentation}).
This code instrumentation enables run-time assertion checking to rely upon
invariants discovered during static analysis.

\item \textbf{Subtraction-directed synthesis of short-circuiting validation} automatically
transforms expensive whole-data-structure validation \code[Spec]{assert}s into
cheaper, short-circuiting \code[Spec]{assert}s that can soundly skip portions of the data structure at run time (\secref{dynamic-validation}).
To synthesize the code for short-circuiting checkers,
we fix a template (i.e., a sketch~\cite{DBLP:conf/asplos/Solar-LezamaTBSS06}) similar to the shallow short-circuiting checker $\fmtidef{bstsc\_a}$ in \figref{bstsc:shallow} and then reuse our synthesizer
on the template code to fill it in
with deep short-circuiting, as in $\bstsc$ (\figref{bstsc:final}).
Short-circuiting synthesis can raise an alarm if the spatial part of the safe replacement criterion cannot be proven statically to soundly synthesize short-circuiting checkers like $\bstsc$.
\item \textbf{Assertion execution} with the rewritten \code[Spec]{assert}s in a normal execution environment realizes incremental invariant validation whenever the short-circuiting condition is triggered. If a data property is violated on execution, an assertion failure is raised.
\end{asparaenum}

\section{Reifying Static Analysis Invariants}
\label{sec:static-analysis}

The static shape-data analysis needed to synthesize short-circuiting validation code from statically-inferred shape-data invariants must 
address two problems.
First, it must have the ability to prove existential weakenings of shape-data invariants even when it cannot prove the original (e.g., prove $\exists \vec{\symu}.\;\icall{\aaddr}{\chker}{\vec{\symu}}$ when $\achkercall$ cannot be proved). While this is challenging in general, we identify a particularly important scenario for short-circuiting validation: tracking arbitrarily complicated data-value constraints in unmodified portions of the data structure. This precise tracking of unmodified portions can be seen as a static analogue of memoization in incremental computation.
We realize this precise tracking of unmodified portions by augmenting a shape-data analysis with what we call a \emph{validated view} abstraction (\secref{validated-views}). 

Second, it must provide access at run time to concrete values corresponding to existentially-quantified symbolic variables in statically-proven invariants. We provide this access by reading off, from the proof, a program expression derived from the pre-condition to witness each symbolic variable (\secref{logic-variable-instrumentation}).




\newcommand{\viewhi}[1]{\text{\relsize{1}\colorbox{black}{\color{white}#1:}}\;\;}
\newcommand{\symuhi}[1]{\text{\colordullx{6}\underline{$#1$}}}

\newsavebox{\SBoxBstExcisePreDM}
\begin{lrbox}{\SBoxBstExcisePreDM}
\begin{tikzpicture}[grow=right]
  \tikzstyle{level 1}=[levelchk]
  \tikzstyle{level 2}=[level distance=11mm, sibling distance=14mm]
  \tikzstyle{level 3}=[level distance=11mm, sibling distance=8mm]
  \tikzstyle{level 4}=[level distance=15mm]
      \node[nd] (n4) {$\aaddr_\vart$}
      child[levelseg] {
        node[nd] (p) {$\aaddr_{\varp}$} edge from parent[chk]
        child[fld] {
          node[nd] (c) {$\aaddr_1$}
          child[fld] {
            node[nd,label=right:{$= \nullval$}] {$\aaddr_2$}
          }
          child[fld] {
            node[nd] {$\aaddr_3$}
            child {
              coordinate edge from parent[chk]
              node[ebelowleft]{$\bst(\aval_1,\aval_{\varp})$}
            }
          }
        }
        child[fld] {
          node[nd] {$\aaddr_4$}
          child[level distance=15mm] {
            coordinate edge from parent[chk]
            node[ebelowleft]{$\bst(\aval_{\varp},\symu_2)$}
          }
        }
        node[ebelowleft]{$\bst(\symvninf,\symvpinf)$}
        node[ebelowright]{$\bst(\symu_1,\symu_2)$}
      }
    ;

    \node[below=6.5mm,left=6.5mm] at (p) {$\spadesuit$};
    \node[nbelow] at (n4) {$\vart$};
    \node[nabove,above=7mm] at (p) {$\viewhi{view at $\aaddr_{\varp}$} \bstcall{\aaddr_{\varp}}{\symu_1}{\symu_2} \lreviwand (\bstcall{\aaddr_1}{\symu_1}{\aval_{\varp}} \lsep \bstcall{\aaddr_4}{\aval_{\varp}}{\symu_2})$};
    \node[nbelow] at (p) {$\varp$};
    \node[nabove,below=9mm] at (c) {$\viewhi{view at $\aaddr_1$} \bstcall{\aaddr_1}{\symu_1}{\aval_{\varp}} \lreviwand (\bstcall{\aaddr_2}{\symu_1}{\aval_1} \lsep \bstcall{\aaddr_3}{\aval_1}{\aval_{\varp}})$};
\end{tikzpicture}
\end{lrbox}

\newsavebox{\SBoxBstExcisePostDM}
\begin{lrbox}{\SBoxBstExcisePostDM}
\begin{tikzpicture}[grow=right]
  \tikzstyle{level 1}=[levelchk]
  \tikzstyle{level 2}=[level distance=11mm, sibling distance=14mm]
  \tikzstyle{level 3}=[level distance=11mm, sibling distance=8mm]
  \tikzstyle{level 4}=[level distance=15mm]
      \node[nd] (nt) {$\aaddr_{\vart}$}
      child[levelseg] {
        node[nd] (p) {$\aaddr_{\varp}$} edge from parent[chk]
        child[fld] {
            node[nd] {$\aaddr_3$}
            child[level distance=15mm] {
              coordinate edge from parent[chk]
              node[ebelowleft]{$\bst(\aval_1,\aval_{\varp})$}
            }
        }
        child[fld] {
          node[nd] {$\aaddr_4$}
          child[level distance=15mm] {
            coordinate edge from parent[chk]
            node[ebelowleft]{$\bst(\aval_{\varp},\symu_2)$}
          }
        }
        node[ebelowleft]{$\bst(\symvninf,\symvpinf)$}
        node[ebelowright]{$\bst(\symu_1,\symu_2)$}
      }
    ;

    \node[nbelow] at (nt) {$\vart$};
    \node[nbelow] at (p) {$\varp$};
    
    \node[below=7.5mm,right=2.5mm] (v1) at (p) {$\aval_1$};
    \node[below=1mm] at (v1) {$\texttt{vmin}$};
\end{tikzpicture}
\end{lrbox}


\setlength{\saveboxsep}{\fboxsep}
\fboxsep 1pt
\newsavebox{\SBoxCodeExampleDMIS}
\begin{lrbox}{\SBoxCodeExampleDMIS}\small\lstnonum
\begin{lstlisting}[language=C,alsolanguage=Spec,alsolanguage=highlighting,alsolanguage=SqueezeOps,style=number]
val excisemin(bt t) {#\lstbeginn#
 assume(#$\bstcall{\vart}{-\infty}{\infty}$#);#\label{line-assumebst}##\hibox{\text{\ttfamily$\aaddr_{\vart}$ = t,\,$\symvninf$ = $-\infty$,\,$\symvpinf$ = $\infty$;}}#
 assume(#$\vart\mathord{\neq}\nullval \land \fldderef{\vart}{\fldl}\mathord{\neq}\nullval$#);#\label{line-assumetough}##\hibox{\text{\ttfamily$\aaddr_{\vart\fldl}$ = $\aaddr_{\vart}$->l,\,$\aaddr_{\vart\fldr}$ = $\aaddr_{\vart}$->r,\,$\aval_{\vart}$ = $\aaddr_{\vart}$->v;}}#
 bt p; for (p = t; true; p = p->l) {                         ###\label{line-excisemin-loop}#
  #\hibox{\text{\ttfamily$\aaddr_{\varp}$ = $\phi(\aaddr_{\vart}, \aaddr_1)$, $\aval_{\varp}$ = $\phi(\aval_{\vart}, \aval_{1})$, $\aaddr_{1}$ = $\phi(\aaddr_{\vart\fldl}, \aaddr_{2})$, $\aaddr_{4}$ = $\phi(\aaddr_{\vart\fldr}, \aaddr_{3})$,}}\label{line-excisemin-join}#
  #\hibox{\text{\ttfamily$\symu_{1}$ = $\phi(\symvninf, \symu_{1})$, $\symu_{2}$ = $\phi(\symvpinf, \aval_{\varp})$;}}\label{line:inst-mapping}#                  
  if (p->l->l == NULL) break;#\label{line-excisemin-loopguard}#
  #\hibox{\text{\ttfamily$\aaddr_2$ = $\aaddr_1$->l, $\aaddr_3$ = $\aaddr_1$->r, $\aval_{1}$ = $\aaddr_1$->v;}}\label{line:instunfoldguard}#
 }
#\InvBox{\scalebox{0.85}{\usebox{\SBoxBstExcisePreDM}}}\label{line:post-findmin}#
 val vmin = p->l->v; p->l = p->l->r;#\label{line-excisemin-excise}#
#\InvBox{\scalebox{0.85}{\usebox{\SBoxBstExcisePostDM}}}\label{line:post-swap}#
 assert(#$\bstcall{\vart}{-\infty}{\infty}$#);#\label{line:origassert}# return vmin;#\lststopn#
}
\end{lstlisting}
\end{lrbox}
\setlength{\fboxsep}{\saveboxsep}

%


In this section, we extend our running example to both illustrate the challenges introduced in \secref{shortcircuit} and provide intuitions for how our analysis and synthesis techniques address them.
We will focus on an example-driven discussion here, and we separately formalize the algorithms for validated views and logic-variable reification in \IfTR{in \appref{validated-view-formal} and \appref{instrumentation-formal}}{Appendix~\ref{TR-app:validated-view-formal} and~\ref{TR-app:instrumentation-formal} of the supplementary material}, respectively.


\subsection{Validated Views}
\label{sec:validated-views}

To describe how our validated view abstraction enriches a shape-data abstraction, we first discuss an example shape-data abstraction.

\PUNCH{We focus on how the abstraction enables
  mitigating inductive imprecision.  The domain itself is
  straightforward.}
  
\paragraph{Preliminaries: Inductive Shape-Data Analysis.}

Inductive shape
analysis\citeverbelse{distefano+2006:local-shape,magill+2006:inferring-invariants,chang+2007:shape-analysis}{distefano+2006:local-shape,chang+2008:relational-inductive}
uses inductive predicates like $\bst$ to statically summarize
unbounded memory regions. A memory region that satisfies $\bst$
includes the pointer-shape property (i.e., is a binary tree) but also
the data-value property (i.e., that it satisfies the search
invariant). Thus, such a summarized $\bst$ memory region
is statically known to satisfy the binary search tree invariant.
On the left side of \figref{bst-unfoldfold}, we
express a shape-data invariant as a separating shape
graph.
This invariant says a
root $\symt$ satisfies $\bst$ with minimum value $\symmin$ and maximum
value $\symmax$. Here the nodes represent
pointers (i.e., memory addresses), and the edges represent abstract memory
regions.  The thick arrow can be read as a static fact indicating
that the memory region rooted at $\symt$ satisfies the inductive predicate $\bstcall{\symt}{\symmin}{\symmax}$.


In inductive shape analysis, the key operations are (1) materializing
abstractions of single memory cells from
summarized regions by unfolding inductive predicates (left-to-right in
\figref{bst-unfoldfold}) and (2) summarizing memory regions by folding
into inductive predicates, using any number of unary abstraction or
binary widening operators (right-to-left).
\SQUEEZE{%
To visually distinguish
non-pointer values in our graphs, we draw them as nodes without the
circle (e.g., the contents of $\fldv$ field from pointer $\symt$ is
value $\symv$ drawn without a circle).
}%
Consider statically analyzing an iteration traversing into the middle
of a binary search tree with a cursor pointer $\varp$: in
\figref{bst-traverse}, we show a precisely inferred loop invariant that is crucial to be able to synthesize short-circuiting validation.
\SQUEEZE{%
We show the loop invariant at the program location where the fields of the node pointed-to by
$\varp$ have been materialized but  $\varp$ has yet to be advanced.
}%
Here, we indicate that the program variable $\varp$ contains the
address $\aaddr_{\varp}$ by annotating the program variable below the
node representing $\aaddr_{\varp}$.  We adopt the naming convention
that $\aaddr_{\var}$ is the symbolic address of the binary search tree
node pointed-to by program variable $\var$, and $\aval_{\var}$ is the
symbolic data value of that node (i.e., $\var$\texttt{->v}).  The
memory region between addresses $\symt$ and $\aaddr_{\varp}$ is
described by a $\bst$ segment~\cite{chang+2008:relational-inductive},
a form of separating implication, corresponding to a
validation tree with a hole~\cite{DBLP:conf/popl/Minamide98} at $\aaddr_{\varp}$.
The thick arrow between nodes $\symt$ and $\aaddr_{\varp}$ is read
as knowing statically that the memory region from $\symt$ satisfies
the validation check $\bstcall{\symt}{\symmin}{\symmax}$ up to
checking $\bstcall{\aaddr_{\varp}}{\symumin}{\symumax}$. The pair of
symbolic values $\symumin$ and $\symumax$ correspond to the lower and
upper bounds needed for a binary search tree rooted at
$\aaddr_{\varp}$ in order for the tree rooted at $\symt$ to satisfy
the validation check $\bstcall{\symt}{\symmin}{\symmax}$.
This $\bst$
segment summarizes the path that pointer $\varp$ has already
traversed and is derived by folding cells materialized
on the previous iteration.
Shape analysis tools vary in how
they represent segments---but some mechanism to do so is fundamental.

\newsavebox{\SBoxBstChecker}
\begin{lrbox}{\SBoxBstChecker}
\begin{tikzpicture}[grow=right]
  \node[nd] {$\symt$}
    child[levelchk] {
      coordinate edge from parent[chk]
        node[ebelowleft]{$\bst(\symmin,\symmax)$}
    };
\end{tikzpicture}
\end{lrbox}
\newsavebox{\SBoxBstTraverse}
\begin{lrbox}{\SBoxBstTraverse}
\begin{tikzpicture}[grow=right]
  \tikzstyle{level 2}=[levelfld, sibling distance=6mm]
  \tikzstyle{level 3}=[levelchk]
  \node[nd] {$\symt$}
    child[level distance=47mm] {
      node[nd] (c) {$\aaddr_{\varp}$} edge from parent[chk]
      child[fld] {
        node[nd] {}
        child {
          coordinate edge from parent[chk]
          node[ebelowleft]{$\bst(\symumin',\aval_{\varp})$}
        }
        edge from parent node[ebelowmid] {$\fldl$}
      }
      child[leveldata] {
        node {$\aval_{\varp}$}
        edge from parent node[fill=white] {$\fldv$}
      }
      child[fld] {
        node[nd] {}
        child {
          coordinate edge from parent[chk]
          node[ebelowleft]{$\bst(\aval_{\varp},\symumax')$}
        }
        edge from parent node[eabovemid] {$\fldr$}
      }
      node[ebelowleft]{$\bst(\symmin,\symmax)$}
      node[ebelowright]{$\bst(\symumin,\symumax)$}
    };

    \node[nbelow] at (c) {$\varp$};
\end{tikzpicture}
\end{lrbox}
\newsavebox{\SBoxBstUnfold}
\begin{lrbox}{\SBoxBstUnfold}
\begin{tikzpicture}[grow=right]
  \tikzstyle{level 1}=[levelfld, sibling distance=5mm]
  \tikzstyle{level 2}=[levelchk]
  \node[nd] {$\symt$}
    child {
      node[nd] {$\syml$}
      child {
        coordinate edge from parent[chk]
          node[ebelowleft]{$\bst(\symmin,\symv)$}
      }
      edge from parent node[ebelowmid] {$\fldl$}
    }
    child[leveldata] {
      node {$\symv$}
      edge from parent node[fill=white] {$\fldv$}
    }
    child {
      node[nd] {$\symr$}
      child {
        coordinate edge from parent[chk]
          node[ebelowleft]{$\bst(\symv,\symmax)$}
      }
      edge from parent node[eabovemid] {$\fldr$}
    };
\end{tikzpicture}
\end{lrbox}
\begin{figure}\centering\small
%

\subfloat[Unfolding and folding rules for binary search tree summaries
derived from the inductive definition $\bst$.]{\label{fig:bst-unfoldfold}
\scalebox{0.83}{
\(\begin{array}{ccc}
    \begin{array}[m]{@{}c@{}}\fbox{\usebox{\SBoxBstChecker}}\end{array}
    &
    \begin{array}[m]{@{}l@{}}
      \unfold
      \\[0.5ex]
      \fold
    \end{array} 
	&
    \fbox{\begin{tikzpicture}
      \node[nd] {$\symt$};
    \end{tikzpicture}
    \quad
    \raisebox{0.85ex}{$\symt = \nullval$}}
    \\
    & &
    \raisebox{3ex}{$\bigvee$} \qquad
    \fbox{$\begin{array}{@{}c@{}}
      \usebox{\SBoxBstUnfold}
      \\
      \symt \neq \nullval
      \spland \hibox{ \symmin \leq \symv \spland \symv < \symmax }
    \end{array}$}
\end{array}\)
}
}

\subfloat[A precise loop invariant for a traversal of a binary search tree using a cursor pointer \texttt{p} initially summarized by $\bstcall{\symt}{\symmin}{\symmax}$.]{\label{fig:bst-traverse}
\scalebox{0.83}{
  \(\begin{array}{@{}c@{}}
    \usebox{\SBoxBstTraverse} \\
    \aaddr_{\varp} \neq \nullval
    \spland \symmin \leq \hibox{
      \colordullul[3]{\vphantom{\leq_g}\symumin = \symumin'}
      \leq \symv_{\varp} <
      \colordullul[3]{\vphantom{\leq_g} \symumax' = \symumax } } <
      \symmax
\end{array}\)
}
}
\caption{The intertwined shape-data binary search tree invariant.}
\label{fig:example-checker}
\end{figure}
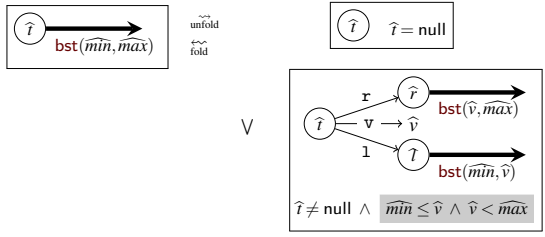
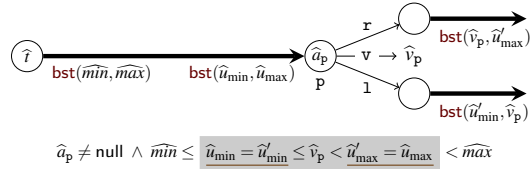

The key observation to make in the rules shown
in \figref{bst-unfoldfold} is that folding (i.e., going
right-to-left) requires verifying a pure, non-memory, data constraint
(shaded) just to summarize a memory region into
a summary predicate like $\bst$ or $\bst$-segment.
In the idealized loop invariant shown in \figref{bst-traverse},
the total ordering constraints on symbolic data values (shaded) must
be tracked by the base data-value domain in order to derive that the
whole memory region from $\symt$ still satisfies
$\bstcall{\symt}{\symmin}{\symmax}$. The two middle constraints
$\symumin' \leq \aval_{\varp} < \symumax'$ are not difficult to derive (come from unfolding
at the current binary search tree node $\aaddr_{\varp}$), but the two outer
equality constraints (underlined) are crucially important to connect
the regions before and after $\aaddr_{\varp}$.
If any of these data-value constraints are lost across joins or widens from interactions with the shape domain, the inferred loop invariant will be too imprecise to witness all symbolic variables and thus to synthesize short-circuiting validation.  Imagine the likelihood that some such imprecision creeps in with arbitrarily complex data-value constraints!

\paragraph{Uninterpreted Shape-Data Predicates.}

The goal of the validated view abstract domain is to compensate for such lost precision by tracking
uninterpreted shape-data predicates.  While it will be uninteresting by itself, when it is combined in a reduced product~\cite{cousot+1979:systematic-design} with a shape-data abstraction, our validated view abstract domain
enables two key operations:
\begin{inparaenum}[(1)]
  \item folding of non-invalidated memory regions regardless of any
    imprecision (in the base data-value domain) when tracking data-value constraints and
    \label{mainpt-views-folding}
  \item precise tracking of equalities that connect data-value
    parameters across memory regions.
    \label{mainpt-views-equalities}
\end{inparaenum}
To illustrate what a validated view provides, consider in more detail the \code{excisemin} function shown in \figref{excisemin}.

The \code{excisemin} binary search tree operation excises the node in the tree 
with the minimum data value.
%
Ignore the shaded lines of code for the moment. 
The first \code[Spec]{assume} (on \reftxt{line}{line-assumebst}) says that the input tree \code{t} is a binary search tree, while the second \code[Spec]{assume} (on \reftxt{line}{line-assumetough})
requires that we are in the case where the root \code{t} and the immediate left child \code{t->l} are non-null. For presentation, we focus on this case, as it is corresponds to when \code{excisemin} loops and thus requires loop invariant inference.  The \code[C]{for} loop on \reftxt{line}{line-excisemin-loop} walks down the left spine of the binary search tree until \code{p} points to the parent node of the minimum data value node (i.e., the leftmost leaf) as terminated by the guard on \reftxt{line}{line-excisemin-loopguard}.  We write the loop in this somewhat non-idiomatic way for presentation purposes to explicitly expose the join point of the loop on \reftxt{line}{line-excisemin-join}; a more standard \code[C]{while$\;$(p->l->l != NULL)} loop is equivalent and does not affect the program analysis.
 On \reftxt{line}{line-excisemin-excise}, the programmer saves the data value of the minimum node in \code{vmin} and excises the node. Finally, the code checks that tree \code{t} is still a binary search tree with the \code[Spec]{assert} on \reftxt{line}{line:origassert}.
Our overall goal is to 
replace this whole-data-structure assertion at the end of the \code{excisemin} operation with a synthesized short-circuiting validation check.
Together with \code{setroot} from \secref{shortcircuit}, we have the helper operations for deleting a node from a binary search tree.

\begin{figure}
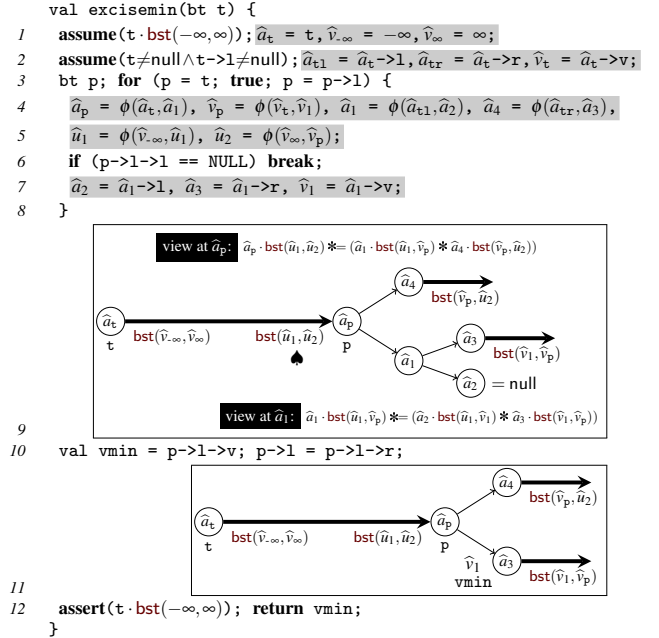

\rule{1.5em}{0pt}\scalebox{0.88}{\usebox{\SBoxCodeExampleDMIS}}
\caption{Inductive shape analysis with validated views and an instrumented programmatic valuation on \texttt{excisemin}, which excises the node with the minimum data value.}
\label{fig:excisemin}
\end{figure}

Now consider the static invariant shown on \reftxt{line}{line:post-findmin} right after the minimum-finding loop. To obtain this invariant, the shape analysis must infer a loop invariant at the head of the \code[C]{for} loop on \reftxt{line}{line-excisemin-loop} similar to the invariant discussed previously in \figref{bst-traverse}. The shape graph on \reftxt{line}{line:post-findmin} is the same, except that the exit condition \code{p->l->l == NULL} is reflected.

The additional annotations in the invariant on \reftxt{line}{line:post-findmin} correspond to validated view predicates that have form $\achkreg \lreviwand \amultichkreg$ where $\achkreg \bnfdef \achkercall$ is an instance of an inductive predicate $\chker$ and
$\amultichkreg  \bnfdef 
    \amultichkreg_1 \lsep \amultichkreg_2
    \bnfalt \lemp
    \bnfalt \achkreg$
is a memory region composed of disjoint instances of inductive predicates.  The semantics of this predicate $\achkreg \lreviwand \amultichkreg$ is a memory region that satisfies $\achkreg$ up to endpoints given by $\amultichkreg$. The predicate $\achkreg \lreviwand \amultichkreg$ implies the separating implication $\achkreg \lrevwand \amultichkreg$ meaning that $((\achkreg \lreviwand \amultichkreg) \lsep \amultichkreg \limp \achkreg)$.
An important point is that a validated view concretizes to a concrete memory, just like the shape graph. The interpretation of the product of a shape graph and a validated view is that the validated view constrains a \emph{sub-store} described by the shape graph. These details are described further in our formalization (\IfTR{\appref{validated-view-formal}}{\refappendix{validated-view-formal}}).

The reason to use validated views is that on an unfolding, the data constraints specified in the inductive predicate are always satisfied on the materialized cells corresponding to the unfolded node---at least until a modification to these materialized cells.  We can see a validated view as an uninterpreted shape-data predicate that is applied to a materialized region until the analysis interprets an update that could invalidate the predicate.  The view at $\aaddr_{\varp}$ shown on \reftxt{line}{line:post-findmin} in \figref{excisemin}
$$\bstcall{\aaddr_{\varp}}{\symu_1}{\symu_2} \lreviwand (\bstcall{\aaddr_1}{\symu_1}{\aval_{\varp}} \lsep \bstcall{\aaddr_4}{\aval_{\varp}}{\symu_2})$$
governs the materialized points-to edges from $\aaddr_{\varp}$ up to $\aaddr_1$ and $\aaddr_4$ (that is, the fields corresponding to \code{p->l}, \code{p->r}, and \code{p->v}).
This view predicate says that the materialized region consisting of these three cells satisfies whatever data constraint is specified in the ``inductive step'' of the inductive predicate $\bst$: in this case, that $\symu_1 \leq \aval_{\varp} < \symu_2$.
By tracking the data-value constraints of a materialized region in this uninterpreted manner, we are able to infer when the shape-data invariants are preserved regardless of the complexity of the data-value constraints.


Hypothetically, consider the worst case for data reasoning where we have no ability to capture any data-value constraints (e.g., there is no part of our abstraction that captures the inequality constraints of $\bst$) or more realistically, that important data-value constraints were lost in widening.
Even in this case, the validated view domain enables our analysis to capture the key facts needed for folding:
(1) that the $\bst$ allowable range parameters for $\aaddr_{\varp}$ are indeed $\symu_1, \symu_2$ (and thus they correspond to the allowable range at the segment endpoint $\aaddr_{\varp}$); and
(2) that the data value at $\aaddr_{\varp}$ (that is, $\aval_{\varp}$) is in this allowable range $(\symu_1, \symu_2]$, as is required to fold into the $\bst$ predicate.  Without fact~(1) maintained by validated views, the allowable range at the segment hole (marked with $\spadesuit$ in \figref{excisemin}) could be inferred as some other allowable range $\symu_1', \symu_2'$.  While this invariant is sound, it is too weak to imply that the entire memory region reachable from \code{t} is a binary search tree (i.e., $\bstcall{\aaddr_{\vart}}{\symvninf}{\symvpinf}$).

Validated views can be seen as a combination of ideas drawing from
separation logic~\cite{reynolds2002:separation-logic:} and predicate
abstraction~\cite{DBLP:conf/cav/GrafS97}, where a set of
inductive-predicate labels further constrains the heap.
The validated view domain itself can be seen as a predicate abstraction domain and thus the
abstract domain operators for validated views are straightforward; for completeness, we give a formal description \refappendix{validated-view-formal}.



\SQUEEZE{
The validated view predicate is in fact a multi-endpoint generalization of an inductive segment.
Essentially, the validated view abstract domain
provides arbitrary views of the same memory in terms of the set of
inductive checkers.
}

\subsection{Logic-Variable Reification}
\label{sec:logic-variable-instrumentation}

At run time, a short-circuiting validation checker must have access to the concrete values corresponding to the symbolic logic variables in the statically proven invariants.  We provide access to
these values by using the structure of the static analysis proof to
automatically construct program expressions that correspond to logic variables.
The shaded lines in \figref{excisemin} are the instrumentation for such expressions.
%
%
%
We add a new program variable for each symbolic variable
in the analysis state and initialize these variables with expressions
from what we call the \emph{programmatic valuation}. 
For presentation purposes, we conflate symbolic variable
names with program variables---so, for example,
the assignment {\ttfamily$\aaddr_{\vart}$ = t} at \reftxt{line}{line-assumebst}
is really assigning to a fresh program
variable whose purpose
is to hold the symbolic value $\aaddr_{\vart}$. 

The programmatic valuation is generated by following transfer functions of the static analysis, specifically for \code[Spec]{assume}, assignment, and join.  Here, we describe how the programmatic valuation is computed by following the \code{excisemin} example.  We give details in a more formal manner \refappendix{instrumentation-formal}.

The key challenge we address is that the static analysis must be precise enough to derive a witness to any newly introduced symbolic variable as a program expression from the existing ones.  This requirement is more demanding than that for a related instrumentation in \citet{magill+2010:automatic-numeric}, where the symbolic variables can be ``witnessed'' by a non-deterministic value (because the output is consumed by another static analysis as opposed to dynamic validation in our case).
In \code{excisemin} on \reftxt{line}{line-assumebst},
the assignments to $\aaddr_{\vart}$, $\symvninf$, and $\symvpinf$ come from the interpretation of the
\code[Spec]{assume} to create the static analysis state. 
The subsequent assignments to $\aaddr_{\vart\fldl}$, $\aaddr_{\vart\fldr}$, and $\aval_{\vart}$ on
\reftxt{line}{line-assumetough} come from unfolding a $\bst$ inductive predicate at $\aaddr_{\vart}$
(stored in \code{t}). 
\SQUEEZE{%
Similarly, the assignments on \reftxt{line}{line:instunfoldguard} come from unfolding at
$\aaddr_1$ (\code{p->l}) in the previous line.
}%
Following the proof to find witnesses to newly introduced symbolic variables for these commands is relatively straightforward.

The complicated case for obtaining program expression witnesses is at join points, such as at \reftxt{line}{line-excisemin-join}.
\SQUEEZE{This join point is at the loop head of the \code[C]{for} loop.}%
We write the assignments here using standard SSA
$\phi$ notation, where the first component is the value from the loop entry edge, and the second
component is the value from the loop back edge.  On \reftxt{line}{line-excisemin-join}, the $\phi$-assignment to $\aaddr_{\varp}$ shows that $\aaddr_{\varp}$ (the value of the
pointer \code{p}) is $\aaddr_{\vart}$ on entry and $\aaddr_1$ on the loop back edge.  This $\phi$-assignment to $\aaddr_{\varp}$ is because
\code{p} is advanced via \code{p = p->l} and can be witnessed, along with the others on this line, by following the join of the shape graphs.
\SQUEEZE{%
The assignments on \reftxt{line}{line-excisemin-join} follow
from similar reasoning.
%
}%

The $\phi$-assignments \code{$\symu_{1}$ = $\;\phi(\symvninf, \symu_1)$, $\symu_2$ = $\;\phi(\symvpinf, \aval_{\varp})$} on \reftxt{line}{line:inst-mapping} (also at the same join point) crucially come from following the join with the validated view constraint.
%
The first component of the two assignments $\symvninf$ and 
$\symvpinf$ are derived from the meaning of an empty segment between $\aaddr_{\vart}$ and
$\aaddr_{\varp}$, but the second comes from the validated view constraint
at $\aaddr_1$ (stored in \code{p->l}).  This view constraint 
$\bstcall{\aaddr_1}{\symu_1}{\aval_{\varp}} \lreviwand \cdots$ (shown on \reftxt{line}{line:post-findmin}) 
allows the analysis to derive that $\symu_1$ and $\symu_2$ (the allowable range at $\aaddr_{\varp}$) are 
equal to $\symu_1$ and $\aval_{\varp}$ (the allowable range at $\aaddr_1$ in the previous loop iteration), respectively.
Without this view constraint at $\aaddr_1$, the inferred static invariant---while sound---is not precise enough to find program expressions that witness $\symu_1$ and $\symu_2$ with concrete values at run-time.

Intuitively, these $\phi$-assignments come from computing the join of shape invariants. As symbolic variables correspond to existential variables, the join computes mappings from the symbolic variables of the resulting shape invariant to the variables of the input invariants. These mappings witness the existential variables in the resulting invariant in terms of variables of the input invariants and thus are reified as these $\phi$-assignments.



Generating programmatic valuations enables an instrumentation
that creates a ``run-time shadow'' of the static analysis state.
Crucially, even though this shadow refers to the heap, it
tracks all symbolic variables in stack locations---the instrumentation
does not add any storage on the heap. This approach allows
our short-circuiting validation to effectively incrementalize
checking of data structure invariants with a constant-space overhead rather than the linear-space overhead of a shadow heap with dynamic memoization.


\section{Synthesizing Short-Circuiting Validation}
\label{sec:dynamic-validation}

\newcommand{\shortify}{\ensuremath{\text{\normalfont\bfseries shortify}}}

The final step of our approach is to synthesize short-circuiting checks that
incrementally validate data structure invariants.  In this section, we describe a proof-directed synthesis algorithm that generates short-circuiting validation checks. 
There are two parts to
this process:
\begin{inparaenum}[(1)]
\item \label{syncallpart} synthesizing calls to the short-circuiting validation checkers
at assertion sites (such as the call to \code[highlighting]{|1:bstsc|} in \figref{validation:bstsc}
in \secref{shortcircuit});
and
\item \label{syncheckerpart} synthesizing the code for the short-circuiting validation checkers themselves (such as the code for \code[highlighting]{|1:bstsc|} itself).
\end{inparaenum}

\newsavebox{\SBoxPreShortify}
\begin{lrbox}{\SBoxPreShortify}\small
\begin{lstlisting}[language=Spec,alsolanguage=highlighting,alsolanguage=SqueezeOps,style=number,firstnumber=11]
$\astore\colon$#\lststopn#
$\scalebox{1}{\usebox{\SBoxBstExcisePostDM}}$#\lststartn#
assert(#$\bstcall{\vart}{-\infty}{\infty}$#);
\end{lstlisting}
\end{lrbox}

\newsavebox{\SBoxPostShortify}
\begin{lrbox}{\SBoxPostShortify}\small
\begin{lstlisting}[language=Spec,alsolanguage=highlighting,alsolanguage=SqueezeOps]
$\symu_1'$ = $\symu_1$; $\symu_2'$ = $\symu_2$; assert(      ##
 |1:bstsegsc|(t, $\aaddr_{\varp}$, #$-\infty$#, #$\infty$#, $\symvninf$, $\symvpinf$,
          &$\symu_1'$, &$\symu_2'$, $\symu_1$, $\symu_2$) &&
 $\symu_1'$ <= $\aaddr_{\varp}$->v && $\aaddr_{\varp}$->v < $\symu_2'$ &&
 |1:bstsc|($\aaddr_3$, $\symu_1'$, $\aaddr_{\varp}$->v, $\aval_1$, $\aval_{\varp}$) &&
 |1:bstsc|($\aaddr_4$, $\aaddr_{\varp}$->v, $\symu_2'$, $\aval_{\varp}$, $\symu_2$)
);
\end{lstlisting}
\end{lrbox}

\paragraph{Synthesis Overview.}
Our approach employs a static verification to drive
synthesis.
\figref{preshortify} shows the \code[Spec]{assert} at the end of the \code{excisemin} procedure from
\figref{excisemin} and the inferred shape-data invariant (now labeled $\astore$) at that program point.
From a static verification perspective, the \code[Spec]{assert} is a request to prove 
that
the  fact $\bstcall{\aaddr_t}{\symvninf}{\symvpinf}$ holds in a sub-heap at that program point, where $\aaddr_t$ is the symbolic
value stored in program variable \code{t}.
%
%
%
In other words, the analysis is asked to prove the abstract inclusion $\astore \sqsubseteq \bstcall{\aaddr_t}{\symvninf}{\symvpinf} \lsep \ltrue$.  In separation-logic--based shape analysis, 
abstract inclusion is verified by
unfolding and subtraction\citeverbelse{berdine+2005:symbolic-execution,
  distefano+2006:local-shape,chang+2007:shape-analysis,rival+2011:calling-context}{berdine+2005:symbolic-execution,distefano+2006:local-shape,rival+2011:calling-context}.
Subtraction-based inclusion checking $\astore_1 \sqsubseteq \astore_2$ works by unfolding the righthand-side $\astore_2$ to match the shape structure (i.e., spatial formulas) of the lefthand-side $\astore_1$ and ``subtracting'' matching shapes until the inclusion is trivial.
(A base data-value domain, decision procedure, or SMT-solver can then be used to try to discharge any pure, data-value constraints for inclusion.)

\subparagraph{\reftxt{Part}{syncallpart}) Synthesizing Calls to Short-Circuiting Checkers.}
We adapt the standard unfolding-subtraction inclusion algorithm
to synthesize short-circuiting validation calls.
%
We describe this algorithm subsequently in \secref{shortify} and \secref{segmentsyn-appendix}.
For the moment, consider \figref{postshortify}, where
we present the calls to short-circuiting checkers that
our approach synthesizes to replace the \code[Spec]{assert} from \figref{preshortify}.

Intuitively, our approach is to synthesize this code
that corresponds to an unfolding of the assertion
$\bstcall{\vart}{-\infty}{\infty}$
to match the shape structure of the inferred static invariant $\astore$.
Here we write \code[highlighting]{|1:bstsegsc|} and \code[highlighting]{|1:bstsc|} as calls to short-circuiting validation checkers
corresponding to full, non--short-circuiting validation checkers for a binary search tree segment
and a complete binary search tree \code[highlighting]{|bst|}, respectively. We synthesize calls to short-circuiting checkers at assertion sites
by (a) unfolding the assertion to match the inferred static invariant and
(b) replacing calls to
non-short-circuiting checkers with calls to their short-circuiting variants.


\subparagraph{\reftxt{Part}{syncheckerpart}) Synthesizing Short-Circuiting Checkers.}
We can also use 
this subtraction-directed synthesis
to generate the deep short-circuiting validation
checkers themselves.
%
%
The key observation is that we can represent the safe replacement criterion
from \secref{shortcircuit} in the form of a generic code template---which we show in 
\figref{sctemplate} for an arbitrary non--short-circuiting checker
$\chker$ (here, type \texttt{ds} is a generic placeholder for the data structure type in question).
We can then apply the same subtraction-directed synthesis approach that
we use in \reftxt{Part}{syncallpart} to generate deep short-circuiting.
Note the similarity of the template to the shallow short-circuiting checker $\fmtidef{bstsc\_a}$
from \figref{bstsc:shallow} in \secref{shortcircuit}.  We write $\shortify$ to
indicate the location where we apply short-circuiting synthesis to generate the recursive calls.


\newcommand{\jsubtractsyn}[3]{\ensuremath{#1 \vdash [\,#3\,] \;\trianglelefteq\; #2}}

\newsavebox{\SBoxScTemplate}
\begin{lrbox}{\SBoxScTemplate}\small\lststopn
\begin{lstlisting}[language=C,alsolanguage=Spec,alsolanguage=SqueezeOps,style=number]
bool $\chker$sc(ds n,
         $\vec{\texttt{val v}}$, $\vec{\texttt{val hv}}$) {#\lstbeginn#
 assume($\;\icall{\texttt{n}}{\chker}{\vec{\texttt{hv}}}\;$);#\label{line-ksc-assume}#
 if ($\vec{\texttt{v}}$ == $\vec{\texttt{hv}}$) return true;#\label{line-ksc-short}#
##
 return $\shortify(\;\icall{\texttt{n}}{\chker}{\vec{\texttt{v}}}\;)$;#\label{line-ksc-ify}##\lststopn#
}
\end{lstlisting}
\end{lrbox}


\newsavebox{\SBoxScsegTemplate}
\begin{lrbox}{\SBoxScsegTemplate}\small
\begin{lstlisting}[language=C,alsolanguage=Spec,alsolanguage=SqueezeOps,style=number]
bool $\chker_1\chker_2$segsc(ds $\fmtvar{n}_1$, ds $\fmtvar{e}_2$, $\vec{\texttt{val v}_1'}$, $\vec{\texttt{val hv}_1}$, $\vec{\texttt{val* ev}_2'}$, $\vec{\texttt{val ehv}_2 }$) {#\lstbeginn#
 assume($(\icall{\texttt{n}_1}{\chker_1}{\vec{\texttt{hv}_1}} \lreviwand \icall{\texttt{e}_2}{\chker_2}{\vec{\texttt{ehv}_2}}) \;\;\lor\;\; \icall{\texttt{n}_1}{\chker_1}{\vec{\texttt{hv}_1}}$);#\label{line-ksegsc-assume}#
 if ($\vec{\texttt{v}_1'}$ == $\vec{\texttt{hv}_1}$) { return true; }#\label{line-ksegsc-short}#
 if ($\fmtvar{n}_1$ == $\fmtvar{e}_2$) { $\vec{\texttt{*ev}_2'}$ = $\vec{\texttt{v}_1'}$; return true; }#\label{line-ksegsc-endpoint}#
 return $\shortify(\;\icall{\texttt{n}_1}{\chker_1}{\vec{\texttt{v}_1'}} \lreviwand \icall{\texttt{e}_2}{\chker_2}{\vec{\texttt{*ev}_2'}} \;\lor\;\; \icall{\texttt{n}_1}{\chker_1}{\vec{\texttt{v}_1'}}\;)$;#\label{line-ksegsc-ify}##\lststopn#
}
\end{lstlisting}
\end{lrbox}

\newsavebox{\SBoxStoreSyntax}
\begin{lrbox}{\SBoxStoreSyntax}\footnotesize
\begin{grammar}
  \astore \in \astoreset & \bnfdef &
  \astore_1 \lsep \astore_2
  \bnfalt \lemp
  \bnfalt \fldpt{\aaddr}{\fld}{\aval}
  \bnfalt \astorepure
  \bnfalt
  \achkreg
  \bnfalt
  \achkreg_1 \lreviwand \achkreg_2
  \bnfalt \ltrue
  \\
      \achkreg & \bnfdef & \achkercall
  \\
  \multicolumn{3}{@{}c@{}}{\(
    \apure\;\;\;\text{pure fact}
    \quad
    \fld\;\;\;\text{fields}
    \quad
    \chker\;\;\;\text{validation checker}
    \quad
    \aaddr, \aval\in\avalset\;\;\;\text{symbolic variables}
  \)}
\end{grammar}
\end{lrbox}

\begin{figure}
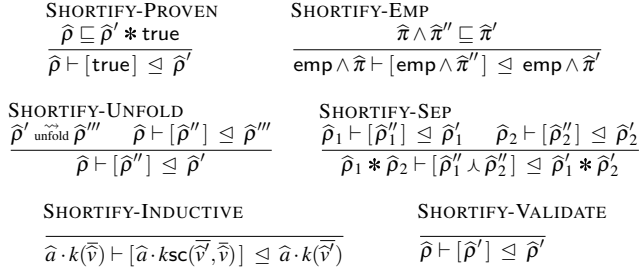
\small\centering
\subfloat[Assert from \texttt{excisemin} (\figref{excisemin}).]{\label{fig:preshortify}
\rule{1em}{0pt}\ovalbox{\scalebox{0.9}{\usebox{\SBoxPreShortify}}}
}
\\
\subfloat[Short-circuiting template.]{\label{fig:sctemplate}
\rule{1em}{0pt}\ovalbox{\scalebox{0.9}{\usebox{\SBoxScTemplate}}}\rule{1em}{0pt}
}
\subfloat[Synthesized short-circuiting.]{\label{fig:postshortify}
\ovalbox{\scalebox{0.9}{\usebox{\SBoxPostShortify}}}
}
%
\hfill\subfloat[Abstract store.]{\label{fig:abs-store-syntax}
\scalebox{0.9}{\usebox{\SBoxStoreSyntax}}
}
\par\noindent\subfloat[Subtraction-directed synthesis of short-circuiting validation calls.]{\label{fig:subtraction}
\begin{mathpar}
\inferrule[Shortify-Proven]{
  \astore \sqsubseteq \astore' \lsep \ltrue
}{
  \jsubtractsyn{\astore}{\astore'}{\ltrue}
}

\inferrule[Shortify-Emp]{
  \apure \land \apure'' \sqsubseteq \apure'
}{
  \jsubtractsyn{\lemp \land \apure}{\lemp \land \apure'}{\lemp \land \apure''}
}
\\
\inferrule[Shortify-Unfold]{
  \astore' \unfold \astore'''
  \\
  \jsubtractsyn{\astore}{\astore'''}{\astore''}
}{
  \jsubtractsyn{\astore}{\astore'}{\astore''}
}

\inferrule[Shortify-Sep]{
  \jsubtractsyn{\astore_1}{\astore_1'}{\astore_1''}
  \\
  \jsubtractsyn{\astore_2}{\astore_2'}{\astore_2''}
}{
  \jsubtractsyn{\astore_1 \lsep \astore_2}{\astore_1' \lsep \astore_2'}{\astore_1'' \ltrueandtrue \astore_2''}
}
\\
%
\inferrule[Shortify-Inductive]{
}{
  \jsubtractsyn{ \achkercall }{ \icall{\aaddr}{\chker}{\vec{\aval'}} }{
    \icall{\aaddr}{\ichk{\chker}}{\vec{\aval'}, \vec{\aval}}
  }
}
%

\inferrule[Shortify-Validate]{
}{
  \jsubtractsyn{\astore}{\astore'}{\astore'}
}
\end{mathpar}
}
\caption{Synthesizing short-circuiting validation with short-circuiting checker templates and inductive shape analysis.}
\label{fig:synthesis-bigfig}
\end{figure}

\subsection{Shortify: Subtraction-Directed Synthesis}
\label{sec:shortify}

In this subsection, we formalize our subtraction-directed short-circuiting synthesis algorithm, which is an adaption of a standard inclusion over stores.
%
We assume a programmatic valuation (\secref{logic-variable-instrumentation}) 
to convert from a symbolic variable
to code that, at run time, retrieves a concrete value represented by that variable.
So, we will, for clarity, only use symbolic variables
in the remainder of this section.

In \figref{subtraction}, we define a judgment of the form $\jsubtractsyn{\astore}{\astore'}{\astore''}$.
The judgment form is a three-place relation between the \emph{statically proven store} $\astore$, the \emph{store to be dynamically validated} $\astore''$, and the \emph{assertion store} $\astore'$.
With this judgment form, we want to derive when it is sound to optimize \code[Spec]{assert(}$\astore'$\code{)} in the original program to \code[Spec]{assert(}$\astore''$\code{)} under the assumption of the sound statically-inferred invariant $\astore$.
%
In other words, we have inferred the invariant $\astore$ and are asked to prove an assertion $\astore'$, so
we abduct a hypothesis $\astore''$ that allows us to prove the following implication:
\begin{equation}\label{eqn:shortify-implication}
\astore \limp \left((\astore'' \lsep \ltrue) \limp (\astore' \lsep \ltrue)\right)
\end{equation}
We have thus reduced our synthesis problem to an abduction problem.
%
%
The syntactic form of stores $\astore$ is given in \figref{abs-store-syntax}, which are separation logic formulas with inductive summaries $\achkreg$ and $\achkreg_1 \lreviwand \achkreg_2$.

%
The \TirName{Shortify-Proven} rule describes a degenerate case of short-circuiting synthesis:
if the analysis can prove the inclusion statically, then there is no need
to check anything dynamically and thus it can eliminate the check (i.e., no store constraints $\ltrue$).
%
The remaining rules express the power of our approach: they enable us to go beyond the
false ``all-or-nothing'' dichotomy between static proof and dynamic validation.


The \TirName{Shortify-Emp} rule expresses the axiom for subtraction-based inclusion---the empty store $\lemp$ is contained in $\lemp$---except that the analysis is \emph{not} obligated to prove the pure constraint $\apure'$ assuming $\apure$.  Instead, it is permitted to leave a residual dynamic check $\apure''$. 

The next two rules describe the basic decomposition of the store for subtraction.
Rule \TirName{Shortify-Unfold} permits unfolding of the assertion store $\astore'$. That is, it can unfold $\astore'$ in order to match the structure of statically-inferred invariant $\astore$.  Existential variables introduced as part of unfolding the $\astore'$ are unified with the corresponding value in $\astore$.
To check inclusion of matching disjoint regions, rule \TirName{Shortify-Sep} uses the statically-inferred invariant $\astore_1 \lsep \astore_2$ to prove the separation constraint in assertion store $\astore_1' \lsep \astore_2'$. The synthesized validation for the two regions $\astore_1''$ and $\astore_2''$ are combined with $\astore_1'' \ltrueandtrue \astore_2''$, which is our shorthand for $(\astore_1'' \lsep \ltrue) \land (\astore_2'' \lsep \ltrue)$ meaning that $\astore_1''$ and $\astore_2''$ can be disjoint or overlapping. While the rule would be sound with $\astore_1'' \lsep \astore_2''$, this more general rule emphasizes that the separation constraint in the assertion store is proven by the statically-inferred invariant, so the dynamic validation need not check separation.

The \TirName{Shortify-Inductive} rule is the key rule that leverages the statically-inferred invariant. It attempts to synthesize short-circuiting validation for $\smash{ \icall{\aaddr}{\chker}{\vec{\aval'}} }$ when the statically-inferred invariant has shown $\smash{ \achkercall }$.
The synthesis is guided by the statically-inferred invariant in that the root address $\aaddr$ and the inductive checker $\chker$ match but permits the additional arguments to be different.
%
This rule motivates the safe replacement criterion with the short-circuiting checker \code[highlighting]{|1:bstsc|} given in \secref{shortcircuit}. Stated in terms of $\ichk{\chker}$, the short-circuiting checker $\ichk{\chker}$ should satisfy:
\[
\achkercall
\limp
\left(\icall{\aaddr}{\ichk{\chker}}{\vec{\aval'}, \vec{\aval}}
\limp
\icall{\aaddr}{\chker}{\vec{\aval'}}\right) \;.
\]
An application of this rule means that when asked to shortify $\smash{ \icall{\aaddr}{\chker}{\vec{\aval'}} }$, we rewrite it
to
$\smash{ \icall{\aaddr}{\ichk{\chker}}{\vec{\aval'}, \vec{\aval}} }$. This rewriting is sound
because we have statically proven the assumption $\smash{ \achkercall }$ and checking $\smash{ \icall{\aaddr}{\ichk{\chker}}{\vec{\aval'}, \vec{\aval}} }$ gives us the desired condition $\smash{ \icall{\aaddr}{\chker}{\vec{\aval'}} }$.
%

Separately, we need to synthesize code for the short-circuiting invariant checker $\ichk{\chker}$ itself and ensure that it satisfies its safe replacement criterion.  In \figref{sctemplate}, we show the synthesis template for $\ichk{\chker}$.  The $\shortify$ directive indicates application of this synthesis routine to generate short-circuiting validation of $\icall{\texttt{n}}{\chker}{\vec{\texttt{v}}}$.  Without this directive, we have a shallow short-circuiting invariant checker that satisfies its safe replacement criterion.  The $\shortify$ then yields a deep short-circuiting invariant checker.

The last rule \TirName{Shortify-Validate} is not allowed in our implementation, but we include it here for discussion. This is a degenerate case of short-circuiting synthesis that is dual to \TirName{Shortify-Proven}: it simply gives up on using the static information $\astore$ and simply checks the assertion $\astore'$ dynamically.
By disallowing this rule, we can observe that there are no separation constraints that need to be validated in $\astore''$.


\paragraph{Shortifying Asserts.}

The final output is a code transformation that rewrites all \code[Spec]{assert}s to perform short-circuiting validation at run time---we apply $\shortify$ to every \code[Spec]{assert(}$\astore'$\code{)}.
To see intuitively why this transformation is sound, consider such an \code[Spec]{assert(}$\astore'$\code{)}. Before this synthesis phase, the static shape analysis will have derived the following program assertion:
\[
\{\,\astore\,\}\;\mbox{\code[Spec]{assert(}}\astore'\mbox{\code{)}}\;\{\, \astore \land (\astore' \lsep \ltrue)\,\}
\]
where $\astore$ is the inferred invariant (i.e., the proven store) at the program point before this \code[Spec]{assert}.
Now, shortification infers a predicate $\astore''$ such that the judgment
$\jsubtractsyn{\astore}{\astore'}{\astore''}$ holds, which means the implication
$\astore \limp ((\astore'' \lsep \ltrue) \limp (\astore' \lsep \ltrue))$ holds.
Because the post-condition of the transformed \code[Spec]{assert} is
\[
\{\,\astore\,\}\;\mbox{\code[Spec]{assert(}}\astore''\mbox{\code{)}}\;\{\, \astore \land (\astore'' \lsep \ltrue)\,\} \;,
\]
we get the original post-condition using the above implication.
%

To state the soundness condition for $\trianglelefteq$, we consider 
concrete semantic domains for a concrete memory $\cmem : \caddrset \finitemap \cvalset$
that maps addresses to values where we assume the set of addresses are
contained in values (i.e., $\caddrset \subseteq \cvalset$) and a
valuation $\cvalua : \avalset \rightarrow \cvalset$ that maps symbolic variables to concrete values.
Symbolic variables are existential variables
naming heap addresses and values, and they form
the coordinates for the base pure, data-value domain.
The concretization of an abstract store $\astore$ thus yields a pair of a concrete memory and a valuation (i.e., $\inconc{\cmemvalua[\cmem][\cvalua]}{\astore}$).

\begin{theorem}[Shortify Soundness]\label{thm:shortify-soundness-main}
If $\inconc{\cmemvalua[\cmem][\cvalua]}{\astore}$ and \\
$\jsubtractsyn{\astore}{\astore'}{\astore''}$ and
$\inconc{\cmemvalua[\cmem''][\cvalua]}{\astore''}$
where $\cmem'' \subseteq \cmem$,
then $\inconc{\cmemvalua[\cmem'][\cvalua]}{\astore'}$ for some $\cmem' \subseteq \cmem$.
\end{theorem}
%
%
This statement is a formalization of the implication shown in \eqref{eqn:shortify-implication};
a proof is given \refappendix{shortify-soundness}.

\subsection{Synthesizing Short-Circuiting Segments}
\label{sec:segmentsyn-appendix}

The previous subsection describes a procedure for synthesizing short-circuiting validation sufficient for the \code{setroot} example in \figref{validation} and for most recursive procedures over data structures. In particular, the previous subsection focuses on the case where the static shape analysis need only summarize complete data structures (e.g., a entire binary search tree).
An additional challenge for static shape analysis on iterative programs versus recursive programs is the need to summarize prefixes of data structures.
For example, in the shape invariant shown in \figref{preshortify} before the final \code[Spec]{assert} in \code{excisemin}, there is an inductive segment from $\aaddr_\vart$ to $\aaddr_\varp$.
Recall that we write this segment as the formula
\[
    \bstcall{\aaddr_{\vart}}{\symvninf}{\symvpinf} \lreviwand \bstcall{\aaddr_{\varp}}{\symu_1}{\symu_2} \;,
\]
which describes a binary search tree at $\aaddr_{\vart}$ with allowable range $(\symvninf,\symvpinf]$ except with a hole at $\aaddr_{\varp}$ such that the binary search tree could be completed if the hole is filled with a binary search tree with allowable range $(\symu_1, \symu_2]$. With such inductive segments, the shape analysis has a different kind of summary form to summarize the binary search tree \emph{prefix} from \code{t} to \code{p}.

We can synthesize short-circuiting validation for this different summary form by introducing a new short-circuiting template and a new \TirName{Shortify} rule that introduces calls to this new short-circuiting template.
In
\figref{preshortify}, the inductive segment from the root $\aaddr_\vart$ to the end point $\aaddr_\varp$ is incrementalized with a call to a synthesized \emph{short-circuiting segment checker} \code[highlighting]{|1:bstsegsc|} in
\figref{postshortify}.
To see how the call to \code[highlighting]{|1:bstsegsc|} is synthesized, consider the new rule \TirName{Shortify-Segment} shown in \figref{subtraction-appendix}. This rule synthesizes a call to a short-circuiting segment checker $\iseg{\chker_1}{\chker_2}$ even when the additional arguments differ. That is, when asked to synthesize short-circuit validation for the specific inductive segment $\smash{ \icall{\aaddr_1}{\chker_1}{\vec{\aval_1'}} \lreviwand \icall{\aaddr_2}{\chker_2}{\vec{\aval_2'}} }$, we require only a proof of $\smash{ \icall{\aaddr_1}{\chker_1}{\vec{\aval_1}} \lreviwand \icall{\aaddr_2}{\chker_2}{\vec{\aval_2}} }$ that may differ in the additional arguments.
Requiring this weaker property in the additional arguments is analogous to the difference permitted in \TirName{Shortify-Inductive}.

\newsavebox{\SBoxShortifySegmentRule}
\begin{lrbox}{\SBoxShortifySegmentRule}\small
\(
%
%
%
%
%
\inferrule[Shortify-Segment]{
  \astore'' = \icall{\aaddr_1}{\iseg{\chker_1}{\chker_2}}{\aaddr_2, \vec{\aval_1'}, \vec{\aval_1}, \vec{\aval_2'}, \vec{\aval_2}}
}{
  \jsubtractsyn{
    \left(\icall{\aaddr_1}{\chker_1}{\vec{\aval_1}} \lreviwand
\icall{\aaddr_2}{\chker_2}{\vec{\aval_2}}\right)
  }{
    \left(\icall{\aaddr_1}{\chker_1}{\vec{\aval_1'}} \lreviwand
\icall{\aaddr_2}{\chker_2}{\vec{\aval_2'}}\right)
  }{
    \astore''
  }
}
\)
\end{lrbox}

\begin{figure}\small\centering
%
\subfloat[Short-circuiting segment checker.]{\label{fig:scstemplate-appendix}
\rule{1em}{0pt}\ovalbox{\scalebox{0.85}{\usebox{\SBoxScsegTemplate}}}
}
%
\\
\subfloat[Subtraction-directed synthesis of short-circuiting validation calls.]{\label{fig:subtraction-appendix}
\scalebox{1}{\usebox{\SBoxShortifySegmentRule}}
}
\caption{Synthesizing short-circuiting validation of inductive segments.}
\label{fig:synthesis-bigfig-appendix}
\end{figure}

\newsavebox{\SBoxSegmentSafeReplacement}
\begin{lrbox}{\SBoxSegmentSafeReplacement}\small
\(\begin{array}{@{}l@{\;}l@{}}
\text{if} & \icall{\aaddr_1}{\chker_1}{\vec{\aval_1}} \lreviwand \icall{\aaddr_2}{\chker_2}{\vec{\aval_2}} \\
\text{then} &
\icall{\aaddr_1}{\iseg{\chker_1}{\chker_2}}{\aaddr_2, \vec{\aval_1'}, \vec{\aval_1}, \vec{\aval_2'}, \vec{\aval_2}}
\limp
(\icall{\aaddr_1}{\chker_1}{\vec{\aval_1'}} \lreviwand \icall{\aaddr_2}{\chker_2}{\vec{\aval_2'}})
\end{array}\)
\end{lrbox}

The difference that we permit in the \TirName{Shortify-Segment} rule leads to a safe replacement criterion that the short-circuiting segment checker $\iseg{\chker_1}{\chker_2}$ should satisfy:
\[
\scalebox{0.95}{\usebox{\SBoxSegmentSafeReplacement}}
\]
based on its use in that rule. This safe replacement criterion motivates the short-cir\-cuit\-ing segment checker template given in \figref{scstemplate-appendix} (we have annotated the parameters with subscripts and primes to more readily see related components). It is captured in the first disjunct of the \code[C]{return} condition on \reftxt{line}{line-ksegsc-ify} and the first disjunct of the \code[Spec]{assume} on \reftxt{line}{line-ksegsc-assume}.
	If we ignore the assignment to $\vec{\texttt{*ev}_2'}$ on \reftxt{line}{line-ksegsc-endpoint} and the second disjuncts for now, we have the analogous template as for inductives $\ichk{\chker}$---\reftxt{line}{line-ksegsc-short} checks the short-circuiting condition.

To see where the assignment to $\vec{\texttt{*ev}_2'}$ comes from, let us consider how the symbolic variables $\smash{ \vec{\aval_2'} }$ arise when the \TirName{Shortify-Segment} rule is applied. Because segments are not directly specified in assertions, this segment with symbolic variables $\smash{ \vec{\aval_2'} }$ must result from unfolding during subtraction (using the \TirName{Shortify-Unfold} rule) to match the shape of the inferred invariant.	 For this reason, they are not present in the analysis state and thus not in the programmatic valuation. However, we observe that we can instrument the short-circuiting segment checker template to find concrete witnesses for these symbolic variables.
In the case that the short-circuiting condition applies (\reftxt{line}{line-ksegsc-short}), it must be the case that $\smash{ \vec{\aval_2'} = \vec{\aval_2} }$, that is, the arguments at the end point $\aaddr_2$ correspond to the statically-inferred invariant.
Our synthesis procedure will thus initialize the contents of $\smash{ \vec{\texttt{*ev}_2'} }$ to $\smash{ \vec{\texttt{ehv}_2} }$ at the top-level call site, as we have shown with the assignments to $\symu_1'$ and $\symu_2'$ with $\symu_1$ and $\symu_2$, respectively, in \figref{postshortify}.
In the case that we reach the endpoint \texttt{e} on \reftxt{line}{line-ksegsc-endpoint}, we witness $\smash{ \vec{\aval_2'} }$ with $\smash{ \vec{\aval_1} }$ using the assignment \code{$\vec{\texttt{*ev}}$ = $\;\vec{\texttt{v}}$} using the semantics of the empty segment.

This template design enables additional precision in dynamic validation. Alternatively, we could strengthen the static obligation so that $\smash{ \vec{\aval_2'} = \vec{\aval_2} }$ (i.e., disallow a difference between the additional arguments at $\aaddr_2$) at the cost that short-circuiting applies less often.

Finally, the second disjuncts of the \code[C]{return} condition on \reftxt{line}{line-ksegsc-ify} and the \code[Spec]{assume} on \reftxt{line}{line-ksegsc-assume} are another optimization in our template design. An inductive segment $\achkreg_1 \lreviwand \achkreg_2$ is defined by unfolding $\achkreg_1$ until it either terminates in a base case or matches $\achkreg_2$ (cf. \secref{validated-views}).
The second disjuncts generalize the replacement criterion so that this short-circuiting checker can also be used in the context the static analysis has proven $\icall{\texttt{n}}{\chker}{\vec{\texttt{hv}}}$, that is, where there is a full validation from \code{n} with no endpoint.  This additional assumption permits this same short-circuiting checker to be used in the recursive calls generated by $\shortify$ at \reftxt{line}{line-ksc-ify}, regardless of whether the recursive check is a full validation or a segment.

\paragraph{Discussion: Limitations and Generalizations.}
Fundamentally, synthesizing short-circuiting validation checkers from templates requires proving implications in inductive separation logic and thus is necessarily incomplete.  To synthesize from the templates 
shown in Figures~\ref{fig:sctemplate} and~\ref{fig:scstemplate-appendix}, we unfold once in the \code[Spec]{assume} on \reftxt{line}{line-ksc-assume}
and in the \code[C]{return} on
\reftxt{line}{line-ksc-ify}.
This strategy is effective at proving the implication of interest for common inductive definitions, like $\bst$, that unfold ``one-node-at-a-time''---but, in general, is not guaranteed to succeed.

As we can observe from the above discussion about inductive segments, there is a design trade-off between the \TirName{Shortify} rules and the short-circuiting templates. The \TirName{Shortify} rules specify the desired weaker static verification obligation, which dictates the remaining obligation that must be satisfied by the synthesized short-circuiting checker.
In this paper, we have asked the static verification to prove the same structure (cf. \TirName{Shortify-Inductive} and \TirName{Shortify-Segment}) with potentially different additional arguments.
This design decision leads to short-circuiting templates that are both effective in realizing short-circuiting validation and reasonable to synthesize.
But clearly, this choice is simply one in a potentially large design space. The dynamic validation of different kinds of summaries may be shortified with a similar approach using different \TirName{Shortify} rules but potentially requiring different synthesis strategies for short-circuiting checker templates.

\section{Empirical Evaluation}
\label{sec:empirical-evaluation}



\PUNCH{Setting up the section.  This section will consider the
  hypotheses.}


We test our short-circuiting validation approach using a
prototype implementation called \Tool{}. In particular, we seek to
test the following expectations:
\begin{inparaenum}[(1)]
\item Theoretical asymptotic improvements from short-circuiting
  translate to dramatic
  empirical speedups in run-time validation on workloads even
  assuming worst-case inference imprecision in the static verifier.
\item Short-circuiting validation is effective even when the
  pre-conditions are too weak to imply the validated post-condition
  (i.e., when it is not sound to eliminate the post-condition).
\item The absence of a shadow heap in \Tool{} results in
  incremental validation with low overhead.
\end{inparaenum}

\paragraph{Experimental Setup and Static Analysis Time.}

%
%
We use \Tool{} to incrementally verify-validate the set of workloads
shown in \figref{ivv-performance}. Each data structure has an
invariant validation checker that requires numerical constraints on
the data structure contents.
All workloads consist of repeatedly applying a data structure
operation assuming the data structure invariant holds on input and then asserting
that the data structure invariant again holds for the output.

%
%
We consider both classic data structures, such as ordered
singly-linked lists (\liststruct{}) and binary search trees
(\bststruct{}), as well as structures less commonly
considered by static verifiers: treaps (\treapstruct{}) and hash tries
(\hashtriestruct{}). A treap is a randomized version of a binary
search tree that is heap-ordered on a set of randomly-generated
priorities; a hash trie is a tree-structured hash map that uses
bit-blocks from hash keys to form a
trie~\cite{bagwell2001:ideal-hash}---challenging even for state-of-the-art static
verification.
And even for binary search trees---where the invariant seems amenable to current static reasoning techniques---certain operations like deletion (that have non-local modification) can be surprisingly challenging to verify statically compared to simpler operations like insertion~\cite{Magill-Personal-Communication-2014}.
%
Our  validated views domain handles
these structures for short-circuiting synthesis because it can preserve uninterpreted data constraints on unmodified regions.
\SQUEEZE{Overall, we have specifically considered tree-like
structures in order to obtain reliable inference of shape invariants as
input to the synthesizer.}

%
%
To stress test our approach, we intentionally instantiate the
shape-data analyzer in \Tool{} with the ``most lossy'' data
domain. This domain can only keep constraints in
straight-line code and loses all of them (i.e., goes to $\top$) at any
join point. Complex operations, like bit operations, are treated as
uninterpreted functions. This imprecision represents the
scenario in which the data property is out of reach of static verification.

We measured the static analysis time to infer invariants
for our examples. As our primary concern is not the efficiency of the static
analysis, our data-value domain is not particularly optimized.  It
calls the Z3 SMT-solver for inclusion checking.  The total analysis time for
all our examples was 640.2s.
It is likely that improved engineering can lower the analysis time by
making fewer SMT calls. If we exclude the Z3 time, then the remaining
analysis time was 27.9s.
\SQUEEZE{
The code size expansion for these benchmark falls in 80-270\% due to
instrumentation. 
Note that because we have deliberately instantiated \Tool{} with the
``lossy'' domain, our analyzer will continue to instrument code along
many infeasible paths.  Thus, while code size was not a chief concern
of ours, our reported figures are essentially the worst case.
}






\begin{figure}\centering\small
  \subfloat[Execution time of data structure workloads
  with short-circuiting incremental verification-validation (\Inc); no
  validation (\Unsafe); and full dynamic validation (\Noninc).  The
  workloads execute $m$ data structure
  operations---with a dynamic validation check between each operation
  in the \Noninc{} and \Inc{} variants.
  We chose the number of operations $m$ to
  ensure that the workloads run long enough to accurately measure elapsed time.
  For most workloads,
  the fastest variant (that with no validation, \Unsafe{})
  takes at least 1 second. In cases where the slowest, fully dynamic validation variant
  (\Noninc{}) would not finish within hours on that size, we chose $m$ to
  ensure that the variant without validation (\Unsafe{}) takes a minimum of
  0.05s.
  These workloads
are compiled with \texttt{gcc -O2} and measured on an Intel Core 2 Quad CPU
Q6600 2.40GHz with 4G RAM. The time reported is averaged over 12 runs.]{\label{fig:ivv-performance}
\begin{tabular*}{\linewidth}{@{\extracolsep{\fill}}lr rrr@{}}\toprule
                & $m$       & \Inc  & slowdown  & speedup     \\
workload        & (K)       & (s)   & wrt \Unsafe  & wrt \Noninc   \\
\midrule
\listconcat     & 1,180     & 1.163 &     1.07x &      2.1x \\
\listdrop       & 600       & 0.053 &     1.01x & 230000~~x \\
\listinsert     & 24        & 1.185 &     1.12x &      2.8x \\
\listdelete     & 23        & 1.087 &     1.08x &      2.9x \\
\listinsdel     & 27        & 1.195 &     1.10x &      2.8x \\
\midrule
\bstinsert      & 200       & 0.075 &     1.40x &   5600~~x \\
\bstdeletemin   & 460       & 0.087 &     1.66x &  83000~~x \\
\bstexciseroot  & 380       & 0.157 &     2.80x &  26000~~x \\
\bstdelete      & 170       & 0.097 &     1.76x &   3300~~x \\
\bstinsdel      & 170       & 0.079 &     1.41x &    360~~x \\
\midrule
\treapdelete    & 900       & 1.171 &     1.12x &      3.9x \\ 
\treapinsert    & 970       & 1.399 &     1.35x &      3.3x \\
\treapinsdel    & 660       & 1.151 &     1.09x &      2.3x \\
\midrule
\hashtriewrite  & 140       & 0.230 &     4.35x &   9300~~x \\
\hashtrieinsdel & 280       & 0.305 &     5.78x &   1000~~x \\
\midrule
geometric mean  &           &       &     1.54x &    $-$ 
\\
\bottomrule
\end{tabular*}
}
\par\subfloat[Short-circuiting incremental verification-validation does not change the apparent
complexity of \bststruct{} $m$
\texttt{insertdelete}s.  The \Inc{} line shows the cost of $m$ validated operations
on a data structure of size $n$. The \Unsafe{} line shows the cost on the original,
unvalidated code.]{\label{fig:asymptotic}
\includegraphics[width=0.97\linewidth]{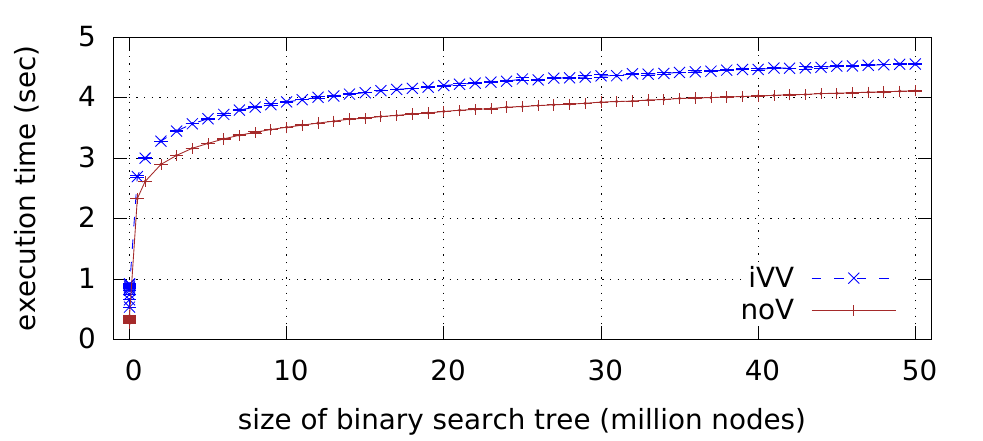} 
}
\caption{Effectiveness of incremental
  verification-validation.}
\end{figure}

\paragraph{Incrementalized Run-Time Validation Cost on Workloads.}
In \figref{ivv-performance}, we report the cost and benefit of incrementalized run-time validation
for a variety of data structure workloads.
%
%
%
The measurements under `\Inc{}' show the execution time with synthesized
short-circuiting dynamic validation. The subsequent
columns give the slowdown of \Inc{} over no validation (`\Unsafe{}') and
speedup of \Inc{} over dynamic validation (`\Noninc{}').
The \Noninc{} configuration runs C-style data structure invariant checks that dynamically validate the data constraints for the complete structure but do not validate separation properties.
This configuration provides more checking than \Unsafe{}, but it is still weaker than \Inc{}, which checks separation properties statically.
Note that strengthening \Noninc{} to dynamically validate separation can only make the \Inc{} speedups larger. We
provide only the speedups here, but we give the detailed execution times \refappendix{workloads}.
%

  For ordered singly-linked
  lists \texttt{olist}, \texttt{concat} iteratively concatenates
  sub-lists until reaching a list of size $m$, \texttt{drop} drops the first element of the list $m$ times
  from a list of size $m$, \texttt{insert}
  performs ordered insertion $m$ times from an empty list,
  \texttt{delete} performs element deletion $m$ times from a list of
  size $m$. For binary search tree \texttt{bst}, \texttt{insert} performs search tree
  insertion, \texttt{deletemin} deletes the
  minimum element, \texttt{exciseroot} excises the root node,
  and \texttt{delete} does search tree deletion.
  For treaps \texttt{treap}, search tree rotations are performed on
  insertion and deletion to restore the heap-ordering invariant; we
  perform validation after each rotation.  For hash tries \texttt{hashtrie}, \texttt{write} writes to the hash trie, and we validate
  that the hash key of each element is along the correct path in the
  trie for the last 14 levels; this dynamic validation limit is
  because of the lack of arbitrary precision integers in C.  The 50-50
  workloads consist of a randomly-selected sequence of insert and
  delete operations with 50-50 split of each kind; the data structure
  size for these workloads was 10K nodes, except for the
  \treapstruct{} workloads, which were at 1M nodes. The code for \texttt{hashtrie write} along with its invariant checker is given \refappendix{hashtrie}.

In most cases, our \Inc{} approach brings the
slowdown with respect to the \Unsafe{} configuration back down to
within a factor of 1.76x. One exception is \texttt{exciseroot}:
this workload is a bad case for short-circuiting because the operation
itself requires just one $O(\lg n)$ walk from the root, while the
theoretically ideal short-circuiting validation
has to perform three $O(\lg n)$ walks. This larger slowdown is thus consistent with our expectations.
We discuss
the exceptions in
the \hashtriestruct{} workloads further below.


The speedup column shows the improvement of incrementalized
verification-validation (\Inc{}) over non-short-circuiting,
whole-structure dynamic validation  (\Noninc{}). We expect
short-circuiting to
obtain asymptotic improvements over \Noninc{}
(or rather, recover the asymptotic losses of the \Noninc{} variant
over the unchecked \Unsafe{} configuration). Our expectation is
supported by the measurements shown in \figref{ivv-performance}: the
cases with single-digit factor speedups are those where there
is no change in asymptotic complexity between \Unsafe{} and \Noninc{},
while the cases with multiple order-of-magnitude speedups
are those where there is a change in asymptotic complexity (the exact factors are thus not meaningful).
Recall that dynamic validation does not check separation properties---these
performance improvements are on top of the
additional separation
assurances provided by our static verifier.

We obtain these speedups with relatively unoptimized
instrumentation.  We never eliminate any \code[Spec]{assert}s entirely---we only
replace non-short-circuiting, whole-structure
assertions
(like on \reftxt{line}{line:setroot-assert} in \figref{validation:bst}) with
calls to short-circuiting checkers
(like on \reftxt{line}{line:setroot-assert} in \figref{validation:bstsc}).
To
focus our measurements on short-circuiting, we do not eliminate data
invariant \code[Spec]{assert}s (e.g., \code{x < y}) even when they are
known to hold statically.  Removing such statically-proven
\code[Spec]{assert}s would especially improve the
\hashtriestruct{} workloads.  For the other data structures, the
data invariant at each node is a constant-time check.  But for
\hashtriestruct{}, the data invariant check involves recomputing the
hash key of an element, which is a $O(\lg n)$ operation (adding a
$O(\lg n)$ factor over \Unsafe{} but better than the $O(n)$ factor
added by \Noninc{}).

Overall, these workload performance measurements are quite promising,
particularly given the two intentional worst-case choices of (1)
no \code[Spec]{assert} eliminations even when they are statically known;
and (2) a worst-case base data-value domain.  Applying 
static verification techniques to data-value reasoning, instead of hamstringing it, would only improve the run-time performance
of \Inc{}.

\PUNCH{We setup the analysis with a lossy pure domain to evaluate this
  approach. Why lossy? Worst case for us.}

\PUNCH{The data structure invariants are checked in the end.  Note not all
  examples are statically verifiable--see later.}

\PUNCH{Discuss the table on ``medium'' comparing full checking with
  incremental checking.}

\PUNCH{Discuss worst case/ general case analysis for bst}

\PUNCH{Discuss asymptotics and scaling}

\paragraph{When Pre-Conditions are Insufficient to Verify.}
%
%
Not only can \Tool{} incrementalize to make up for inference
imprecision in static analysis, but it can also tolerate ``insufficient''
pre-condition specifications. Insufficient pre-condition specifications are ones that are not strong enough for any static verifier to soundly eliminate checking the
post-condition (like the \code{setroot} example from \figref{validation}).  Here, the \texttt{concat} example and
both treap examples use insufficient pre-conditions.  For the
\texttt{concat} example, the pre-condition for each input is that they
are both ordered lists with a lower bound parameter of $-\infty$.  For
the \texttt{treap} examples, the pre-condition to the rotation
operation is that both sub-trees have binary search tree bounds of
$(-\infty,\infty)$.  \Tool{} generates a short-circuiting
validation that checks that the observed concrete instances satisfy
the post-condition. This synthesized
validation soundly fails when passing concrete inputs that
would violate the data structure invariant (e.g., when
passing two individually ordered lists to
\texttt{concat} whose concatenation is not an ordered list).


\PUNCH{Append example where we check for a run-time violation.}

\PUNCH{Able to incrementalize with imprecise assumptions (treap).}


\newcommand{\bstminratio}[0]{1.11}
\newcommand{\bstmaxratio}[0]{2.55}

\newcommand{\bstasymptoticslowdownpercent}[0]{4.7}
\newcommand{\bstasymptoticslowdownpercenterror}[0]{1.4}
\newcommand{\bstasymptoticivvrsquared}[0]{0.984}
\newcommand{\bstasymptoticnovrsquared}[0]{0.995}

\paragraph{Low Overhead Validation.}
As we have seen, incrementalizing dynamic validation checks
can improve their asymptotic complexity---and, in some cases, even
recover the complexity of the original, unvalidated operation while 
providing safety guarantees. Recall the binary search tree
example from \figref{unscalable}---there, adding whole-data-structure
validation to insert and delete
operations changed the asymptotic complexity from $O(\lg n)$ to
$O(n)$, but applying 
\Tool{}'s short-circuiting approach retained the benefits
of dynamic validation without changing the complexity of the operation.

In \figref{asymptotic}, we show a plot of \Inc{} and \Unsafe{} execution
times on a workload of $m$ binary search tree \texttt{insertdelete}
operations (this is the plot from \figref{unscalable} with
the dynV curve removed). An
\texttt{insertdelete} performs a deletion of a random
element
followed immediately by the re-insertion of the same element (this keeps the
data structure size $n$ constant over the $m$ operations); here we fix
$m = 1$ million so that the workload execution time is easily
measurable (on the order of 1 second) on even the smallest data structure
size.
We performed $k = 12$ trials at each data structure
size $n$. The error bars (though negligible) show two standard deviations.

Over this range, the slowdown of \Inc{} over \Unsafe{} ranges from
{\bstminratio}x (for the largest \texttt{insertdelete} 
workloads) to {\bstmaxratio}x (for the smallest).
We use this same data to estimate the
\emph{asymptotic} slowdown---the slowdown in the limit,
as the number of nodes in the tree increases---by fitting
to equations of the form
$\mathit{time} = b + c\cdot\mathrm{lg}(\mathit{size})$.
The asymptotic slowdown is then $c_\mathrm{\Inc{}}/c_\mathrm{\Unsafe{}}$---which
we estimate to a 95\% confidence interval
to be \bstasymptoticslowdownpercent{}\% $\pm$ \bstasymptoticslowdownpercenterror{}.
(For completeness, we provide both a plot of slowdown vs. workload size
and the details of this fitting \refappendix{workloads}).

This asymptotic slowdown of 4.7\% compares quite favorably with purely dynamic approaches
to validation. While not a direct comparison, the overhead for Ditto ranges from
2.5x to 50x,
depending on the data structure~\cite{Shankar-Bodik-Personal-Communication-2014}.
Such dynamic techniques are quite general
but must maintain sizable run-time structures in a shadow
heap.
In contrast, when static shape analysis
is possible---and it often is---then \Tool{}'s combined verification-validation approach
allows for a much lower overhead.

\section{Related Work}
\label{sec:related-work}



Several ways of combining pointer-shape reasoning
and data-value reasoning have been proposed,
including~\cite{chang+2008:relational-inductive, DBLP:conf/sas/McCloskeyRS10,
  DBLP:conf/pldi/BouajjaniDES11}.
\citet{magill+2010:automatic-numeric}
proposes a sequence of two analyses: a shape analysis followed by an
off-the-shelf scalar value analysis.
\SQUEEZE{%
There are clear engineering
benefits to this approach, though the pipeline architecture means that
heap reasoning cannot benefit directly from the contents analysis
results.
}%
The kind of invariants that we consider require a
tight integration of pointer-shape and data-value reasoning, which may
be difficult to realize in a pipelined approach.
\SQUEEZE{%
A simultaneous combination of pointer-shape and data-value approach
implementing a kind of reduction~\cite{cousot+1979:systematic-design}
has been described previously~\cite{chang+2008:relational-inductive}.
}%
Incremental verification-validation can be seen as supplementing these
approaches with the ability to discharge complex data-value
constraints on heap contents at run-time in a fine-grained manner
(whereas simpler data-value constraints can still be discharged
statically).
%
%
There is a wealth of static verification techniques for data structure
properties that combine pointer-shape and data-value constraints,
including\citeverbelse{nguyen+2007:automated-verification,DBLP:conf/popl/MadhusudanQS12,DBLP:conf/pldi/Qiu0SM13,mcpeak+2005:data-structure}{nguyen+2007:automated-verification,DBLP:conf/pldi/ZeeKR08,DBLP:conf/pldi/Qiu0SM13,DBLP:conf/pldi/PekQM14} amongst many others,
that use
user-supplied loop invariants and try to discharge proof obligations using
solvers and decision procedures.
\citet{DBLP:conf/pldi/ZeeKR08} aim to verify not only rich
shape-numerical data structure properties but also functional
properties of algorithms that operate on them.  Our technique is complementary:
it offers more automation and the opportunity
to discharge the hardest constraints dynamically.
In another line of work, \citet{DBLP:conf/vmcai/Christakis0W15} proposed using dynamic validation to study \emph{unsound assumptions} made by static analysis tools.
In contrast, our approach uses dynamic validation to check \emph{unverified assertions} from a sound static analysis (i.e., assertions that the analysis could not prove statically).

Prior approaches to run-time checking of expressive heap
assertions
includes work that parallelizes
validation~\cite{vechev+2010:phalanx:-parallel}, allows
user-specified overhead budgets~\cite{DBLP:journals/tosem/ArnoldVY11}, checks separation
properties~\cite{nguyen+2008:runtime-checking}, and piggybacks on the
GC with an expressive assertion
language~\cite{DBLP:conf/oopsla/ReichenbachISAG10}.  In this space,
our work is most related to Ditto~\cite{shankar+2007:ditto:-automatic}
where our bottom-line goal is the same: to incrementalize 
data structure validation checks.  We present a complementary
approach that offers different trade-offs.  For example, our memory
overhead is dictated by the size of static shape invariants that are
independent of the size of the concrete data structure instances.
\SQUEEZE{%
Ditto's optimistic memoization algorithm can be seen as a dynamic
analogue to segment invariants inferred by a static shape analysis.
Optimistic memoization permits ``jumping'' into ``future'' recursive
validation checks before scrutinizing the optimistic assumption on the
upward return path.  Correspondingly, our short-circuiting segment validators $\iseg{\chker_1}{\chker_2}$
check on the downward call path whether the shape invariant permits a
``jumping'' to the endpoint.
%
}%
Incremental
computation~\cite{DBLP:conf/popl/LiuST96,DBLP:conf/popl/AcarAB08,DBLP:conf/pldi/HammerAC09}
has been used for improving execution efficiency in general, beyond
data structure validation checks.  The basic principle is to memoize
intermediate results and reuse them when a dependency is
discovered.  Our approach can be seen as a static analogue to these principles, applying static analysis to ``statically memoize'' validation checks so
that only a single short-circuiting check at the root is needed to reuse the results
for the entire sub-tree.

\section{Conclusion}
\label{sec:conclusion}

We have presented a proof-directed approach to synthesizing short-circuiting dynamic validation checks for data structures with complex data invariants.
A short-circuiting checker stops checking whenever it detects at \emph{run time} that an assertion for \emph{some sub-structure} has been proven \emph{statically}.
The key insight of our approach is to obtain a static verification for an existential weakening of the original property of interest and then to use this proof to synthesize short-circuiting validation. To achieve this technically,
we first defined an enrichment of inductive shape analysis with validated views. This abstract domain enables run-time reification of static logic variables and
precisely tracks uninterpreted data invariants in unmodified regions.
Finally, we described a subtraction-directed synthesis technique that abducts short-circuiting validation checks for data constraints using the statically-inferred invariants.
%

The result is a hybrid
technique that strikes a unique balance of static
verification and dynamic validation for data structure invariants.  On
one hand, our technique provides benefit with efficient
dynamic validation when faced with the inevitable imprecision of
statically inferred data invariants on inductive structures.  On
the other hand, our technique leverages static verification to
provide additional static assurances (i.e., separation
properties) and substantially lower the run-time overhead of dynamic
validation---in many cases recovering the complexity of the original
unsafe and unvalidated code.

\bibliographystyle{plainnat}
\bibliography{conferences.short,bec.short,reference.short}

\ifTR
\appendix

\section{Data Structure Workloads}
\label{app:workloads}

\begin{table}\centering\small
\nocaptionrule\caption{The execution overhead of running data structure workloads
  with dynamic validation (\Noninc) compared with no validation
  (\Unsafe).}
\begin{tabular*}{\linewidth}{@{\extracolsep{\fill}}l r r rr@{}}\toprule
                & $m$       & \Unsafe   & \Noninc       & slowdown  \\
workload        & (K)       & (s)       & (s)           & wrt \Unsafe \\
\midrule
\listconcat     & 1,180     &   1.086 &     2.399       &      2.2x \\
\listdrop       & 600       &   0.053 & 12472.260       & 236109.1x \\
\listinsert     & 24        &   1.059 &     3.294       &      3.1x \\ 
\listdelete     & 23        &   1.007 &     3.147       &      3.1x \\
\listinsdel     & 27        &   1.084 &     3.390       &      3.1x \\
\midrule
\bstinsert      & 200       &   0.053 &   415.414       &   7782.3x \\
\bstdeletemin   & 460       &   0.052 &  7262.621       & 138725.5x \\
\bstexciseroot  & 380       &   0.056 &  4101.548       &  73395.5x \\
\bstdelete      & 170       &   0.055 &   318.750       &   5809.1x \\
\bstinsdel      & 170       &   0.056 &    28.600       &    508.5x \\
\midrule
\treapdelete    & 900       &   1.046 &     4.595       &      4.4x \\
\treapinsert    & 970       &   1.039 &     4.568       &      4.4x \\
\treapinsdel    & 660       &   1.052 &     2.629       &      2.5x  \\
\midrule
\hashtriewrite  & 140       &   0.053 &  2131.922       &  40326.6x \\
\hashtrieinsdel & 280       &   0.053 &   316.613       &   6000.3x \\
\bottomrule
\end{tabular*}
\label{tbl:validation-slowdown}
\end{table}

\begin{figure}
\includegraphics[width=0.97\linewidth]{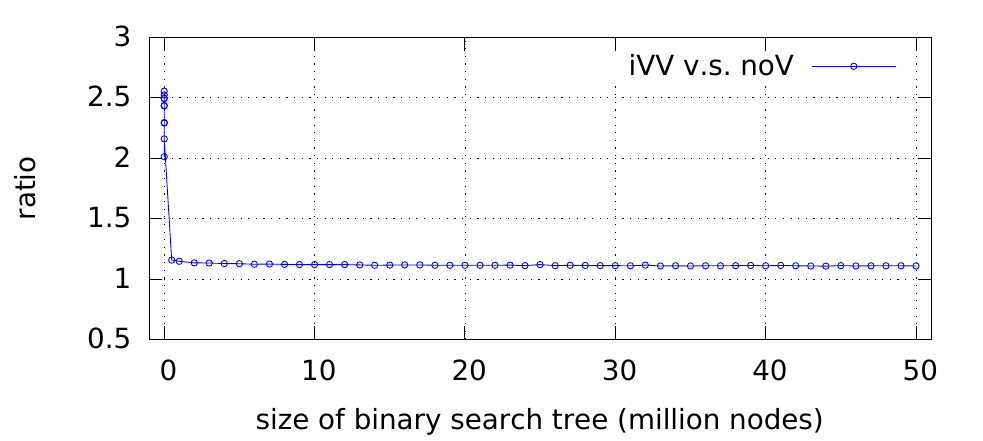} 
\caption{Ratio of \texttt{insertdelete} execution times under incremental
verification-validation (\Inc{}) and for the original, unsafe code
(\Unsafe{}). }
\label{fig:ratios}
\end{figure}

\Tblref{validation-slowdown} gives the execution times over our
workloads from \secref{empirical-evaluation} on the no validation
(\Unsafe{}) and the full, non-incrementalized (\Noninc{}) variants.
These workloads
are compiled with \texttt{gcc -O2} and measured on an Intel Core 2 Quad CPU
Q6600 2.40GHz with 4G RAM. The time reported is averaged over 12 runs.


\subsection{Analysis of BST \texttt{insertdelete} Workload}

In this subsection, we provide additional details for the
measurements of the overhead of \Inc{} over \Unsafe{} that
we reported in \secref{empirical-evaluation}.

\paragraph{Behavior Over Measured Range.}

In \Figref{ratios}, we show the ratio of \Inc{} execution times
over \Unsafe{} executions for the BST \texttt{insertdelete} over the range of tree
sizes that we measured (up to 50 million nodes).
The slowdown of \Inc{} for our validation (over the original
unsafe,
unvalidated code) ranges from
around {\bstminratio}x (for the largest \texttt{insertdelete} 
workloads) to {\bstmaxratio}x (for the smallest). 

\begin{figure}
\centering
\subfloat[Linear fittings of lg-transformed curve.]{\label{fig:asymptotic:fitting-lines}
\includegraphics[width=0.97\linewidth]{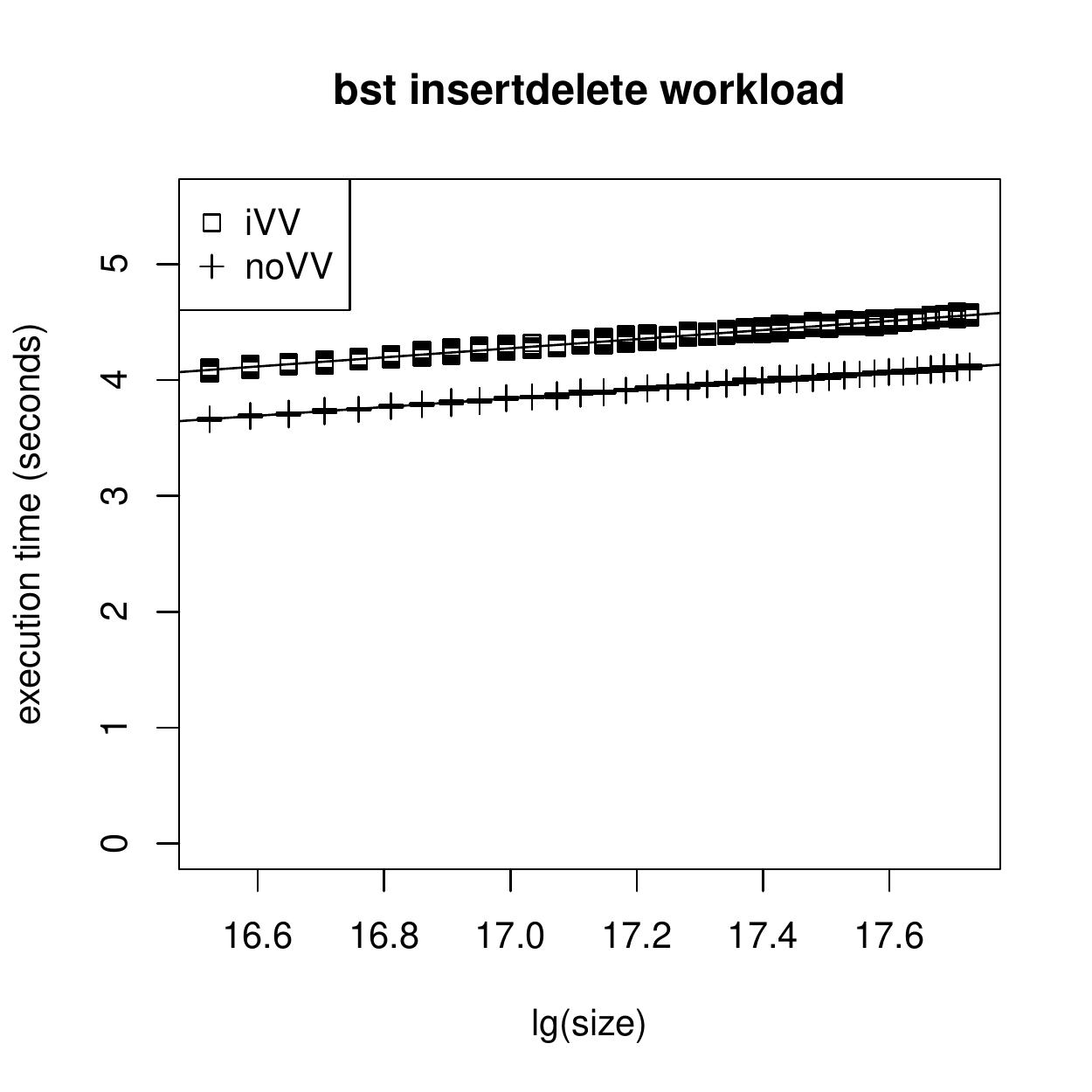}
}

\subfloat[Spread of residuals for fitted curves.]{\label{fig:asymptotic:fitting-residuals}
\includegraphics[width=0.47\linewidth]{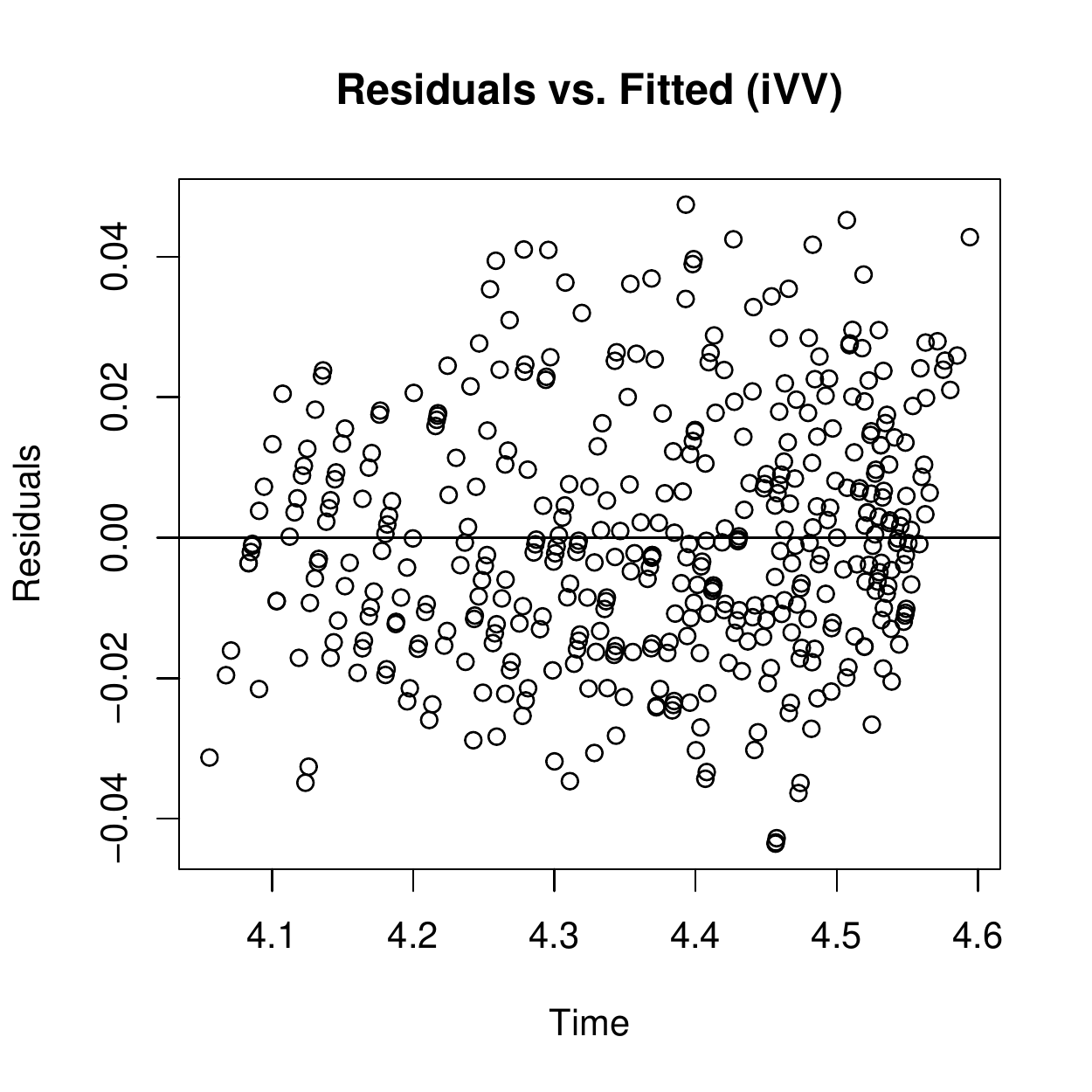}
\hfill
\includegraphics[width=0.47\linewidth]{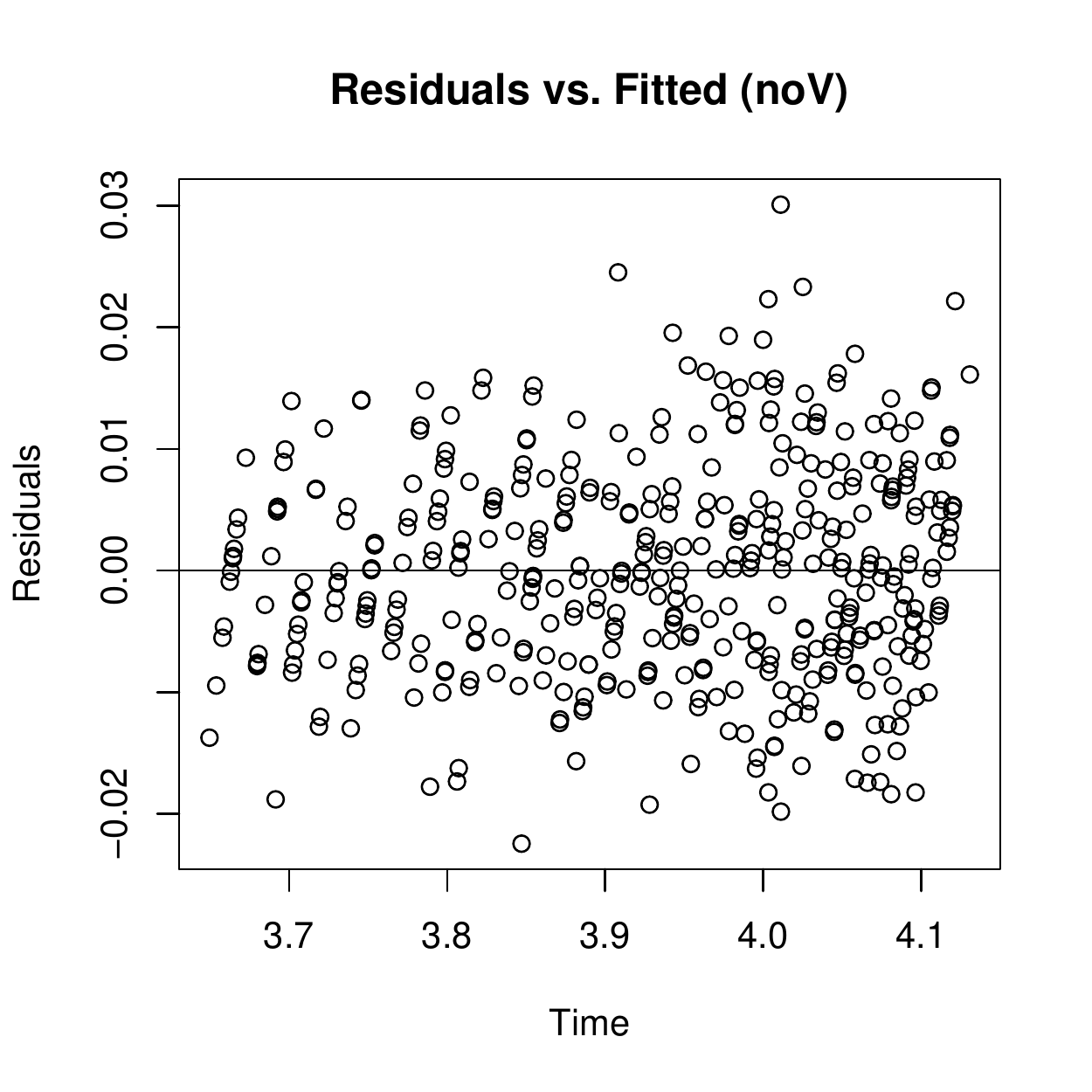}
}

\caption{We estimate the asymptotic overhead of our approach on the BST \texttt{insertdelete} workload by
fitting lg-transformed curves.}
\label{fig:asymptotic:fitting}
\end{figure}

\paragraph{Estimated Asymptotic Slowdown.} We also performed
an analysis to estimate the \emph{asymptotic} slowdown of
\Inc{} vs. \Unsafe{}---that is, the slowdown in the limit,
as the size of the tree grows to infinity. 
To estimate this slowdown, we fit a suffix (starting
at 15 million nodes) of the data from
~\figref{asymptotic} to equations of the form:
\[ \mathit{time} = b + c\cdot\mathrm{lg}(\mathit{size}) \]

\par\noindent
The asymptotic slowdown is then $c_\mathrm{\Inc{}}/c_\mathrm{\Unsafe{}}$---which
we estimate to a 95\% confidence interval
to be \bstasymptoticslowdownpercent{}\% $\pm$ \bstasymptoticslowdownpercenterror{}.
We show the fitted lines in \figref{asymptotic:fitting-lines}. 
The coefficients of determination ($\mathrm{R}^2$)
for the fit for \Inc{} is \bstasymptoticivvrsquared{}; that for \Unsafe{} is \bstasymptoticnovrsquared{}.
The lack of a trend
in the spread of the residuals (\figref{asymptotic:fitting-residuals}) for these
fits provides evidence for a linear relationship between execution time and
$\mathrm{lg}(\mathit{size})$.

\section{The Validated View Abstraction}
\label{app:validated-view-formal}

\begin{figure}\scriptsize\centering
\subfloat[Concrete domains.]{\label{fig:abs-concrete}
\begin{mathpar}
  \cmem : (\caddrset \finitemap \cvalset) = \cmemset

  \cvalua : (\avalset \rightarrow \cvalset) = \cvaluaset
\end{mathpar}
}
\setlength{\saveboxsep}{\fboxsep}
\fboxsep 0pt
\subfloat[A memory abstraction with validated views.]{\label{fig:abs-syntax}
\begin{tabular}{@{}c@{}}
\begin{grammar}
  \amem \in \amemset  & \bnfdef & \amem_1 \lor \amem_2 \bnfalt \bot \bnfalt \astore
  &
  abstract memory
  \\[0.5ex]
 \astore \in \astoreset & \bnfdef &
  \astore_1 \lsep \astore_2
  \bnfalt \lemp
  \bnfalt \fldpt{\aaddr}{\fld}{\aval}
  \bnfalt \astorepure
  &
  exact
  \\
  & \bnfalt &
  \achkreg
  \bnfalt
  \achkreg_1 \lreviwand \achkreg_2
  \hibox{\bnfalt \astoreview}
  &
  summaries, views
  \\[0.5ex]
  \apure \in \apuredom
  & & &
  data-value domain
  \\[0.5ex]
  \aview \in \aviewset & \bnfdef &
  \achkreg \lreviwand \amultichkreg
  \bnfalt \ltrue
  \bnfalt \aview_1 \ltrueandtrue \aview_2
  &
  validated views
  \\[0.5ex]
\end{grammar}
\\[7ex]
\begin{mathpar}[\MathparNormalpar]
  \begin{array}{@{}r@{\;}c@{\;}l@{}}
    \achkreg & \bnfdef & \achkercall
  \end{array}

  \begin{array}{@{}r@{\;}c@{\;}l@{}}
    \amultichkreg & \bnfdef &
    \amultichkreg_1 \lsep \amultichkreg_2
    \bnfalt \lemp
    \bnfalt \achkreg
  \end{array}

  \aaddr, \aval \in \avalset \quad\text{symbolic variables}

  \fld\quad\text{fields}

  \chker\quad\text{validation checker}
\end{mathpar}
\end{tabular}
}
\setlength{\fboxsep}{\saveboxsep}
\subfloat[Introducing validated views on unfolding.]{\label{fig:unfolding}
\begin{mathpar}
\inferrule[Unfold-Frame]{
  \astore \unfold \astore'
}{
  \astore_0 \lsep \astore \unfold \astore_0 \lsep \astore'
}

\inferrule[Unfold-View]{
  \achkercall \unfold \astore \lsep \amultichkreg
}{
  \achkercall \unfold \astoreview[(\astore \lsep \amultichkreg)][(\achkercall \lreviwand \amultichkreg)]
}

\inferrule[Unfold-SegView]{
  (\achkercall \lreviwand \achkreg) \unfold \astore \lsep \amultichkreg \lsep (\achkreg' \lreviwand \achkreg)
}{
  (\achkercall \lreviwand \achkreg) \unfold \astoreview[(\astore \lsep \amultichkreg \lsep (\achkreg' \lreviwand \achkreg))][(\achkercall \lreviwand (\amultichkreg \lsep \achkreg'))]
}
\end{mathpar}
}
\caption{Abstracting memory with shape-data constraints and validated
  views.}
\label{fig:abs-state}
\end{figure}


In \figref{abs-state}, we enrich a
separation-logic--based shape abstraction with validated
views.
The concrete semantic domains are a concrete memory $\cmem$
that maps addresses to values where we assume the set of addresses are
contained in values (i.e., $\caddrset \subseteq \cvalset$) and a
valuation $\cvalua$ that maps symbolic variables to concrete values.
Symbolic variables are existential variables
naming heap addresses and values; they form
the coordinates for the base pure, data-value domain $\apuredom$.  A
valuation $\cvalua$ assigns meaning to symbolic variables and thus
is the concrete domain element that corresponds to an abstract
element $\apure$ of the
data-value abstract domain (i.e.,
$\inconc[\apuredom]{\cvalua}{\apure}$).
\SQUEEZE{
Noticeably missing in these
domains is a notion of program variables.  We make program variables
an orthogonal concern by assuming concrete environments and abstract
environments that map program variables to concrete values or symbolic
variables, respectively.  In our diagrams, such as
\figref{code-example}, we indicate the abstract environment with
labels below the nodes.}  \PUNCH{Why symbolic variables and
  valuations.}

We consider an abstract memory $\amem$ to be a
disjunction of abstract stores $\astore$. Our focus is not on
the shape abstraction, but we briefly
describe one on which we can discuss the
validated view layer $\aview$.  Considering just the top line defining
abstract stores, we have an \emph{exact} abstraction of a finite
concrete memory constrained by the base data-value domain element
$\apure$ (i.e., analogous to two-valued structures in
TVLA~\cite{sagiv+2002:parametric-shape}).
The next line adds
summarization that enriches the abstraction to unbounded memory,
including inductive validation checker constraints $\achkreg \bnfdef
\achkercall$---corresponding to a region of memory from address
$\aaddr$ that satisfies validation checker $c$ with actual arguments
$\vec{\aval}$.  We write an overline $\vec{\cdot}$ for a sequence of
items.  We write $\achkreg_1 \lreviwand
\achkreg_2$ for a segment generalization of inductive predicate
$\achkreg_1$ to a hole at $\achkreg_2$, which is defined by syntactic unfolding of $\achkreg_1$
up to the hole $\achkreg_2$~\cite{chang+2008:relational-inductive}.
We can see an
abstract store as a separating shape
graph\citeverb{laviron+2010:separating-shape} by treating $\astore$ as a
finite map from base addresses $\aaddr$ to field edges
$\fldpt{\aaddr}{\fld}{\aval}$ or summary edges like $\achkercall$ or $\achkreg_1 \lreviwand
\achkreg_2$.
\PUNCH{The syntax of the memory abstraction and syntactic unfolding.}
The concrete domain element corresponding to an abstract store (without views) is
a set of pairs of a concrete memory $\cmem$ and a valuation $\cvalua$
(i.e., $\inconc{\cmemvalua}{\astore}$).
We assume any base shape abstraction of interest comes equipped with a concretization of this form.

Our extension is that an abstract store can be constrained by a view $\astoreview$
(shown shaded).  A validated view predicate $\aview\colon \achkreg \lreviwand \amultichkreg$ is an inductive segment with multiple potential endpoints $\amultichkreg  \bnfdef 
    \amultichkreg_1 \lsep \amultichkreg_2
    \bnfalt \lemp
    \bnfalt \achkreg$ (hence, we reuse the $\lreviwand$ connective).
The concretization of a view also yields a set of concrete
memory-valuation pairs.
That is, a view is a spatial formula---a memory abstraction itself.  The concretization of a $\aview\colon \achkreg \lreviwand \amultichkreg$ can be defined by syntactically unfolding $\achkreg$ until matching the endpoints in $\amultichkreg$.

\SQUEEZE{%
We discuss the final two forms of views $\aview$ further below, as
they are less relevant.
}%
A key point is that a validated view concretizes to a memory-valuation pair that constrains a \emph{sub-store} of the ``main'' abstract store (i.e., the shape graph).  In particular, the concretization of the constraining or restriction of an abstract store with a validated view $\astoreview$ is defined as follows:
\begin{definition}[Concretization of a View-Constrained Store]
\[ \begin{array}{l}
\inconc{\cmemvalua}{\astoreview} \\
\;\;\;\text{iff}\;\;\;
\text{$\inconc{\cmemvalua}{\astore}$ and $\inconc{\cmemvalua[\cmem']}{\aview}$ such that $\cmem' \subseteq \cmem$} \;.
\end{array} \]
\end{definition}
\noindent In other words, using the syntax and semantics of connectives from separation logic, we can see $\astoreview$ as being equivalent to $\astore \land (\aview \lsep \ltrue)$---we weaken $\aview$ to the strongest intuitionistic assertion weaker than $\aview$~\cite{reynolds2002:separation-logic:} before conjoining it with $\astore$.

\newsavebox{\SBoxValidatedViewConcretization}
\begin{lrbox}{\SBoxValidatedViewConcretization}\small
\(\begin{array}{r@{\;}l}
\aview & \colon \bstcall{\aaddr_1}{\symu_1}{\symu_6} \lreviwand (\bstcall{\aaddr_2}{\symu_2}{\symu_3} \lsep \bstcall{\aaddr_3}{\symu_4}{\symu_5})
\\
\cmem & \colon [\fldpt{\cvalua(\aaddr_1)}{\fldl}{\cvalua(\aaddr_2)},
\fldpt{\cvalua(\aaddr_1)}{\fldr}{\cvalua(\aaddr_3)},
\fldpt{\cvalua(\aaddr_1)}{\fldv}{\cvalua(\aval_1)}]
\\
\cvalua & \quad\text{s.t.}\quad
\cvalua(\symu_1) =
\cvalua(\symu_2) \leq
\cvalua(\symu_3) =
\cvalua(\aval_1) =
\cvalua(\symu_4) <
\cvalua(\symu_5) =
\cvalua(\symu_6)
\end{array}\)
\end{lrbox}

\begin{example}[Concretization of a Validated View]
For the following example validated view, we give one element of its concretization (i.e., $\inconc[\aviewset]{\cmemvalua}{\aview}$).  Other elements include memories with further unfoldings of the left or right sub-trees of $\cvalua(\aaddr_1)$: there are simply two holes somewhere because of the arbitrary allowable ranges given the symbolic variables $\symu_1$--$\symu_6$.  Observe that if there are some equality constraints on $\symu_1$--$\symu_6$, then that can have an effect on where the holes are allowed in the concretization.
\[
\scalebox{0.9}{\usebox{\SBoxValidatedViewConcretization}}
\]
\end{example}

While not strictly necessary, the remaining forms of views $\aview$ enable us to speak about all the view predicates as a whole applied to the ``main'' store.  The unit $\aview\colon \ltrue$ is the view that abstracts any concrete memory (as in separation logic), while $\aview_1 \ltrueandtrue \aview_2$ is the symmetric version of $\landtrue$.  More precisely, the concretization is
\[ \begin{array}{l}
\inconc{\cmemvalua[\cmem \uplus \cmem_1 \uplus \cmem_2]}{\aview_1 \ltrueandtrue \aview_2} \\
\quad\text{iff}\quad
\text{$\inconc{\cmemvalua[\cmem \uplus \cmem_1]}{\aview_1}$ and $\inconc{\cmemvalua[\cmem \uplus \cmem_2]}{\aview_2}$}
\end{array} \]
where we write $\uplus$ to union maps with disjoint domains, that is, the sub-store with common addresses in $\aview_1$ and $\aview_2$ must be the same $\cmem$.  Or in terms of separation logic, $\aview_1 \ltrueandtrue \aview_2$ corresponds to $(\aview_1 \lsep \ltrue) \land (\aview_2 \lsep \ltrue)$.  Note that we can always weaken any store $\astore$ to apply views at the top-level (i.e., $(\astore_1 \landtrue \aview_1) \lsep (\astore_2 \landtrue \aview_2)$ implies $(\astore_1 \lsep \astore_2) \landtrue (\aview_1 \ltrueandtrue \aview_2)$ as desired).

\paragraph{Introducing and Eliminating Views.}
The key difference between a validated view $\aview$ and a ``main''
store abstraction $\astore$ is how they are treated in static analysis.
While inductive summaries in the main store are unfolded to
materialize fields for strong updates and then reintroduced to
over-approximate loop invariants, a validated view on the store is auxiliary
information that is dropped when potentially invalidated (somewhat
like a predicate abstraction~\cite{DBLP:conf/cav/GrafS97}).  This
distinction is why we can use inductive multi-segments for views even
though we expect them to be very difficult to work with in the ``main
store.''

\newsavebox{\SBoxUnfoldBstExample}
\begin{lrbox}{\SBoxUnfoldBstExample}\small
\(\begin{array}{r@{\;\;}c@{\;\;}l}
\bstcall{\aaddr}{\symu_1}{\symu_2} & \unfold &
(\fldpt{\aaddr}{\fldl}{\aaddr_1} \lsep \fldpt{\aaddr}{\fldr}{\aaddr_2} \lsep \fldpt{\aaddr}{\fldv}{\aval} \\
& & \quad\lsep \bstcall{\aaddr_1}{\symu_1}{\aval} \lsep \bstcall{\aaddr_2}{\aval}{\symu_2} \land \symu_1 \leq \aval < \symu_2)
\\
\bstcall{\aaddr}{\symu_1}{\symu_2} & \unfold & \lemp \land \aaddr = \nullval
\end{array}\)
\end{lrbox}

%
%
%
For an inductive validation checker instance
$\achkreg$ like $\bstcall{\aaddr}{\symu_1}{\symu_2}$, the corresponding inductive definition yields an
unfolding relation $\achkreg \unfold\astore$ for each disjunctive case 
of the inductive
definition while replacing formal parameters with actual arguments
(cf., \figref{bst-unfoldfold} in \secref{static-analysis}).
For example, the following are in the unfolding relation for $\bst$:
\[
\scalebox{0.9}{\usebox{\SBoxUnfoldBstExample}}
\]
Similarly, we assume an unfolding relation for inductive segment instances
$\achkreg' \lreviwand \achkreg \unfold\astore$.  Given these axioms, we can then lift this syntactic unfolding relation to stores $\astore \unfold \astore'$ by non-deterministically selecting a summary to
unfold (\TirName{Unfold-Frame} in \figref{unfolding}).

A view constraint is introduced on unfolding a summary in the store to strengthen the resulting 
state.  We define these strengthenings with the remaining two rules in \figref{unfolding}.  Rule \TirName{Unfold-View} is for unfolding an inductive predicate where the endpoints $\amultichkreg$ are given by the recursive validation checkers in the definition of $\chker$.  For unfolding $\bst$ in the above, the endpoints $\amultichkreg$ would correspond to the two recursive calls $\bstcall{\aaddr_1}{\symu_1}{\aval} \lsep \bstcall{\aaddr_2}{\aval}{\symu_2}$ in the non-$\nullval$ case.  The unfolding of a segment is similarly defined in \TirName{Unfold-SegView}: it can result in a new segment whose begin point is $\achkreg'$ that becomes an endpoint of the view.
%
On an update, our abstract transformer simply and soundly drops any view that constrains the
updated field.  For example, on \reftxt{line}{line:post-swap} in \figref{excisemin}, there are no remaining views
because of the update to \code{p->l} and the dropping of $\aaddr_1$.
From the perspective of the graph
representation, it is straightforward to track the set of views
that constrain each field edge. \PUNCH{Transfer functions.}

\paragraph{Inclusion, Join, and Widen.}

To define the standard abstract domain operations, such as the inclusion
$\sqsubseteq$, join $\AIjoin$, and widen $\AIwiden$ operators, we need a correspondence
between the symbolic variables (i.e., the set of names) used in each abstract element.
 We call such a correspondence
a \emph{valuation transformer}
\[
\avaluatr : \avalset \finitemap \avalset \;.
\]
This map gives an instantiation of
existentials in one element in terms of existentials in the other.

We assume the shape-data abstraction (without
views) comes equipped with an abstract inclusion relation
\[\jinclusion{\avaluatr}{\astore}{\astore'}\] that over-approximates
inclusion under concretization with the variables in
$\astore'$ renamed to the variables in $\astore$ using $\avaluatr$.  Formally, \[ \AIconcfn(\astore) \subseteq \SetST{ \cmemvalua }{
  \cmemvalua[\cmem][\cvalua \circ \avaluatr] \in \AIconcfn(\astore')
} \;. \]
The setup is analogous for the abstract join operation $$\jjoin{\avaluatr_1}{\avaluatr_2}{\astore_1}{\astore_2}{\astore'} \;,$$
except that we need a pair of valuation transformers $\avaluatr_1, \avaluatr_2$ mapping the
variables in the result element to each of the input elements.
\SQUEEZE{%
Specifically, we assume the join in the combined shape-data abstraction over-approximates union under concretization and renaming of the result (i.e., formally,
$\AIconcfn(\astore_1) \subseteq \SetST{ \cmemvalua }{
  \cmemvalua[\cmem][\cvalua \circ \avaluatr_1] \in \AIconcfn(\astore')
}$ and
$\AIconcfn(\astore_2) \subseteq \SetST{ \cmemvalua }{
  \cmemvalua[\cmem][\cvalua \circ \avaluatr_2] \in \AIconcfn(\astore')
}$).  Similarly,
as is standard, widen $\AIwiden$ should be a sound join with the additional convergence property.
\PUNCH{Valuation transformers and join/widen in base domains.}
}%

For inclusion, joins, and widen of validated views, we treat a view predicate as
uninterpreted, so the domain structure is straightforward.  Treating a view $\aview$ as a finite set of view predicates interpreted conjunctively with $\ltrueandtrue$, abstract inclusion is the set inclusion in the reverse direction among the set of view predicates up to renaming with the valuation transformer $\avaluatr$; $\ltrue$ is the top element; join is set intersection up to renaming; widen can be join because it satisfies the ascending chain property.

For the overall abstract domain combining shape-data abstraction with views,
we define these operations point-wise
\begin{mathpar}
\inferrule{
  \jjoin{\avaluatr_1}{\avaluatr_2}{\astore_1}{\astore_2}{\astore'} \\
  \jjoin{\avaluatr_1}{\avaluatr_2}{\aview_1}{\aview_2}{\aview'}
}{
\jjoin{\avaluatr_1}{\avaluatr_2}{\astoreview[\astore_1][\aview_1]}{\astoreview[\astore_2][\aview_2]}{\astoreview[\astore'][\aview']}
}
\end{mathpar}
The key piece is that they share the same valuation transformers: that is, the join of the view domain additionally constrains the valuation transformers, which
is the crucial precision needed
to connect the summaries around a hole (cf., the allowable range for $\bst$ marked with $\spadesuit$ in \figref{excisemin}).


\paragraph{Reductions and Precision.}
The validated view domain is designed to compensate for inductive
imprecision in the base data-value domain by enabling new
reductions~\cite{cousot+1979:systematic-design} (i.e., exchange of
information) between the pointer-shape, data-value, and view
components of an abstract store $\astore$. A view $\achkercall
\lreviwand \amultichkreg$ can be used to fold the shape from $\aaddr$
provided $\amultichkreg$ can be shown to hold by either constraints in
the shape or view components regardless of the data-value constraints.
Also, important equalities connecting
parameters in unmodified, ``adjacent'' regions of the shape graph
(cf., \figref{bst-traverse}) are easily maintained by the view domain.
These equalities are
typically difficult to derive when widening shape abstractions because
the summarization of regions are considered independently in a
separation logic-based analysis.
While the design of the
validated view domain is to compensate for imprecise data-value
domains, on the flip side, reduction can be applied to derive a view
$\aview$ when the pointer-shape and data-value constraints in
$\astore$ can be shown to imply it (i.e., $\astore \sqsubseteq
\aview$).  This could be useful to ``save'' information in a view from
data-value constraints derived in straight-line code that may be lost
on widening.  \PUNCH{Reduction.}

Overall, the validated view domain remembers all
prog\-ramm\-er-asserted inductive validations of the heap that have
not been possibly invalidated by a heap write. 
As we will see (\secref{dynamic-validation}), 
we can synthesize calls to short-circuiting checkers
that use these remembered validated views \emph{as hypotheses}
(cf.~\figref{validation:bstsc}) and thus avoid checking parts
of the heap that have not changed with respect to the data
structure invariant.

\section{Computing Programmatic Valuations}
\label{app:instrumentation-formal}


We express a \emph{programmatic
  valuation} as a finite map $$\pvalua : \avalset \finitemap \exprset \times \plocset
\times \Nat\;.$$
A $\pvalua$ is flow-insensitive (i.e., is the same within in a lexical scope) and
maps from symbolic
variables to a triple of (1) a program expression $e \in \exprset$ involving the program variables $x$ or
symbolic variables $\aval$; (2) a program location $\ell \in \plocset$ for where the SSA-assignment should be instrumented; and (3) a disjunction index $i$.
Because
an abstract memory $\amem$ is a disjunction of stores $\bigvee_i\;\astore_i$, we
assume every store $\astore$ is indexed with a globally unique number
given by $\disj(\astore)$.  At any location $\ell$ where we need to apply instrumentation, if
there is more than one disjunct in the inferred abstract memory $\amem$ at $\ell$ then we generate
generate an \code{if}-\code{else}
on a special instrumentation variable that 
stores the disjunction of interest from $\disj(\astore)$.  
Thus, we can assume we are working with single stores $\astore$.

\begin{figure}\scriptsize\centering
\begin{mathpar}
\begin{array}{@{}l@{\quad}l@{\;}l@{}}
\text{program commands} &
\multicolumn{2}{@{}l@{}}{
c
\hfill
\text{statements}\quad s \bnfdef \ell\colon \pcmd \bnfalt s_1; s_2 \bnfalt \cdots
}
\\[0.5ex]
\text{annotated commands} & \pcmd
\bnfdef &
c
\bnfalt \pcassume{\astore}{ \icall{\varx}{\chker}{\vec{\vary}} }
\bnfalt \pcunfold{\astore}{\achkreg} \\
&& \bnfalt \pcjoin{\astore}{\avaluatr_1, \astore_1}{\avaluatr_2, \astore_2}
\end{array}
\\
\inferrule[]{
  \achkercall \in \astore
  \\
  \aaddr = \astore(\aenv(\varx))
  \\
  \vec{\aval} = \astore(\aenv(\vec{\vary}))
  \\
  \disji = \disj(\astore)
}{
  \jpvaluacmd{
    \pvalua
    [\aaddr \mapsto \tuple{x, \ploc, \disji}]
    [\vec{\aval} \mapsto \tuple{\vec{\vary}, \ploc, \disji}]
  }{
    \ploccmd[\ploc][
      \pcassume{ \astore }{ \icall{\varx}{\chker}{\vec{\vary}} }
    ]
  }
}

\inferrule[]{
  \disji = \disj(\astore)
}{
  \jpvaluacmd{
    \pvalua[
    \aval \mapsto \tuple{\expr, \ploc, \disji}
    \mid
    \expr \in \bind(\achkreg \unfold \astore)
    ]
  }{ 
    \ploccmd[\ploc][
      \pcunfold{\astore}{\achkreg}
    ]
  }
}

\inferrule[]{
  \pvalua = 
    \pvalua'[
    \aval \mapsto \tuple{\pejoin{\aval_1}{\aval_2}, \ploc, \disji}
    \mid
    \aval \mapsto \aval_1 \in \avaluatr_1
    \;\text{and}\;
    \aval \mapsto \aval_2 \in \avaluatr_2
    ]
    \quad
  \disji = \disj(\astore)
}{
  \jpvaluacmd{
    \pvalua
  }{
    \ploccmd[\ploc][
      \pcjoin{\astore}{\avaluatr_1, \astore_1}{\avaluatr_2, \astore_2}
    ] 
  }
}

\inferrule[]{
}{
  \jpvaluacmd{
    \pvalua
  }{
    \ploccmd[\ploc][
      c
    ] 
  }
}

\inferrule[]{
  \jpvaluacmd{
    \pvalua
  }{
      s_1
  }
  \\
    \jpvaluacmd{
    \pvalua
  }{
      s_2
  }
}{
  \jpvaluacmd{
    \pvalua
  }{
      s_1; s_2
  }
}

\cdots 
\end{mathpar}
\caption{Computing the programmatic valuation $\pvalua$.}
\label{fig:logicvar-instrumentation}
\end{figure}

Assuming the static analysis computes an
invariant map from each program location $\ploc$ to an
abstract memory $\amem$ that over-approximates the set of possible
concrete memories at $\ploc$, we can express the analysis proof by annotating the
program with these inferred invariants $\amem$%
\SQUEEZE{ (or individually for each store $\astore$)}.
For
computing the programmatic valuation, we do not need the annotated invariants at all program
locations, but rather only in the cases where the analysis introduced new symbolic variables.
To make these locations explicit,
we augment the command language $c$ from the underlying programming language (e.g., C) with a
few annotated commands $\pcmd$ to express these analysis operations.  The control-flow of the program
remains unchanged
(e.g., the statement language $s$ in \figref{logicvar-instrumentation} has commands
$\pcmd$ labeled with a program location $\ell$, sequencing, etc.).

%
%

In \figref{logicvar-instrumentation}, we give the annotated commands of interest,
which record the static analysis operations that introduce new symbolic
variables (except memory allocation, whose instrumentation is straightforward).
These commands are annotated with the inferred post-state.
The
$\pcassume{\astore}{ \icall{\varx}{\chker}{\vec{\vary}} }$ command
indicates that the analysis assumed a data structure validation
check $\icall{\varx}{\chker}{\vec{\vary}}$
yielding a  post-store
$\astore$ where $\varx, \vec{\vary}$ are original program
variables.
The $\pcunfold{\astore}{\achkreg}$ command makes explicit the
unfolding of an inductive
predicate $\achkreg$;
and the
$\pcjoin{\astore}{\avaluatr_1, \astore_1}{\avaluatr_2, \astore_2}$
command indicates a join of two input stores $\astore_1, \astore_2$ along with
their corresponding valuation transformers $\avaluatr_1, \avaluatr_2$.  

We describe the generation of a $\pvalua$
via the
judgment form $\jpvaluacmd{\pvalua}{\ploccmd}$, which says that
$\pvalua$ is a programmatic valuation for analysis command $\pcmd$ at
program location $\ploc$.
The analogous judgment form $\jpvaluacmd{\pvalua}{s}$ simply walks over the program structure to constrain $\pvalua$.  The definition of this judgment is sketched in the last line of \figref{logicvar-instrumentation} with the one rule for statement sequencing $s_1; s_2$, and we can see these rules as the checking system for a flow-insensitive fixed-point computation starting from the empty map.
\SQUEEZE{(Note that the output of this generation is $\pvalua$.)}%
We assume an implicit, fixed parameter to
this judgment $\aenv$ for an abstract environment mapping
program variables to symbolic variables for the current static scope.
For simplicity, we let the
symbolic variables in the abstract environment denote the address of
the corresponding program variables and let the abstract store hold
the value of the variable.

For $\pcassumekw$, we
``initialize'' the programmatic valuation $\pvalua$ by finding a
correspondence between
the validation check $\icall{\varx}{\chker}{\vec{\vary}}$
and a shape fact $\achkercall$  in terms of program variables.
We add bindings to $\pvalua$
mapping the symbolic
variables $\aaddr$ and $\vec{\aval}$ in the inductive predicate
to the program variable arguments $\varx$ and $\vec{\vary}$, respectively.
The mapping also remembers the program location $\ell$ and the store disjunct $i$.
%
On $\pcunfoldkw$ing an
inductive predicate $\achkercall$, we need to extend the programmatic valuation $\pvalua$ with
bindings for the existentials introduced in the definition of
$\chker$.  Because all inductive definitions used in
our static shape analysis are derived from executable specifications (e.g.,
static verification predicate $\bst$ from dynamic validation checker \code[highlighting]{|bst|}),
program expressions must
exist that witness an instantiation for each existential.
In other words, any existential introduced in an inductive definition
must correspond to the value of an argument or from reading a field.  We write
this instantiation for an unfolding as $\bind(\achkreg \unfold
\astore)$.  For example, the $\bind$ function for the second disjunct
in $\bst$ from \figref{bst-sepchecker} is $[\symv \mapsto
(\symt\texttt{->}\fldv), \syml \mapsto (\symt\texttt{->}\fldl), \symr
\mapsto (\symt\texttt{->}\fldr)]$\SQUEEZE{ (assuming $\symt$ is the actual
argument and $\symv, \syml, \symr$ are the fresh symbolic variables
chosen for the existentials)}.  The instrumentation at \reftxt{line}{line-assumetough} and \reftxt{line}{line:instunfoldguard} in \figref{excisemin} come from applying this rule.

Finally, we consider joins
$\pcjoin{\astore}{\avaluatr_1, \astore_1}{\avaluatr_2, \astore_2}$.
%
%
Our convention is that the set of symbolic variables in $\astore$ are
fresh with respect to $\astore_1$ and $\astore_2$, so we must extend
the programmatic valuation $\pvalua$ with mappings for these new
variables.  In particular, we should map these to the
corresponding symbolic variable from each input store, which is
precisely given by the valuation transformers $\avaluatr_1$
and $\avaluatr_2$.  We express this mapping using the standard
static single assignment $\pejoin{\cdot}{\cdot}$ function.
The key observation here is that the precision of the programmatic
valuation $\pvalua$ after a join point is determined by the valuation
transformers computed by the static analysis.

\section{Shortify Soundness}
\label{app:shortify-soundness}

In this section, we state in more detail the shortify soundness theorem from \secref{shortify} that relates the statically-proven store with the assertion store to synthesize the short-circuiting validation store.

\begin{figure}\footnotesize
\begin{tabularx}{\linewidth}{RcX}
\inconc[\afinstoreset]{\cmemvalua}{ \ltrue }
& &
for all $\cmem, \cvalua$
\\
\inconc[\afinstoreset]{\cmemvalua[{[]}]}{ \lemp }
& &
for all $\cvalua$
\\
\inconc[\afinstoreset]{\cmemvalua[{[\fldpt{\cvalua(\aaddr)}{\fld}{\cvalua(\aval)}]}]}{ \fldpt{\aaddr}{\fld}{\aval} }
& &
for all $\cvalua$
\\[1ex]
\end{tabularx}
\begin{tabularx}{\linewidth}{RcX}
\inconc[\afinstoreset]{\cmemvalua}{ \astore_1 \lsep \astore_2 } 
& iff &
$\cmem = \cmem_1 \uplus \cmem_2$ such that
$\inconc[\afinstoreset]{\cmemvalua[\cmem_1]}{ \astore_1 }$ and
$\inconc[\afinstoreset]{\cmemvalua[\cmem_2]}{ \astore_2 }$ for some $\cmem_1, \cmem_2$
\\[0.5ex]
\inconc[\afinstoreset]{\cmemvalua}{ \astorepure }
& iff &
$\inconc[\afinstoreset]{\cmemvalua}{ \astore }$ and
$\inconc{\cvalua}{ \apure }$
\end{tabularx}\small
\begin{mathpar}
\hfill\fbox{$\inconc{\cmemvalua}{\astore}$}
\\
\inferrule[CStore-Simp]{
  \astore \in \afinstoreset
  \\
  \inconc[\afinstoreset]{\cmemvalua}{ \astore }
}{
  \inconc{\cmemvalua}{ \astore }
}

\inferrule[CStore-Unfold]{
  \astore \unfold \astore'
  \\
  \inconc{\cmemvalua}{ \astore' }
}{
  \inconc{\cmemvalua}{ \astore }
}
\end{mathpar}
\caption{Concretization of abstract stores.}
\label{fig:abs-concretization}
\end{figure}

\paragraph{Concretization.}

In \figref{abs-concretization}, we define the concretization of
abstract stores $\inconc{\cmemvalua}{\astore}$. We distinguish \emph{simple stores} that do not have any inductive summaries and write $\afinstoreset$ for such simple stores.
In the top portion of the figure, we give a standard
concretization function $\AIconcfn_{\afinstoreset} :
\afinstoreset \rightarrow \cmemset \times \cvaluaset$ that goes from
simple store abstractions to concrete memory-valuation pairs 
following separation logic~\cite{reynolds2002:separation-logic:}.

We then assume a
syntactic unfolding relation that materializes fields from summary
constraints, such as $\achkreg$.  For an inductive validation checker
$\achkreg$ like $\bst$,
the unfolding relation $\achkreg \unfold\astore$ unfolds the inductive
definition while replacing formal parameters with actual arguments
(cf., \figref{bst-unfoldfold}). We write $\achkreg_1 \lreviwand
\achkreg_2$ for a segment generalization of inductive predicate
$\achkreg_1$ to a hole at $\achkreg_2$ whose syntactic unfolding can
be similarly defined~\cite{chang+2008:relational-inductive}. The
specific forms of inductive summaries are not crucial and can vary,
provided that they come equipped with a syntactic unfolding relation.
We can then lift this syntactic unfolding relation to stores $\astore
\unfold \astore'$ by non-deterministically selecting a summary to
unfold while framing the rest of the store.

Finally with syntactic unfolding, we define the concretization of a
separating shape abstraction as either the
concretization of a fully-unfolded, simple store abstraction (\TirName{CStore-Simp}) or in
terms of the concretization of an unfolding (\TirName{CStore-Unfold}).
The concretization of abstract stores $\inconc{\cmemvalua}{\astore}$
is the least relation satisfying the given inference rules.

\newsavebox{\SBoxKscFormal}
\begin{lrbox}{\SBoxKscFormal}\small
\(\begin{array}{r@{\;}c@{\;}cl}
\icall{\aaddr}{\ichk{\chker}}{\vec{\aval'}, \vec{\aval}}
& \defeq &      & \lemp \;\land\; \vec{\aval'} = \vec{\aval} \\
&        & \lor & \Lor_i\; \left(\astore_i'' \;\land\; \vec{\aval'} \neq \vec{\aval}\right)
\quad \text{where $\jsubtractsyn{\astore_i}{\icall{\aaddr}{\chker}{\vec{\aval'}}}{\astore_i''}$} \\
& & &
  \qquad\text{for all $i$ s.t.
  $\icall{\aaddr}{\chker}{\vec{\aval}}
  \unfold
  \astore_i$}
\end{array}\)
\end{lrbox}

\paragraph{Short-Circuiting Template.}

We give a more direct version of the short-circuiting checker template \figref{sctemplate} by unfolding the pre-condition and inlining the static shape analysis. Recall that the proof strategy for statically analyzing the template outlined at the end of \secref{dynamic-validation} was by an unfolding of the pre-condition.
\begin{definition}[Short-Circuiting Invariant Checker Synthesis]
For an inductive validation checker $\chker$, a short-circuiting variant is defined as follows:
\[
\scalebox{0.9}{\usebox{\SBoxKscFormal}}
\]
\end{definition}
The first disjunct corresponds to the short-circuiting condition, while the latter disjuncts correspond to each disjunct in the definition of the original inductive validation checker $\chker$.
We restrict the form of inductive checkers $\achkercall$ to be disjoint disjunctions (i.e., if-then-else) of pure facts and separating conjunctions points-to $\fldpt{\aaddr}{\fld}{\aval}$ and inductive summaries $\achkreg$. As a consequence, we assume that inductive checkers are \emph{precise}, that is, for any concrete memory $\cmem$, there is only one sub-memory that is in the concretization of $\achkercall$ (for a given valuation $\cvalua$). And so each disjunct of the original inductive validation checker $\chker$ corresponds to a case of the short-circuiting template $\ichk{\chker}$.

%
%
%
%
%

We now consider shortify soundness (\thmref{shortify-soundness-main}) from \secref{shortify}. To prove this proposition, we first generalize the induction hypothesis to allow for a frame $\astore_0$ in the statically-proven invariant (that is not needed to decide how to synthesize the short-circuiting validation).
\newcommand{\unf}[1]{#1_{\textrm{u}}}
\begin{theorem}[Shortify Soundness]\label{thm:shortify-soundness-generalized}
  If $\inconc{\cmemvalua[\cmem_0 \uplus \cmem][\cvalua]}{\astore_0 \lsep \astore}$
  with $\inconc{\cmemvalua[\cmem][\cvalua]}{\astore}$
  and $\jsubtractsyn{\astore}{\astore'}{\astore''}$ and
  $\inconc{\cmemvalua[\cmem''][\cvalua]}{\astore''}$
  where $\cmem'' \subseteq \cmem$ ,
  then $\inconc{\cmemvalua[\cmem'][\cvalua]}{\astore'}$ for some $\cmem' \subseteq \cmem$.
\end{theorem}
\begin{proof}
By induction on the derivation $\fsP$ of $\inconc{\cmemvalua[\cmem_0 \uplus \cmem][\cvalua]}{\astore_0 \lsep \astore}$ and the derivation $\fsS$ of $\jsubtractsyn{\astore}{\astore'}{\astore''}$ using the lexicographic order of $\fsP$ followed by $\fsS$.

Assuming $\fsP$, we consider the cases of the shortify derivation $\fsS$. The particularly interesting case is for when $\fsS$ follows from \TirName{Shortify-Inductive}.
\small\begin{case}\(
\fsS =
\inferrule*[Right=Shortify-Validate]{
}{
  \jsubtractsyn{\astore}{\astore'}{\astore'}
}\) \\
Immediate.
\end{case}
\begin{case}\(
\fsS =
\inferrule*[Right=Shortify-Proven]{
  \astore \sqsubseteq \astore' \lsep \ltrue
}{
  \jsubtractsyn{\astore}{\astore'}{\ltrue}
}
\) \\
Given $\astore \sqsubseteq \astore' \lsep \ltrue$, we have that $\AIconcfn(\astore) \subseteq \AIconcfn(\astore' \lsep \ltrue)$, so $\inconc{\cmemvalua[\cmem][\cvalua]}{\astore' \lsep \ltrue}$. By the concretization of $\lsep$, we have that $\cmem = \cmem' \uplus \cmem_1$ for some $\cmem', \cmem_1$ such that $\inconc{\cmemvalua[\cmem'][\cvalua]}{\astore'}$
	and $\inconc{\cmemvalua[\cmem_1][\cvalua]}{\ltrue}$.
\end{case}
\begin{case}\(
\fsS =
\inferrule*[Right=Shortify-Emp]{
  \apure \land \apure'' \sqsubseteq \apure'
}{
  \jsubtractsyn{\lemp \land \apure}{\lemp \land \apure'}{\lemp \land \apure''}
}\) \\
Given $\apure \land \apure'' \sqsubseteq \apure'$, we have that $\AIconcfn(\apure) \intersect \AIconcfn(\apure'') \subseteq \AIconcfn(\apure')$. From the hypotheses, we see that $\cvalua \in \AIconcfn(\apure) \intersect \AIconcfn(\apure'')$, so $\inconc{\cvalua}{\apure'}$ and thus $\inconc{\cmemvalua[{[]}][\cvalua]}{\lemp \land \apure'}$. Trivially $[] \subseteq []$.
\end{case}
\begin{case}\(
\fsS =
\inferrule*[Right=Shortify-Unfold]{
  \deriv{\fsU}{\astore' \unfold \astore'''}
  \\
  \deriv{\fsS_1}{\jsubtractsyn{\astore}{\astore'''}{\astore''}}
}{
  \jsubtractsyn{\astore}{\astore'}{\astore''}
}\) \\
By the i.h. on $\fsP$ and $\fsS_1$, we have that $\inconc{\cmemvalua[\cmem'][\cvalua]}{\astore'''}$ for some $\cmem' \subseteq \cmem$. Thus, we can construct the derivation
\[
\inferrule*[right=CStore-Unfold]{
  \deriv{\fsU}{\astore' \unfold \astore'''}
  \\
  \inconc{\cmemvalua[\cmem'][\cvalua]}{ \astore''' }
}{
  \inconc{\cmemvalua[\cmem'][\cvalua]}{ \astore' }
}
\]
\end{case}
\begin{case}\(
\fsS =
\inferrule*[Right=Shortify-Sep]{
  \deriv{\fsS_1}{\jsubtractsyn{\astore_1}{\astore_1'}{\astore_1''}}
  \\
  \deriv{\fsS_2}{\jsubtractsyn{\astore_2}{\astore_2'}{\astore_2''}}
}{
  \jsubtractsyn{\astore_1 \lsep \astore_2}{\astore_1' \lsep \astore_2'}{\astore_1'' \ltrueandtrue \astore_2''}
}\) \\
Given $\inconc{\cmemvalua[\cmem''][\cvalua]}{\astore_1'' \ltrueandtrue \astore_2''}$, we have that $\cmem'' = \cmem_1'' \union \cmem_2''$ for some $\cmem_1'', \cmem_2''$ (that need not be disjoint)
such that
$\inconc{\cmemvalua[\cmem_1''][\cvalua]}{\astore_1''}$ and
$\inconc{\cmemvalua[\cmem_2''][\cvalua]}{\astore_2''}$.
Given $\inconc{\cmemvalua[\cmem][\cvalua]}{\astore_1 \lsep \astore_2}$, we have that $\cmem = \cmem_1 \uplus \cmem_2$ for some $\cmem_1, \cmem_2$ such that $\inconc{\cmemvalua[\cmem_1][\cvalua]}{\astore_1}$ and $\inconc{\cmemvalua[\cmem_1][\cvalua]}{\astore_1}$.
By the i.h. on $\fsP$ and $\fsS_1$ with $\inconc{\cmemvalua[\cmem_1''][\cvalua]}{\astore_1''}$, we have that $\inconc{\cmemvalua[\cmem_1'][\cvalua]}{\astore_1'}$ for some $\cmem_1' \subseteq \cmem_1$.
Analogously by the i.h. on $\fsP$ and $\fsS_2$, we have that $\inconc{\cmemvalua[\cmem_2'][\cvalua]}{\astore_2'}$ for some $\cmem_2' \subseteq \cmem_2$. Since $\cmem_1$ and $\cmem_2$ have disjoint domains, $\cmem_1'$ and $\cmem_2'$ must have disjoint domains. Thus $\inconc{\cmemvalua[\cmem_1' \uplus \cmem_2'][\cvalua]}{\astore_1' \lsep \astore_2'}$ and $\cmem_1' \uplus \cmem_2' \subseteq \cmem$.
\end{case}
\begin{case}
\(\fsS =\)
\begin{center}
\(\inferrule*[right=Shortify-Inductive]{
}{
  \jsubtractsyn{ \achkercall }{ \icall{\aaddr}{\chker}{\vec{\aval'}} }{
    \icall{\aaddr}{\ichk{\chker}}{\vec{\aval'}, \vec{\aval}}
  }
}\)
\end{center}
Because we construct $\ichk{\chker}$ so that $\icall{\aaddr}{\ichk{\chker}}{\vec{\aval'}, \vec{\aval}}$ is precise, we have that
\begin{mathpar}
\fsA =
\inferrule{
  \icall{\aaddr}{\ichk{\chker}}{\vec{\aval'}, \vec{\aval}} \unfold \vec{\aval'} = \vec{\aval}
  \quad
  \deriv{\fsA_1}{\inconc{\cmemvalua[\cmem'']}{ \vec{\aval'} = \vec{\aval} }}
}{
  \inconc{\cmemvalua[\cmem'']}{ \icall{\aaddr}{\ichk{\chker}}{\vec{\aval'}, \vec{\aval}} }
}
\\
\text{or}
\\
\fsA =
\inferrule{
  \icall{\aaddr}{\ichk{\chker}}{\vec{\aval'}, \vec{\aval}} \unfold \unf\astore'' \land \vec{\aval'} \neq \vec{\aval}
  \quad
  \deriv{\fsA_1'}{\inconc{\cmemvalua[\cmem'']}{ \astore \lsep \unf\astore'' \land \vec{\aval'} \neq \vec{\aval} }}
}{
  \inconc{\cmemvalua[\cmem'']}{ \icall{\aaddr}{\ichk{\chker}}{\vec{\aval'}, \vec{\aval}} }
}
\end{mathpar}
for some $\unf\astore''$.

In the first subcase, from $\inconc{\cmemvalua[\cmem'']}{ \vec{\aval'} = \vec{\aval} }$ and the hypothesis $\inconc{\cmemvalua[\cmem][\cvalua]}{\achkercall}$, we have that $\inconc{\cmemvalua[\cmem][\cvalua]}{\icall{\aaddr}{\chker}{\vec{\aval'}}}$.

In the second subcase, we now consider the concretization of $\achkercall$. The \TirName{CStore-Simp} cannot apply because $\achkercall$ is an inductive summary. Thus, it must be case that
\[
\fsP =
\inferrule*[right=CStore-Unfold]{
  \achkercall \unfold \unf\astore
  \\
  \deriv{\fsP_1}{\inconc{\cmemvalua}{ \unf\astore }}
}{
  \inconc{\cmemvalua}{ \achkercall }
}
\]
such that 
$\fsS'::\;\jsubtractsyn{ \unf\astore }{ \icall{\aaddr}{\chker}{\vec{\aval'}} }{ \unf\astore'' }$.
By the i.h. on $\fsP_1$ with $\fsS'$  and $\lemp$ for $\astore_0$, we have that
$\inconc{\cmemvalua[\cmem'][\cvalua]}{\icall{\aaddr}{\chker}{\vec{\aval'}}}$ for some $\cmem' \subseteq \cmem$.
\end{case}
\end{proof}

\section{Hash Trie Benchmark}
\label{app:hashtrie}

In this section, we describe some additional details about the hash trie benchmark.
The hash trie data structure maps keys to values. To store this mapping, a key
is hashed, and the bit blocks of the result are used to traverse down the
trie to the storing node. Hash collisions are resolved by rehashing at certain
levels to provide infinite bit blocks. We consider this data structure as a
significant challenge for static verification, as it involves arbitrary hash functions,
non-trivial rehashing, and bit-wise concatenation of the rehashed keys.
\Tool~demonstrates the synergy with dynamic validation for these constraints.

\newsavebox{\SBoxHTCh}
\begin{lrbox}{\SBoxHTCh}\scriptsize\lstnonum
\begin{lstlisting}[language=C,alsolanguage=Spec,alsolanguage=SqueezeOps]
val hashblocks(val key, val level) {
  val blocks = 0, hkpart = 0, hk = hash(key, 0);
  for (val l = 0; l < level; l++, hkpart += 2) {
    if (hkpart >= 2*12) hk = hash(key, l), hkpart = 0;
    blocks = (blocks << 2) | ((hk >> hkpart) & 0x3);
  }
  return blocks;
}

bool hashtrie(ht t, val level, val path) {
  val k, dv lv;
  return t == NULL ? true :
   (k = t->k, v = t->v, lv = level + 1, p00 = path << 2,
    hashblocks(k, level) == path &&
    hashtrie(t->blk00, lv, p00) &&
    hashtrie(t->blk01, lv, p00 + 0x1) &&
    hashtrie(t->blk10, lv, p00 + 0x2) &&
    hashtrie(t->blk11, lv, p00 + 0x3));
}
\end{lstlisting}
\end{lrbox}
\setlength{\fboxsep}{\saveboxsep}

\newsavebox{\SBoxSepHTCheckerExample}
\begin{lrbox}{\SBoxSepHTCheckerExample}\footnotesize
\( \begin{array}{@{}l@{}}
 \icall{\symt}{\hashtrie}{\sym{level}, \sym{path}}
    \defiff \lemp \spland \symt = \nullval\\[0.5ex]
\;\;\lor\; \exists \sym{k},\sym{d}, \sym{lv}, \sym{blk00}, \sym{blk01}, \sym{blk10}, \sym{blk11}, \sym{p00}, \sym{p01}, \sym{p10}, \sym{p11}. \\
      \qquad
      \fldpt{ \symt }{ \fmtfld{k} }{ \sym{k} }
      \splsep
      \fldpt{ \symt }{ \fldv }{ \symv }
      \splsep \\
      \qquad
      \fldpt{ \symt }{ \fmtfld{blk00} }{ \sym{blk00} }
      \splsep
      \fldpt{ \symt }{ \fmtfld{blk01} }{ \sym{blk01} }
      \splsep \\
      \qquad
      \fldpt{ \symt }{ \fmtfld{blk10} }{ \sym{blk10} }
      \splsep
      \fldpt{ \symt }{ \fmtfld{blk11} }{ \sym{blk11} }
      \splsep \\
      \qquad
      \icall{\sym{blk00}}{\hashtrie}{{\sym{lv},\,\sym{p00}}}
      \splsep
      \icall{\sym{blk01}}{\hashtrie}{{\sym{lv},\,\sym{p01}}}
      \splsep \\
      \qquad
      \icall{\sym{blk10}}{\hashtrie}{{\sym{lv},\,\sym{p10}}}
      \splsep
      \icall{\sym{blk11}}{\hashtrie}{{\sym{lv},\,\sym{p11}}}
      \\ \qquad
      \spland
      \symt \neq \nullval
      \spland \sym{lv} = \sym{level} + 1
      \spland \sym{p00} = \fmtidef{lshift}(\sym{path}, 2)
      \\ \qquad
      \spland \mathsf{hashblocks}(\sym{k}, \sym{level}) = \sym{path}
      \\ \qquad
      \spland \sym{p01} = \sym{p00} + 1
      \spland \sym{p10} = \sym{p00} + 2
      \spland \sym{p11} = \sym{p00} + 3

\end{array}\)
\end{lrbox}

\newsavebox{\SBoxHTEx}
\begin{lrbox}{\SBoxHTEx}\scriptsize\lstnonum
\begin{lstlisting}[language=C,alsolanguage=Spec,alsolanguage=SqueezeOps]
HT *hashtrie_write(HT t, val key, val val)
{
  assume(#$\icall{\vart}{\fmtidef{hashtrie}}{0, 0}$#);
  if (t == NULL) MKLEAF(t, key, val);
  else {
    HT p; val l, hk = hash(k, 0), hkpart = 0;
    for (p = t, l = 0; ; l++, hkpart+=2) {
      if (key == p->k) { UPDATE(p, key, val); break; }
      if (hkpart >= 2*12) hk = hash(k, l), hkpart = 0;
      block = (hk >> hkpart) & 0x3;
      if (block == 0x0) {
        if (p->blk00 == NULL) { MKLEAF(p->blk00, key, val); break; }
        else p = p->blk00;
      }
      else if (block == 0x1) {...}
      else if (block == 0x2) {...}
      else {...}
    }
  }
  assert(#$\icall{\vart}{\fmtidef{hashtrie}}{0, 0}$#); return t;
}
\end{lstlisting}
\end{lrbox}
\setlength{\fboxsep}{\saveboxsep}

\begin{figure}[t]
\subfloat[An inductive definition for hash tries.]{\label{fig:htpred}\ovalbox{\usebox{\SBoxSepHTCheckerExample}}}
\hfill
\subfloat[Corresponding C-style invariant checker.]{\label{fig:htch}\ovalbox{\usebox{\SBoxHTCh}}}
\hfill
\subfloat[The hash trie write.]{\label{fig:htex}\ovalbox{\usebox{\SBoxHTEx}}}
\caption{The write operation for hash tries.}
\label{fig:htexample}
\end{figure}

The inductive predicate $\fmtidef{hashtrie}$ shown in
\figref{htpred} constrains \sym{path} with an uninterpreted function
$\mathsf{hashblocks}$ that corresponds to a complex pure constraint. These uninterpreted functions enables the elimination of some
infeasible states in the analysis, but constraints involving them are left to
dynamic validation. In \figref{htch}, we give the corresponding C-style invariant checker that validates the \code{path} from the root to every
stored key equivalent to the bit blocks of the hashed key. The blocks are
computed by \code{hashblocks} which potentially rehashes the key up to the
current \code{level}. This function includes bit operations, a loop, and
invocations to a hash function.
We fixed the example to blocks of two bits and rehashing at every 12 levels.


The write operation of the hash trie is shown in \figref{htex}. The main loop
traverses down the trie to locate the storing node for the \code{key}. In an
iteration, bit operations extract two bits to determine the direction, and
rehashing is invoked at every 12 levels. Note that the rehashing function
\code{hash} is different from \code{hashblocks} in invariant validation.
For a static verifier to prove the assertion at the end of \code{hashtrie_write}, it must be able to reason about how the code in this function maintains the \code{hashblocks} property.
For \Tool, it does not need to be able to relate \code{hash} and \code{hashblocks} to synthesize short-circuiting dynamic validation. 
The only static reasoning that it needs to show that 
the trie invariant is maintained
across traversal iterations, which is done so by the validated view abstract domain.

After the traversal locates a node to update or reaches an empty sub-trie, the
write operation modifies the trie and asserts the invariant before returning.
This modification essentially stores the \code{key} in the trie. While
\Tool~cannot verify the constraints over it, the synthesized incremental validation code simply invokes \code{hashblocks} on the modified node while short-circuiting the
validation of the rest of the hash trie.

\fi 


\end{document}